\documentclass[12pt]{article}
\usepackage{amsfonts}

\usepackage{graphicx}
\usepackage{amsmath}
\usepackage{hyperref}
\usepackage{color}
\usepackage[mathscr]{eucal}

\newtheorem{theorem}{Theorem}[section]

\newtheorem{corollary}[theorem]{Corollary}

\newtheorem{example}[theorem]{Example}

\newtheorem{lemma}[theorem]{Lemma}
\newtheorem{sublemma}[theorem]{Sublemma}
\newtheorem{notation}[theorem]{Notation}

\newtheorem{proposition}[theorem]{Proposition}
\newtheorem{remark}[theorem]{Remark}

\newenvironment{proof}[1][Proof]{\textbf{#1.} }{\ \rule{0.5em}{0.5em}}

\def \<{\langle}
\def \>{\rangle}

\def \D{{\mathcal D}}

\def \ep{\epsilon}

\def \R{\mathbb R}

\def \ka{\kappa}

\def \l{\lambda}

\def \w{\omega}

\def \beq{\begin{equation}}
\def \eeq{\end{equation}}
\def \ben{\begin{align}}
\def \een{\end{align}}
\def \enda{\end{align}}

\def \n{\nabla}

\def \eref{\eqref}

\def \({\Big(}
\def \){\Big)}
\numberwithin{equation}{section}

\begin{document}


\title{Invariance of intrinsic hypercontractivity under perturbation of  Schr\"odinger operators
\author{Leonard Gross \\
Department of Mathematics\\
Cornell University\\
Ithaca, NY 14853-4201\\
{\tt gross@math.cornell.edu}}
}

\maketitle

\begin{abstract}

      A Schr\"odinger operator that is bounded below and has a unique
positive ground state can be transformed into a Dirichlet form operator
by the ground state transformation. If the resulting Dirichlet form
operator is hypercontractive, Davies and Simon call the Schr\"odinger
operator ``intrinsically hypercontractive". I will show that if one adds a
suitable potential onto an intrinsically hypercontractive Schr\"odinger
operator it remains intrinsically hypercontractive. 
The proof uses a fortuitous relation between the WKB equation and 
logarithmic Sobolev inequalities.
     All bounds are dimension independent.  
    The main theorem will be applied to several examples.

\end{abstract}

\tableofcontents

\newpage
\section{Introduction}

An operator whose quadratic form is a Dirichlet form has some particularly nice properties.
Suppose that $m$ is a measure on a Riemannian manifold $X$ and  $A$ is a self-adjoint operator, 
densely defined in $L^2(X,m)$, such that 
\begin{align}
(Af, g)_{L^2(X,m)} = \int_X \<\n f, \n g\> dm     \label{I1}
\end{align}
for an appropriate set of functions $f$ and $g$. Here $\<\cdot, \cdot\>$ denotes the Riemannian metric.
Such operators have been studied systematically for many years. 
\cite{BD58,BD59,Sil74,Fu80,BH91,MR92,FOT94,ChF2012}.
Some divergence form operators are included  in this class,  \cite{GT}. 
 The semigroup $e^{-tA}$  associated to such an operator is positivity preserving, 
 is the generator of a Markov process, is a contraction in all $L^p$ spaces, 
 and frequently has  useful smoothing properties. 
There is an equivalence between hypercontractivity properties  of the semigroup $e^{-tA}$ and 
coercivity properties of its Dirichlet form generator $A$, \cite{G1}. 
The latter  take the form of logarithmic Sobolev inequalities.

 A Schr\"odinger operator  with a non-zero potential is  not a Dirichlet form operator, but can often
 be unitarily transformed into one: Suppose that the
 operator $H:= -\Delta +V$
 acts in  $L^2(\R^n, dx)$ and has an eigenvalue $\l_0$ at the bottom of its spectrum with 
 multiplicity one.
 The corresponding 
 normalized eigenfunction $\psi$ may  typically  
 be chosen to be strictly positive almost everywhere. 
 The measure $dm_\psi := \psi^2 dx$
 is then a probability measure on $\R^n$ and the map $U:f\to f\psi$ is 
 a   unitary operator from $L^2(m_\psi)$
 onto $L^2(m)$, as is easily verified. A simple computation shows that the operator 
 $\hat H :=U^{-1}(H - \l_0)U$, which acts in $L^2(\R^n, m_\psi)$, 
will then be  the  Dirichlet form operator for $m_\psi$. That is, 
 $(\hat H f,g)_{L^2(m_\psi)} = \int_{\R^n} \<\n f(x), \n g(x)\> dm_\psi(x)$,
where $\<\cdot,\cdot\>$  is the inner product on $\R^n$. The semigroups $e^{-t(H-\l_0)}$  
 and $e^{-t\hat H}$ are 
 unitarily equivalent via $U$, but differ in very important respects.  
 The transformation of the Schr\"odinger operator $H$ into the Dirichlet form operator $\hat H$ 
 is nowadays called the ground state transformation. An early incarnation of this transformation
 goes back to an 1837 paper of Jacobi, \cite{J1837}, whose interest was to remove the zeroth order term
 from an ordinary differential operator. Indeed $\hat H$ has no zeroth order term. The potential $V$
 is now encoded in the measure $m_\psi$. The ground state transformation was used in \cite{G1}
 to produce Dirichlet form operators from Schr\"odinger operators by this method.

 The notions of  {\it intrinsic hypercontractivity} and  {\it intrinsic ultracontractivity}  were introduced by
 Davies and Simon in their paper \cite{DS84}: Suppose that the operator
  $\hat H$ above is hypercontractive  or ultracontractive
 in the sense that the semigroup $e^{-t\hat H}$ is hypercontractive, resp. ultracontractive 
 in $L^2(\R^n, m_\psi)$. 
 They then call  $H$ itself 
 {\it intrinsically hypercontractive} (resp. {\it intrinsically ultracontractive$)$}. 
 They showed that
 intrinsic ultracontractivity is invariant under perturbation of the potential $V$ by a bounded potential 
but left open the question as to whether intrinsic hypercontractivity is also invariant 
under perturbation by bounded potentials.
The goal of this paper is to show that intrinsic hypercontractivity for semigroups generated 
by Schr\"odinger operators  is invariant under 
perturbation of the potential by a class of unbounded potentials, including all bounded potentials in particular. 
We will do this in a dimension independent way over arbitrary Riemannian manifolds. 
We  also show, by examples, how to combine this perturbation
theorem with the convexity techniques of the Bakry-Emery method to produce a large class of Dirichlet forms
satisfying a logarithmic Sobolev inequality.

A proof of invariance of intrinsic hypercontractivity  requires
 showing that if $-\Delta + V_1$ is intrinsically
hypercontractive then   $-\Delta + V_1 + V$ is also intrinsically hypercontractive  
under  suitable conditions on $V$. 
The ground state transformation  for $-\Delta + V_1 + V$ can be realized as the 
 composition of two successive ground state  transformations, 
 one for $-\Delta + V_1$, giving a Dirichlet form operator $\hat H_1$, and a second 
 one for the Schr\"odinger operator $\hat H_1 + V$.  We will elaborate on this 
 composition property  of the ground state transformation in Section \ref{secgst}.
 By hypothesis, the Dirichlet form operator $\hat H_1$ is hypercontractive. 
    Using this and the known equivalence of hypercontractivity to logarithmic Sobolev inequalities,
     the invariance of intrinsic hypercontractivity  can be phrased directly  in terms of 
     the perturbation of a Dirichlet form as follows:

       Suppose that $m$ is a probability measure  on a Riemannian manifold $X$ and  
 that the logarithmic Sobolev inequality  
\begin{align}
Ent_m(u^2) \le 2c \int_X |\n u|^2 dm          \label{I4}
\end{align}
holds for some constant $c$. 
(Here $m$ plays the role of the ground state measure for $-\Delta + V_1$ in the example of the 
preceding paragraph.)
Denote by $\n^*\n$ the  Dirichlet form operator associated to $m$. 
It is defined, as in \eref{I1}, by $(\n^*\n u, v)_{L^2(m)} = \int_X \<\n u, \n v\> dm$. 
(Thus $\n^*\n = \hat H_1$ in the example.)
 Let $V$ be a potential on $X$.  
 If the Schr\"odinger operator $\n^*\n + V$ has an eigenvalue $\l_0$ of multiplicity one 
at the bottom of its spectrum  
with a normalized, a.e. strictly positive eigenfunction $\psi$ then 
 the ground state transformation for  $\n^*\n + V$ associates to $\psi$ the new ground 
 state measure $m_\psi = \psi^2 m$
and its corresponding Dirichlet form $\int_X |\n f|^2 dm_\psi$. 
The  problem of invariance of  intrinsic hypercontractivity  asks  
  for  conditions on $V$  that will ensure that the new
Dirichlet form also satisfies a logarithmic Sobolev inequality.

We will prove that   if  \eref{I4} holds and if there are constants  $\ka >0$ and $\nu > 2c$ 
such that 
\beq
M:= \|e^V\|_{L^\ka(m)}\| e^{-V}\|_{L^\nu(m)} < \infty \label{I5}
\eeq 
then the operator $\n^*\n + V$ is bounded below,
the bottom
of its  spectrum is an eigenvalue of multiplicity one,  there is a normalized ground state $\psi >0$
a.e.  and  there  is a constant $c_1$ such that
\begin{align}
Ent_{m_\psi}(f^2) \le 2c_1 \int_X |\n f|^2 dm_\psi.                 \label{I6}
\end{align}
Moreover there are constants $a$ and $b$ depending
only on $c,\ka$ and $\nu$ such that $c_1 \le a M^b$.  In particular, the Schr\"odinger operator 
$\n^*\n +V$ has a gap at the bottom of its spectrum of at least  $2M^{-b}/a$.   
All bounds are dimension independent.
This is the main theorem of the present paper.

\bigskip
  
There is a large literature on a related problem: Suppose that
$F:X \to \R$ is measurable and $\int_Xe^{-2F} dm =1$ for some probability measure $m$. 
Then $m^F := e^{-2F} dm$ is another 
probability measure and one can
 ask for conditions on $F$ which ensure 
that the Dirichlet form for $m^F$ 
 satisfies a logarithmic Sobolev 
inequality when 
$m$ does.
If, given a potential $V$ with its ground state $\psi$, one puts $F = -\log \psi$  
 then $m_\psi = m^F$ and the desired conclusion is the same for the two perturbation problems.
  But the hypotheses
are very different. For us it is essential to impose conditions only on the potential $V$ and deduce from them
any properties of $\psi$ 
 that may be needed for proving  \eref{I6}.
If, on the other hand, 
 one takes $F$ as the primary data rather than $V$, then it is natural to impose conditions directly
on $F$. This is the case for the application of logarithmic Sobolev inequalities to classical statistical mechanics,
and is frequently used in the application of logarithmic Sobolev inequalities to large deviations,
concentration of measure and optimal transport.
An early perturbation theorem taking $F$ as the given data 
is the Deuschel-Holley-Stroock (DHS) theorem, \cite{HS1987,DS1990}, 
 which asserts that boundedness of $F$ is a sufficient condition.
 One may  take $c_1 = c\exp(osc 2F)$ in \eref{I6}. (cf. also \cite[Proposition 3.1.18]{Roy07} 
 or \cite[Proposition 5.1.6]{BGL} for a proof of this.)
 The two papers  \cite{HS1987,DS1990}
 link logarithmic Sobolev inequalities with classical statistical mechanics.
 See also e.g. Royer, \cite{Roy99,Roy07},
  Guionnet and Zegarlinski, \cite{GZ03}, Helffer  \cite{Helffer2002}, 
 and Ledoux, \cite{Led2001u},   
  for further 
  early 
 expositions of this connection with classical statistical mechanics.
 See Ledoux,  \cite{Led1999, Led2000, Led2001} for expositions of the connection with concentration of measure, and see Villani, \cite{Vil2003,Vil2009} and Gigli-Ledoux \cite{LedG2013}  
 for expositions of the connection with optimal transport.

 Whether one perturbs the measure  $m$ directly, by 
inserting a density $e^{-2F}$,  or  perturbs $m$  indirectly,
 via the Schr\"odinger equation, the identities that accompany  the ground state transformation
 play a central role, as will be explained in Section \ref{secgs}.  
 Even if $F$ is the primary object, these identities
 suggest the use of hypotheses on $F$ that include its relation to 
an artificial  potential $V_F$,  constructed  from $F$, for which the ground state of $\n^*\n + V_F$
is exactly $e^{-F}$.
Many works hypothesize conditions on $F$, which are in fact conditions on a combination of $F$ and $V_F$.
 Further historical discussion of this will be given in Section \ref{secgs}
 after more details of the ground sate transformation are described and also in Section \ref{secwp},
 which contains some comparisons of results.

\bigskip

   Several papers  
aimed at developing techniques for proving spectral gaps and 
logarithmic Sobolev inequalities  directly over infinite dimensional spaces 
led to some  of the methods that we will be building on. 
S. Kusuoka, \cite{Kus1991}, \cite{Kus1992}, seeking an infinite dimensional analog of the Hodge-deRham theorem  for an open subset of an abstract Wiener space, developed a method for proving a weak kind of 
spectral gap for a Dirichlert form 
over an infinite dimensional manifold.  Aida, Masuda and Shigekawa \cite{AMS94}, \cite{AS94}  proved   a perturbation theorem for Gaussian measure 
on an abstract Wiener space that imposed hypotheses on the perturbing density $e^{-2F} $.  
They replaced the hypothesis that $F$ be bounded,   required in the DHS theorem, by a 
size condition on the gradient of $F$.
The notions of spectral gap and positivity improving were themselves 
better understood through various kinds of weaker or stronger versions developed further by
M. Hino, \cite{Hino1997}, \cite{Hino2000},  
S. Aida, \cite{Aida1998},  \cite{Aida2001},
Gong and Ma, \cite{GongMa1998a}, 
Liming Wu, \cite{Wu2000}, 
P. Mathieu, \cite{Mat98}, 
M. Rockner and F-Y Wang, \cite{RW2001},
and culminating in the resolution, by  Gong and Wu \cite{GongWu2000} 
 and  F. Gong, M. R\"ockner and L. Wu in \cite{GRW2001}, 
of a spectral gap  conjecture for loop groups  made in \cite{G93},  
which was itself aimed at proving a Hodge-deRham type theorem over loop groups.
See the introductions to \cite{GRW2001} and \cite{Aida2001} for histories of these 
techniques up to that time and  in particular see  Remark 4.13 in \cite{Aida2001} for an 
illuminating comparison  of some of the historical conditions on the log density $F$. 
See   \cite{CLW2011} for later historical perspective  and development of
more quantitative bounds on the rate function for the weak Poincar\'e inequality over loop spaces.

\bigskip

This paper depends heavily on techniques developed by Aida, \cite{Aida2001}. 
Aida derived  a lower bound on the spectral gap of the perturbed operator largely  in terms of 
 information about the distribution of the ground state wave function $\psi$.  
 We will build on his techniques.  We will first derive bounds on $\|\psi^{-1}\|_{L^s(m)}$, 
 for some $s >0$, that depend only on $c,\kappa, \nu$ and $M$. 
We will use these bounds to derive a defective logarithmic Sobolev inequality  
 and then use them again to derive information needed about the distribution of $\psi$
 for producing a spectral gap via Aida's method.  Rothaus' theorem \cite{Rot5}  in the 
 form of \cite[Proposition 5.1.3]{BGL}, then yields \eref{I6}.
 All bounds are quantitatively dependent on the input data $c,\ka, \nu$ and $ M$.

\bigskip

\section{Statements}

\subsection{The main theorem}    \label{secmainthm}

\begin{notation} \label{notso} {\rm   (Schr\"odinger operator in its ground state representation). 
     Denote by $X$ a Riemannian manifold, by $dx$ its Riemann-Lebesgue measure and
by $\n$ the gradient operator.   $m$ will denote a  measure on $X$ with a 
density:  $dm = \rho^2 dx$ with $\rho >0$ and $\n \rho \in L^2_{loc}(dx)$.
The  adjoint of the gradient operator with respect to $m$ is defined on smooth vector fields over $X$ by 
\begin{align}
\int_X (\n^* v)h\, dm = \int_X v\cdot (\n h) dm  \ \ \text{for all}\ \ h \in C_c^\infty(X).   \label{div1}
\end{align}
Here we have written $v\cdot u = g(v, u)$, where $g$ is the Riemannian metric and 
$v$ and $u$ are vector fields.
The technical condition on $\rho$ ensures that $\n^* v \in L^2_{loc}(m)$ for every smooth 
vector field $v$ on $X$,  (cf. \cite[Theorem 3.1.3]{FOT94}).
             Then
 \begin{align}
(\n^* \n f, g)_{L^2(m)} = \int_X \n f \cdot \n g\, dm\ \ \text{for all}\ f,g \in C_c^\infty(X).            \label{div2}
\end{align}     
 The Dirichlet form on the right is closable in $L^2(m)$, (cf.  \cite[Theorem 3.1.3]{FOT94}). 
  Its closure is associated to a non-negative self-adjoint operator, which we refer to as 
  the Dirichlet form operator  associated to $m$ and $g$. We denote it by $\n^*\n$.
For example if $m$ is Lebesgue measure  on $\R^n$  and $g$ is the Euclidean metric then
$\n^*\n = - \Delta$ with its usual self-adjoint domain in $L^2(\R^n, dx)$.

Let $V$ be a real valued function on $X$. The 
 Schr\"odinger operator we are interested in is given informally by
\beq
H = \n^*\n + V.
\eeq
 We will impose conditions on $V$ which ensure that this expression is essentially self-adjoint, 
 that $\l \equiv \inf(\text{spectrum}\, H)$ is an eigenvalue
with multiplicity one and that $H$ has a corresponding normalized eigenfunction $\psi$ which is  strictly
positive a.e. on $X$.

The corresponding ground state measure $m_\psi$ is given by
\beq
d m_\psi = \psi^2 dm.
\eeq
$m_\psi$ is a 
 probability measure on $X$ and has its own Dirichlet form 
operator $\n_\psi^* \n$ acting in $L^2(m_\psi)$ and given by 
\begin{align}
(\n_\psi^*\n f, g)_{L^2(m_\psi)} = \int_X \n f \cdot \n g\, dm_\psi.
\end{align}
The map $U: L^2(m_\psi) \to L^2(m)$ defined by
\beq
U f = f \psi
\eeq
is clearly unitary. It is a standard computation, which we will repeat in Section \ref{secgs}, to show that
$U$ intertwines $H -\l$ with $\n_\psi^*\n$:
\beq
U^{-1}(H - \l) U = \n_\psi^*\n.
\eeq
Thereby the ground state transformation $U$ converts the  Schr\"odinger operator $H-\l$ 
to another Dirichlet form operator.

 In case $m$ is a probability measure we define the $m$ entropy of a non-negative integrable function $f$ by
 \begin{align}
 Ent_m(f) = \int_X  f \log f dm - \(\int_X f dm\)\(\log \int_X f dm\).        \label{E1}
 \end{align} 
}
\end{notation}

\bigskip
\noindent
\begin{theorem} \label{thmM} $($Main theorem$)$.
Assume that  $m(X) =1$ and that 
\begin{align}
&1. \ \  Ent_m (u^2) \le 2c \int_X |\n u|^2 dm.          \label{mt1}\\ 
&2. \ \  \|e^V\|_\kappa < \infty \ \text{and}\ \|e^{-V}\|_\nu < \infty\ \text{for some}\ \ \ka >0\  
                                      \text{and}\ \  \nu > 2c.        \label{mt2} 
\end{align}
 Then
 
 a. $ \n^*\n + V$ is essentially self-adjoint  on $\D(\n^*\n) \cap L^\infty$. Let $H =$ closure of  $\n^*\n + V$.

b. $\l_0 \equiv \inf$ spectrum $H$ is an isolated eigenvalue of multiplicity one.  It has an eigenfunction
$\psi > 0\ a.e.$ with $ \int_X \psi^2 dm =1$. 

c. Let
\begin{align}
 M =\|e^V\|_\kappa \|e^{-V}\|_\nu.       \label{mt3}
 \end{align}
 There is a constant $c_1$ depending only on $c, \kappa, \nu$ and $M$,  
 such that
\beq
Ent_{m_\psi}(f^2) \le 2c_1 \int_X |\n f|^2 dm_\psi.   \label{mt5}
\eeq

d. In particular $H$ has a spectral gap of at least $1/c_1$  above the eigenvalue $\l_0$.
 
e.  There are constants $\alpha$ and $\beta$, depending only on $c, \ka, \nu$, such 
that  $c_1 \le \alpha M^\beta$ and therefore
$H$ has a spectral gap above $\l_0$ of at least $\alpha^{-1}M^{-\beta}$.
\end{theorem}

\begin{remark}\label{remsg1} {\rm (Spectral gap). Our procedure for proving \eref{mt5} requires
proving both a Poincar\'e inequality for $m_\psi$ and a defective logarithmic Sobolev inequality.
The spectral gap associated to this Poincar\'e inequality is typically larger than the one 
listed in item d.  
 See Remark \ref{remsg2} for more details.

}
\end{remark}

      \begin{remark} \label{remsum} {\rm (Overview).  The main ingredient in the proof of 
      Theorem \ref{thmM}
is the derivation of $L^p(m)$ bounds for  
the ground state $\psi$ and for  its inverse $1/\psi$.
Bounds on $\|\psi\|_{L^p(m)}$ can be derived from hyperboundedness estimates for the 
 Schr\"odinger operator $\n^*\n +V$ 
      by techniques that were initially  developed in the early 1970's for the purposes of 
      constructive quantum field theory. 
In addition to the logarithmic Sobolev inequality \eref{mt1} the key hypothesis needed for this step
      is the assumption that $\| e^{-V}\|_\nu < \infty$ for some $ \nu >2c$,  but not the drastic 
condition  $\|e^{V}\|_\ka < \infty$.
    The proofs  of essential self-adjointness of $\n^*\n +V$ and the  existence and uniqueness 
 of its  ground state also depend only on these two hypotheses and not on the condition 
 $\| e^V\|_\ka < \infty$.  The proofs and   
 bounds on $\| \psi\|_{L^p(m)}$ are given in Section \ref{sechyperp}.

 The techniques needed to establish 
       bounds on $\|(1/\psi)\|_p$ are very different.
       They have their origin partly in the work of Aida, \cite{Aida2001}, which was itself motivated 
       by attempts to prove a Hodge-deRham theorem over certain infinite dimensional loop spaces.       
Aida derived information about the distribution of $\log \psi$, which he needed to prove a spectral gap,
from an identity related to the WKB equation.  
 We will see that Aida's identity  also bears a fortuitous relation to logarithmic Sobolev 
  inequalities. We will use this relation to derive  bounds on the entropy of $\psi^{-s}$ for small positive $s$.
 From this, using Herbst's method, we will 
  derive bounds on $\|\psi^{-s}\|_{L^1(m)}$ for such $s$. 
   These bounds make use of the strong condition $\|e^{V}\|_\ka < \infty$ assumed in \eref{mt2}.
 These steps are carried out in Sections \ref{secpm} and \ref{seclp}. 
    
  Our bounds on $\|\psi\|_p$ and $\|\psi^{-1}\|_s$ allow us to derive  a defective logarithmic Sobolev
  inequality for the ground state measure $m_\psi$.     The final step in proving \eref{mt5} consists
  in removing the defect by  proving a spectral gap for $\n^*\n + V$ (or equivalently,  for the 
  Dirichlet form operator   for $m_\psi$) and then applying Rothaus' theorem.
  Our  technique for proving a spectral gap is  largely  due to Aida, \cite{Aida2001}. We are able to make some
  simplifications of his method by using our $L^p(m)$ bounds for  $\psi^{\pm 1}$.
          These  bounds will allow us to  derive  the quantitative bounds on $c_1$ given in item e. of 
 Theorem  \ref{thmM}. 
}
\end{remark}

\subsection{Non-standard hyperboundedness in $L^p(m)$} \label{secns}

   We will 
   establish  logarithmic Sobolev inequalities  for the operator $\n^*\n +V$  in the 
   spaces $L^p(m)$ and  derive  corresponding hyperboundedness  in these spaces  
  in order to prove existence, uniqueness and properties of its 
  ground state.   This must be done before transforming to the ground state representation.
  
  Since $\n^*\n + V$ is not a Dirichlet form operator the minimum time to boundedness
  from $L^q(m)$ to $L^p(m)$ of $e^{-t(\n^*\n +V)}$ takes a different form from Nelson's classical time.
  Moreover $q$ and $p$ must be restricted to a small neighborhood of 2 for any such boundedness
  to hold. We will see by example in Section \ref{secgp} 
  that the peculiar restrictions on $q$ and $p$ in Corollary \ref{corhb2} are not artifacts of the proof.

 \begin{notation} {\rm The quadratic equation 
 \beq
  p\frac{p}{p-1} = 2\nu/c     \label{L313m}
 \eeq 
 is self-conjugate in the sense that it is invariant under the map $ p \to p/(p-1)$. 
 If $\nu > 2c$ then it has two  solutions,
 which are conjugate exponents as we will see in Section \ref{sechyperp}. 
 Denote them by $q_0, p_0$ with $1 < q_0 <2 < p_0 < \infty$.
 }
 \end{notation}
 
\begin{theorem} \label{thmns2} Assume that \eref{mt1} holds. Suppose that $\nu > 2c$ and that
$\|e^{-V}\|_{L^\nu(m)} < \infty$. Suppose also that $V \in L^{p_1}(m)$ for some $p_1 \ge 2p_0/(p_0-2)$. 
Then $\n^*\n +V$
is essentially self-adjoint. Its closure $H$ is bounded below.  The semigroup $e^{-tH}$ that it generates
extends uniquely to a strongly continuous semigroup of bounded operators on $L^q$ for $q \in [q_0,2]$ and
restricts to a strongly continuous semigroup of bounded operators on $L^q$ for $q \in [2, p_0]$. 
If $q_0 < p <p_0$ then 
\begin{align}
 &Ent_m(|u|^p)  
 \le pc_\nu(p) \<(H+ \log \|e^{-V}\|_\nu) u, u_p\>_{L^2(m)}\ \    
 \label{L325}
 \end{align}  
 for $u$ in the $L^p$ domain of $H$, where $u_p = (\text{sgn}\, u) |u|^{p-1}$ and 
 \begin{align}
c_\nu (p)  = \frac{\nu p}{(p_0-p)(p-q_0)}\ \ \ \text{for}\ \ \ p \in (q_0, p_0).     \label{L329}
\end{align} 
 In particular, at $p=2$  the defective logarithmic Sobolev inequality
 \begin{align}
 &Ent_m(u^2)  
 \le 2c_\nu  \<(H+ \log \|e^{-V}\|_\nu) u, u\>_{L^2(m)}\ \      \label{L325h}
 \end{align} 
 holds with 
 \begin{align}
 c_\nu =  \frac{c}{1- (2c/\nu)}.       \label{L329a}
 \end{align}
\end{theorem}

  \begin{corollary} \label{corhb2}  
  $($Non-standard hyperboundedness$)$.
Continuing the notation and assumptions of Theorem \ref{thmns2}, let  
\begin{align}
a_\nu &= \sqrt{1 - (2c/\nu)}   \ \ \ \ \ \text{and}       \label{L341d}\\
\tau(p) &= \frac{c}{2a_\nu} \log\frac{q_0^{-1} - p^{-1}}{p^{-1} - p_0^{-1}},\ \ \ q_0 < p < p_0.   \label{L505}
\end{align}  
Then
 \begin{align}
 \|e^{-tH} \|_{q \to p} \le  \| e^{-V}\|_\nu^t\ \    
                 \ \ \text{for}\ \ t \ge \tau(p) - \tau(q) \ \ \text{ if}\ \ \  q_0 < q \le p <p_0.     \label{L289}
 \end{align} 
 Moreover, if $ q \in [q_0, p_0]$  then   
 \beq
 \|e^{-tH} \|_{q\to q} \le \| e^{-V}\|_\nu^t\ \ \ \ 
                          \ \   \text{for all}\ \ t \ge 0. \label{L291}
 \eeq
 For fixed $q$ and $p$  in $(q_0, p_0)$ with $q\le p$ the function $t_{q,p}\equiv \tau(p) - \tau(q)$
 decreases as $\nu$ increases. 
 \end{corollary}

\begin{remark} {\rm   The function $\tau(p)$ does not give  
 the standard Nelson time to contraction in \eref{L289}. The Nelson time is 
 determined by   $\tau_0(p) = (c/2) \log(p-1)$.
(See e.g. \cite{G1}.) But if $V$ is
 bounded below, then we may let $\nu \uparrow \infty$ and, as  we will see in Section \ref{secbelow},
    $\tau(p)-\tau(q) \downarrow \tau_0(p) - \tau_0(q)$.   
}
\end{remark}

 \begin{corollary} \label{corEU} Under the assumptions of Theorem \ref{thmns2}, $\l_0 \equiv$ inf spectrum $H$ is an 
 eigenvalue of multiplicity one. It has an eigenvector $\psi$ which is strictly positive a.e..
 \end{corollary}
 
  The proofs will be given in Section \ref{sechyperp}.  
  We will also establish upper bounds on $\|\psi\|_p$ for $2 < p <p_0 $ and lower bounds on $\|\psi\|_r$ for
  $0 < r < 2$.

\subsection{A product of moments} \label{secpms}
 
 The Schr\"odinger equation for the ground state $\psi$ can be written in WKB form simply.
Let $F = - \log \psi$. Since $\psi$ is strictly positive almost everywhere, $F$  
is real valued almost everywhere. A computation, which will be sketched in Remark \ref{remWKB},
 yields 
 \begin{align}
    \n^*\n F + |\n F|^2 =  V - \l_0.\ \ \ \text{WKB}         \label{wkb}
    \end{align} 
    Suppose that $v$ is a real valued function on $\R$. Multiply \eref{wkb} by the composed function
    $v\circ F$ and, using $\n (v\circ F) = v'(F) \n F$, integrate over $X$ to find informally, 
    after an integration by parts      
    \begin{align}
\int_X (v'(F) + v(F)) |\n F|^2 dm = \int_X v(F) (V -\l_0) dm \ \ \ \ \text{Aida's identity}   \label{W30} 
\end{align}  
A more precise derivation will be given in Theorem \ref{thmA3}.
    Aida used this identity cf. \cite[Equ. (3.26) in Lemma 3.3]{Aida2001} 
    to derive information about the distributions 
    of $|\n F|$, $F$ and  $\psi$, which was crucial for 
    his proof of a spectral gap.

  We will exploit Aida's identity in a different way. Suppose that $\phi$ is a real valued function
  on $\R$. 
 We may apply the  logarithmic Sobolev inequality \eref{mt1} to the composed function $\phi \circ F$
to find
  \begin{align}
  Ent_m((\phi\circ F)^2) \le 2c \int_X (\phi'\circ F)^2 |\n F|^2 dm,       \label{pm5}
 \end{align}   
 wherein we have used $\n (\phi\circ F) = (\phi'\circ F) \n F$.
    If $\phi$ and $v$ are  chosen in \eref{pm5} and \eref{W30} so that the two integrands involving $|\n F|^2$
    are equal then Aida's identity, together with \eref{pm5} give a bound on $Ent_m((\phi\circ F)^2)$ 
in terms of the potential. In this way the quadratic nonlinearity in the WKB equation meshes well with
the use of logarithmic Sobolev inequalities. We will show that this procedure can be carried 
out for several different kinds
of functions $\phi$. In particular, taking $\phi(s) = e^{ts}$ (giving $\phi\circ F = \psi^{-t}$), we will derive
entropy bounds on $\psi^{-t}$ for $t $ in an 
 open interval containing zero. 
We will then derive moment bounds from these entropy bounds using Herbst's method. 
The interval  of $t$ for which this procedure works depends on $\ka$ in the 
condition \eref{mt2}, and on the solutions to the quadratic equation \eref{W752}. 
 We will show by example in Section \ref{secpp} that the peculiar interval of $t$ for which this procedure
works is not an artifact of the proof. The moment bounds that we arrive at take the form of
a bound on a product of moments, as in the following simplified theorem. 

\begin{theorem} \label{thmmp0} 
Suppose that the hypotheses of Theorem \ref{thmns2} hold. 
Let $\ka >0$. Assume that
\begin{align}
\|e^V\|_\ka < \infty.                      \label{W751}
\end{align}
 Let $s_0$ and $-r_0$  be the positive  and negative roots of the quadratic equation
 \begin{align}
 t^2 - (2\ka/c)(t+1) = 0.                                                                                                \label{W752}
 \end{align}
 Then there is a function $f:(0,r_0)\times (0, s_0) \to [0, \infty)$ such that
\begin{align}
\|\psi\|_r \| \psi^{-1}\|_s \le   \| e^{V-\l_0}\|_\kappa^{f(r,s)}, \ \ \ 0 < r < r_0,\ \ 0 < s <s_0.     \label{W753h}
\end{align}
The function $f(r,s)$ will be given explicitly in Theorem \ref{thmmp1}.
\end{theorem}

\begin{remark}   {\rm(Upper bound on $\|\psi^{-1}\|_s$). 
 Typical perturbation proofs of a defective LSI for the ground state measure $m_\psi$ rely on 
some information about
the behavior of $\psi$ in the regions where $\psi$ is large or 
 where $\psi$ is close to zero. For example
the classical condition of Deuschel-Holley-Stroock \cite{HS1987,DS1990} 
requires that $F \equiv - \log \psi$ be bounded both above and below;
equivalently,
 $0< \ep \le \psi \le K < \infty$ on all of $X$ for some $\ep, K$.  Aida     
 relaxed the condition that $\psi$ be 
bounded away from zero by  assuming instead that
$\psi^{-1} \in L^p(m)$ for some $p >0$,  along with hyperboundeness 
assumptions on $e^{-tH}$ in the spaces $L^p(m)$, 
cf.  \cite[Lemma 4.12]{Aida2001}.  
He proved, moreover, that these hypotheses actually hold  for 
finite and infinite dimensional Gauss measure if
 $E(e^{q V}) < \infty$ for sufficiently large $q$, cf. \cite[Lemma 5.5]{Aida2001}.

We will derive an upper bound on $\| \psi^{-1}\|_s$,  depending only on  
$c,\ka, \nu$ and $M$,   by combining \eref{W753h} 
 with the  lower bound on  $\|\psi\|_r$ derived  in Section \ref{seculb}.  The upper bound on $\|\psi^{-1}\|_s$
 is the key input to the derivation of a DLSI.
 }
 \end{remark}

\subsection{A defective LSI for $m_\psi$.}

\begin{theorem}    Assume that \eref{mt1} and \eref{mt2} hold.
 Let $b_\ka = \sqrt{1 +(2c/\ka)}$ and define $c_\nu$ as in  \eref{L329a}.  Let $a > c_\nu b_\ka$.
Then there exists a number $D$, depending only on $c,\nu,\ka, M$ and the choice of $a$,
 such that
\begin{align}
Ent_{m_\psi}(u^2) \le 2a \int_X |\n u|^2 dm_\psi + D \|u\|_{L^2(m_\psi)}^2  \label{D1}
\end{align}
\end{theorem}

A more detailed version of this theorem, showing the dependence of $D$ on the various 
parameters, and in particular
its dependence on our bounds of the norms $\|\psi^{-1}\|_s$,  is given in 
Section \ref{secdlsi}.

\subsection{Spectral gap}

To complete the proof of Theorem \ref{thmM} we will show that the Dirichlet form operator for $m_\psi$
has a spectral gap. A theorem of Rothaus then shows that the defect in \eref{D1} can be removed 
at the cost of increasing the Sobolev coefficient $2a$. 

In case the defect in \eref{D1} is sufficiently small,
a theorem of F-Y. Wang  can be used to show  
 that there is a spectral gap.
In general, a proof that $m_\psi$ has a spectral gap depends on data encoded in $\psi$ and
not just on the size of $D$ and $a$ in \eref{D1}. 
 We adapt a method of Aida, \cite{Aida2001},
which produces  a spectral gap dependent on the distribution of $\psi$ and its gradient. With the help
of quantitative bounds on $D$ and $a$ in  \eref{D1} we then obtain quantitative bounds on 
the spectral gap,  and, by Rothaus'  theorem,  a 
quantitative bound on the Sobolev constant $c_1$ in \eref{mt5}. This will be carried out in Section \ref{secsg}.

\section{Hyperboundedness of $\n^*\n +V$ in $L^p(m)$} \label{sechyperp}

\subsection{Interval of validity } \label{seciv1}

\begin{lemma} \label{lemiv1} 
  $($Interval of validity$)$. 
  Suppose that $1 < \nu/(2c) < \infty$. Define $a_\nu$ by \eref{L341d}.  
Then the quadratic equation 
\beq
p^2 - (2\nu/c) (p-1)  =0         \label{L313} 
\eeq
has two real roots, $q_0 < p_0$, which are given by
\begin{align}
&p_0 = (\nu/c)\( 1+ a_\nu\),\ \ \ \ \ q_0 = (\nu/c)\( 1- a_\nu\).   \label{L313q}
\end{align}
 They satisfy the following identities.
\begin{align}
&(2\nu/c) (p-1)  -p^2 = (p_0 -p)(p - q_0)\ \ \forall\ \ p \in \R. \label{L313a}\\
&p_0^{-1} = (1/2)\(1-a_\nu\), \ \ \ q_0^{-1} =(1/2)\(1+ a_\nu\) \label{L313qi} \\
&(1/p_0) + (1/q_0) = 1.                            \label{L313b}\\
&1 < q_0 < 2 < p_0 < \infty .                    \label{L313d}\\
&(p_0 -2) (2-q_0) = (2\nu/c) a_\nu^2.      \label{L313c}   \\  
& \nu/(p_0 -q_0)    =c/(2a_\nu).               \label{L313j}\\
&(1/2) - (1/p_0) =(a_\nu/2) = (1/q_0) - (1/2) .      \label{L325p}
\end{align}
In particular $q_0 $ and $p_0$ are conjugate indices. Define $\tau(p)$ by  \eref{L505}. 
Then
\begin{align}
&a) \ \tau\ \text{is a strictly increasing function on}\ (q_0, p_0). \label{L325s} \\
&b) \lim_{p\uparrow p_0} \tau(p) = + \infty, \ \ \ \lim_{p\downarrow q_0}\tau(p) = - \infty \label{L325v} \\
&c) \ \tau(2) = 0.          \label{L325t} \\
&d) \ \tau(p') = - \tau(p)\ \ \text{if} \ \ p' =p/(p-1). \label{L325u}
\end{align}
\end{lemma}
\begin{proof} By the quadratic formula the quadratic equation \eref{L313} has two positive real roots 
   given by     $p = (\nu/c)\( 1\pm \sqrt{1 -2c/\nu}\)$. The roots are therefore correctly given by \eref{L313q}, in view of the definition  \eref{L341d}.    The inverse of the roots are are therefore given by
 $ 1/p = (1/2)(1\mp \sqrt{1 -2c/\nu}\)$, from which follows \eref{L313qi}.     
\eref{L313b} and \eref{L313d} follow from \eref{L313qi} 
while \eref{L313a} just restates that $q_0, p_0$ are the roots of \eref{L313}. 
     Insert $p=2$ in \eref{L313a}
to find  $(p_0 -2) (2-q_0) = 2\nu/c - 4=  (2\nu/c) a_\nu^2$, which is  \eref{L313c}. 
\eref{L313q} shows  that $p_0 -q_0 =  (2\nu/c)a_\nu $,  which is  \eref{L313j}. 
\eref{L325p} follows from \eref{L313qi}.

That $q_0$ and $p_0$ are conjugate indices
 follows from \eref{L313b}, but also from writing the equation \eref{L313} in the form \eref{L313m},
 which exhibits the equation as self conjugate.
 
 Concerning the function $\tau$ defined in \eref{L505}, the properties   \eref{L325s} and \eref{L325v} are  
 clear from the definition, \eref{L505}. 
 \eref{L325t} follows from \eref{L325p}. 
 Replacing $p^{-1}$ by $1- p^{-1}$ in the numerator and denominator of \eref{L505} 
 interchanges the numerator and denominator, in view of \eref{L313b}. This proves  \eref{L325u}.
\end{proof}

\subsection{Proof of non-standard  hyperboundedness for \linebreak bounded $V$} \label{secnsh}

 We assume in this subsection that $V$ is bounded. $\n^*\n$ denotes
the self-adjoint Dirichlet form operator for $m$. The Schr\"odinger operator 
$H \equiv \n^*\n + V$ is then self-adjoint on the domain of $\n^*\n$ and there are 
no serious domain issues.  We will prove all of the inequalities of Section \ref{secns}
in this case. In Section \ref{secesa}  we will remove the boundedness assumption for $V$ and 
show that $\n^*\n + V$ is  essentially self-adjoint and that its closure, $H$,  also
satisfies the inequalities of Section \ref{secns}. Section \ref{secesa} has a technical character.

\bigskip
\noindent
    \begin{proof}[Proof of Theorem \ref{thmns2} for bounded $V$]    By \cite[Lemma 6.1]{G1},
   the logarithmic Sobolev inequality \eref{mt1}    implies 
\begin{align}
Ent_m(|u|^p) \le c \frac{p^2}{2(p-1)} \<\n^*\n u, u_p\> , \ 1 < p < \infty ,   \ \ \ \          (LSp)     \label{L320p}
\end{align}
where $u_p = (\text{sgn}\, u) |u|^{p-1}$.
     We will frequently use Young's inequality in the form 
 \begin{align}
 E(gu) \le Ent(g) +\(\log E(e^u)\) E(g),        \label{BG500c}
 \end{align}
where  $g$ and $u$ are real valued measurable functions on some probability 
space, $g \ge 0$ and $E(g) < \infty$.

In particular, if $v \in L^p(m)$ then, choosing $u = -\nu V$ and $g = |v|^p$ in \eref{BG500c}, we find 
\begin{align}
\int_X (-V)|v|^p dm &\le \nu^{-1} \Big\{Ent_m(|v|^p) + \(\log E(e^{-\nu V})\) E(|v|^p) \Big\}    \\ 
 &= \nu^{-1} Ent_m(|v|^p) + \(\log \|e^{-V}\|_\nu\)  E(|v|^p)    \label{L48}
\end{align}

 It follows from \eref{L320p}, \eref{L48}   and from the definition $H = \n^*\n + V$ that  
\begin{align}
- \<Hv, v_p\> &=-\<\n^*\n v, v_p\> + \int (-V) |v|^p dm       \notag\\
&\le - \frac{2(p-1)}{cp^2} Ent_m(|v|^p) + \nu^{-1} Ent_m(|v|^p) + \alpha \int |v|^p dm        \notag\\
&= \(\nu^{-1}  -\frac{2(p-1)}{cp^2} \) Ent_m(|v|^p) + \alpha \int |v|^p dm,  \label{L49a}
\end{align}
where $\alpha =   \log \|e^{-V}\|_\nu$. 
Rearrange to find
\begin{align}
 \( \frac{2(p-1)}{cp^2} - \nu^{-1}\)Ent_m(|v|^p)&\le   \<Hv, v_p\> +  \alpha \int |v|^p dm  \notag\\
 &= \<(H+\alpha)v, v_p\>    .        \label{L327}
 \end{align}
 With the help of \eref{L313a} we find
 \begin{align}
 \frac{2(p-1)}{cp^2} - \nu^{-1} =\frac{(2\nu/c)(p-1) - p^2}{\nu p^2} =\frac{(p_0 -p)(p - q_0)}{\nu p^2}.
  \end{align}
  For $q_0 < p <p_0$ the last  expression is strictly positive. We may therefore divide \eref{L327} by it to find
  \eref{L325}.  
 Put $p=2$ in \eref{L329} and use \eref{L313c} to arrive at \eref{L325h}.
\end{proof}

\bigskip
\noindent
      \begin{proof}[Proof of Corollary \ref{corhb2} for bounded $V$]
    For $q <p$ the time $t_{q,p}$ that  it takes for $e^{-tH}$  to map $L^q(m)$ into $L^p(m)$ is determined by the equation (cf. \cite[Equation (2.4) of Theorem 1]{G1})
\beq
\hat c(p(t)) dp(t)/dt = p(t), \ \ \ \ \ \ p(0,q) = q,  \ \   p(t_{q,p}, q) = p.   \label{L330}
\eeq
That is, $t_{q,p}$ is the first time that the increasing function $p(t)$ reaches $p$ when starting at $q$.
$\hat c(p)$ is determined by the definition $Ent_m(|u|^p) \le p \hat c(p) \< (H + \alpha)u, u_p\>$
and $\alpha$ is the ``local norm" at index $p$,
(cf. \cite[Definition 1]{G1}).
In our case, \eref{L325}, $\alpha = \log \| e^{-V}\|_\nu$ and  
$\hat c(p)= c_\nu(p)$, which is given by \eref{L329}.

Upon separating variables in \eref{L330} the equation becomes
\begin{align}
\nu \frac{dp}{(p_0-p)(p-q_0)} =  dt.   \label{L332a}
\end{align}
Using \eref{L313j} in the second line below, we have  
\begin{align}
\frac{\nu}{ (p_0-p) (p-q_0)} &= \frac{\nu}{(p_0-q_0)} \{(p_0-p)^{-1} + (p- q_0)^{-1}\} \\ 
&= \frac{c}{2 a_\nu}\{(p_0-p)^{-1} + (p- q_0)^{-1}\}.
\end{align}
The solution to \eref{L332a} is therefore given by
\begin{align}
 \frac{c}{2 a_\nu}\int_q^p\{(p_0-r)^{-1} + (r- q_0)^{-1}\} dr = \int_0^{t_{q,p}} dt  =  t_{q,p}.   \label{L333a}
\end{align}
Thus
\begin{align}
t_{q,p} &= \frac{c}{2 a_\nu} \log\frac{r-q_0}{p_0-r}\Big|_q^p  
   = \frac{c}{2 a_\nu} \(\log\frac{(r/q_0) -1}{1- (r/p_0)} +\log(q_0/p_0)\)\Big|_q^p  \notag \\
   &=  \frac{c}{2 a_\nu} \(\log\frac{q_0^{-1} -r^{-1} }{r^{-1} - p_0^{-1}}\)\Big|_q^p = \tau(p) - \tau(q) \ \ \  
\end{align}
This proves that the minimum assured time to boundedness of $e^{-tH}$ from $L^q(m)$ to $L^p(m)$
is correctly given in  \eref{L289}.
From  \cite[Equation (2.5)]{G1}  we find that $\|e^{-t_{q,p}H}\|_{q\to p} \le e^{t_{q,p} \alpha}$, 
where        $\alpha = \log \| e^{-V}\|_\nu$,     
   because
 the integrand in \cite[Equation (2.5)]{G1} is just the constant  $\log \| e^{-V}\|_\nu$ that appears in 
 \eref{L325}.
 This proves \eref{L289} in case $t = t_{q,p}$. 
If $t > t_{q,p}$  then there exists $p_1 \in (p, p_0)$ such that  
$t = t_{q,p_1}$ because $ \tau(p)$  
is a continuous and strictly  increasing function of $p$ by \eref{L325s},   and 
 goes to $\infty$ as $p \uparrow p_0$ by \eref{L325v}.  
 Therefore  $\|e^{-tH}\|_{q\to p} \le \|e^{-tH}\|_{q\to p_1} \le e^{t_{q, p_1} \alpha} = e^{t\alpha}$. 
 This proves \eref{L289}  for all $t \ge t_{q,p}$.

 The representation \eref{L333a} shows that $t_{q,p}$ is decreasing as a function of $\nu$, as asserted 
 in the corollary, because, as $\nu$ increases
$a_\nu$ increases , as we see from  \eref{L341d},  
 while $p_0$ increases, as we see from \eref{L313q}, and consequently $q_0$ decreases, 
 implying that the integrand in  \eref{L333a} decreases. This proves the last line of Corollary \ref{corhb2}.

 For the proof of \eref{L291} set $p = q $  in \eref{L289}.  
 Since $t_{q,q} =0$ it follows that
 \eref{L291}  holds for all $t \ge0$,  provided $q \in (q_0, p_0)$. 
      (A short, abstract, but less illuminating  proof of \eref{L291} for $q \in (q_0, p_0)$ 
       is given in  \cite[Remark 3.5]{G1993}  that just uses the Hille-Yosida theorem.)
 To prove \eref{L291} for $q \in \{ q_0, p_0\}$ 
 choose first $v \in L^{p_0}$. Then $v \in L^q$ for all
 $q \in (q_0, p_0)$  and $\|v\|_q \to \|v\|_{p_0}$ as $q \uparrow p_0$. By \eref{L291} for $q < p_0$ we have 
 \begin{align}  
 \int_X |e^{-tH} v|^q dm \le \|e^{-V}\|_\nu^{tq} \| v\|_q^q .   \label{L335}
 \end{align}
 Choose a sequence $q_n\uparrow p_0$ and apply Fatou's lemma on the left side of \eref{L335} 
 to find \eref{L291} for $q = p_0$.
 To prove \eref{L291} for $q = q_0$ observe that the previous argument shows
 that \eref{L335} holds for $q = q_0$ if  $v$ is bounded. 
  Now the semigroup $e^{-tH}$ is positivity preserving. So it suffices to prove 
 \eref{L335} for $0 \le v\in L^{q_0}$. For such a function $v$, let $v_n = min(v,n)$ for each positive integer $n$.
Then each function $v_n$  is bounded and  \eref{L335} holds for $q = q_0$. 
We can now apply the monotone convergence theorem to find that \eref{L335} holds for $v$. 
\end{proof}

\begin{remark} \label{remtpf} {\rm  
 The Trotter product formula offers a good heuristic for an
understanding of the   inequality \eref{L289}. If, putting $H_0 = \n^*\n$, one writes    
 $e^{-tH} = \lim_{n\to \infty}\(e^{-tV/n} e^{-t H_0/n} \)^n$ and if $f \in L^q$ 
 then $e^{-t H_0/n}f$ will be in $L^{q_1}$ for some $q_1 > q$ by hypercontractivity. Then
 $e^{-tV/n} e^{-t H_0/n} f$ will be in $L^{q_2}$ for some $q_2 < q_1$ by H\"older's inequality. The exponents
 $q_1$ and $q_2$ are explicitly computable. Continuing in this way $n$ times one can bound the product
 and take the limit as $n \to \infty$ to derive a version of \eref{L289}. 
 This procedure was carried out by I.E. Segal in \cite[Lemma 2.1]{Seg70}, though
 not with Nelson's shortest  time to contraction, which was not known at that time.  
  Segal's method for showing boundedness of $e^{-tH}: L^q\to L^p$, based on the Trotter 
 product formula,  was refined in  \cite[Chapter 2]{SHk72}  and in \cite[Theorem X.58]{RS2}.
 Our forced confinement of $q, p$  to the interval $(q_0, p_0)$ in \eref{L289} does not 
 show up in these three sources    
 because it was always assumed   
 that  $\|e^{-V}\|_\nu < \infty$ for all $\nu <\infty$.
 Our proof may be considered to be an infinitesimal version  of Segal's method.
}
\end{remark}

\begin{remark} \label{remfed5}  {\rm (Federbush's semi-boundedness theorem). 
If   $H_0$ is a non-negative self-adjoint operator on $L^2(\text{probability measure}\ m)$  
satisfying  a logarithmic Sobolev inequality  
\beq
Ent_m(u^2) \le 2c (H_0 u, u)_{L^2(m)}                  \label{L337}
\eeq
then for any real valued measurable function $V$ there holds
\begin{align}
((H_0 + V) u, u)_{L^2(m)} \ge \(-\log \|e^{-V}\|_{L^{2c}(m)}\) \|u\|_2^2       \label{L338}
\end{align}
for all $u \in D(H_0) \cap D(V)$. 
This is the Federbush semi-boundedness theorem, \cite{Fe, G1, G1993}. 
        $H_0$ need not be a Dirichlet form operator. \eref{L337} and \eref{L338} are in a sense equivalent.
   See \cite[Theorem 2.1]{G1993}.
        
   However if $H_0$ is a Dirichlet form operator then \eref{L338} can be regarded as a limiting form
   of the hyperboundedness inequality  \eref{L291}: Taking $V$ bounded for simplicity 
   and $q=2$ in \eref{L291}  we have  $ \|e^{-tH} \|_{2\to 2} \le \| e^{-V}\|_\nu^t\    \text{for all}\ \ t \ge 0$.
  We may apply the spectral theorem to find $\text{inf  spectrum}\ H \ge - \log \|e^{-V}\|_\nu$. 
   Let $\nu\downarrow 2c$ to find \eref{L338}.
   
   Notice that as $\nu \downarrow 2c$ the interval of validity in Corollary  \ref{corhb2}, 
$(q_0, p_0)$, collapses to the one point set
 $\{2\}$, as one can see from \eref{L313q}, since $a_\nu \downarrow 0$ as $\nu \downarrow 2c$. 
For $\nu = 2c$,  hyperboundedness inequalities such as \eref{L291}, involving the exponential $e^{-tH}$,
 can fail. See for example Theorem \ref{thmgp1} for the modes of such failure.   
}
\end{remark}

\begin{example} \label{ex2p}{\rm The case $q=2$ and  $2 < p < p_0$ will be important for us.   
Suppose that $t$ is the minimum time for \eref{L289} to hold when $q=2$. 
That is, $t = \tau(p)$ because $\tau(2) =0$ by \eref{L325t}. Then by \eref{L505} we have 
\begin{align} 
2a_\nu t/c =  \log\frac{q_0^{-1} - p^{-1}}{p^{-1} - p_0^{-1}},\ \ \ q_0 < p < p_0. \label{L340}
\end{align}
Let $b= e^{2a_\nu t/c}$ and take the exponential of \eref{L340} to find 
$q_0^{-1} - p^{-1} = b ( p^{-1} - p_0^{-1})$. Therefore $(1+b)p^{-1} = q_0^{-1} + b p_0^{-1}$.
Hence the function 
\begin{align}
p(t) \equiv \frac{1+e^{2a_\nu t/c}}{q_0^{-1} + e^{2a_\nu t/c} p_0^{-1}}    \label{L340a}
\end{align}
gives the maximum Lebesgue index for boundedness from $L^2(m)$ to $L^p(m)$ 
predicted by \eref{L289}. That is,
\begin{align}
\| e^{-tH}\|_{2 \to p(t)} \le \|e^{-V}\|_\nu^t\ \ \  \text{for all} \ \ t \ge 0. \label{L289a}
\end{align}
It is instructive to observe that as $t\uparrow \infty$ the index $p(t) \uparrow p_0$.

In Section \ref{secbelow} we will show that if $V$ is bounded below then we may take $\nu = \infty$
and $t_{2,p}$  reduces exactly to Nelson's shortest time to contraction.
}
\end{example}

\subsection{Essential self-adjointness}   \label{secesa}

The computations in Section \ref{secnsh} were proven in case the potential $V$ is bounded. 
If $V$ is unbounded
the operator $H$, defined as the closure of the operator $\n^*\n+ V$, must be shown to be self-adjoint
before inequalities such as \eref{L289} can be given meaning. 
We will show in this section that $\n^*\n + V$ is essentially 
self-adjoint  and that Theorem \ref{thmns2} and Corollary \ref{corhb2} hold in the generality stated.
  We will also prove that the self-adjoint operator $H$ has an eigenvalue of multiplicity one 
at the bottom  of its spectrum 
 which belongs to a unique positive eigenfunction $\psi$.

The methods of this section are based on techniques that have their origin  in 
 the early attempts 
to prove the internal consistency of quantum field theory.  The 
problem there, as here, was to prove that a particular 
 operator of the form $H_0 + V$ is essentially self-adjoint and that its closure has a unique ground state.
The operator $H_0$, of interest at that time, was similar in many ways to our 
operator $\n^*\n$, but had additional special structure. 
All three properties, essential self-adjointness, existence and uniqueness of a ground state, were first
proved by Glimm and Jaffe   \cite{GJ68}, \cite{GJ70}.  
Their proofs made use of some of the special structures of $H_0$ available in that setting, 
that are not shared by our operators $\n^*\n$. 
         I.E. Segal,  \cite{Seg1970}, subsequently  removed     the need
  for the special structure in the 
  proof of essential self-adjointness and replaced it 
  by a hypercontractivity argument.     
  The present author, \cite{G1972}, subsequently 
    removed need  
  of the special  structure in the proof of existence
     of a ground state, replacing it again by hypercontractive notions.
The proofs we will give here are modifications of the latter proofs. They depend 
only on the positivity preserving character of the operators $e^{-tH}$  and the 
hypercontractivity bounds that are already available to us.   
 Simon and Hoegh-Krohn, \cite[Section 2]{SHk72} developed the methods of Segal, \cite{Seg1970}, further.
We will make use of their  techniques also.

The statements and techniques of proof of essential self-adjointness and 
existence and uniqueness of a ground state are dimension independent. Although 
the underlying manifold in this paper is assumed to be finite dimensional, 
these results  can be formulated and the proofs carried out directly in infinite 
dimensions once a suitable notion of differentiation is available. See, 
for example, \cite{AKR1995}  or \cite{MR92} for a systematic exposition 
of Dirichlet forms over infinite dimensional spaces.

\begin{theorem}\label{thmesa2}  $($Essential self-adjointness$)$. Assume that \eref{mt1} holds. Suppose that 
\beq
\int_X e^{-\nu V} dm < \infty \ \ \ \text{for some} \ \ \  \nu > 2c.      \label{EU6}
\eeq
Define $p_0$ as in Lemma  \ref{lemiv1} and assume that
\beq 
V \in L^{p_1}(m)  \ \  \text{for}\ \  p_1 =  2\frac{p_0}{p_0 -2} .     \label{EU7}
\eeq
Then $\n^*\n + V$ is essentially self adjoint on 
\beq
D(\n^*\n) \cap L^{p_0}(m)                                                  \label{EU8}
\eeq
and is bounded below. Denote by $H$ its closure. The semigroup $e^{-tH}$ extends to a strongly continuous
semigroup of bounded operators on $L^q$ for $q \in [q_0,2]$ and restricts to a strongly continuous semigroup
of bounded operators on $L^q$ for $q \in [2,p_0]$.      
For these extensions we have  
\begin{align}
\|e^{-tH} f\|_q \le \|e^{-V}\|_\nu^t\  \|f\|_q \ \ \text{for}\ \ q_0 \le q \le p_0\ \ and \ \ t \ge 0.  \label{EU9}
\end{align}
Moreover
\begin{align}
\|e^{-tH} f\|_p \le \|e^{-V}\|_\nu^t\  \|f\|_q \ \ \text{for} \ \ q_0 < q \le p <p_0\ \ 
          \text{if} \ \ t \ge \tau(p) - \tau(q).                                                                      \label{EU10}
\end{align}
$e^{-tH}$ is positivity preserving for all $t \ge 0$.
\end{theorem}
The proof depends on the following three lemmas.

\begin{lemma}\label{lemEU1a}  
Let $V_1$ and $V_2$ be bounded potentials. Let $H_i = \n^*\n + V_i$, $i = 1,2$.
Let $p_1 =2\frac{p_0}{p_0 -2}$.  
Then
\begin{align}
\|(e^{-t H_1} - e^{-tH_2})f\|_2
 &\le \(\int_0^t\|e^{-V_1}\|_\nu^{t-u} \|e^{-V_2}\|_\nu^u  du\) \  \|V_1-V_2\|_{p_1} \|f\|_{p_0}  . \label{EU19}
\end{align}
\end{lemma}
       \begin{proof} If $H_1$ and $H_2$ are two self-adjoint operators on $L^2(m)$ which have a common domain
$\D$ and are both bounded below then the DuHamel formula
\begin{align}
(e^{-tH_2} - e^{-tH_1}) f =    \int_0^t e^{-(t-u)H_1}(H_1 - H_2) e^{-uH_2} f\ du\ \ \ f  \in \D  \label{EU20}
\end{align}
follows by integrating from $0$ to $t$ the identity 
$(d/du)\(e^{-(t-u)H_1} e^{-uH_2}\) f =\(e^{-(t-u)H_1}(H_1 - H_2) e^{-uH_2}\) f$,
which is valid for $f \in \D$. If $H_1 - H_2$ is a bounded (albeit only densely defined) operator
then \eref{EU20} extends by continuity to all $f \in L^2(m)$.
Thus if  $H_i = \n^*\n + V_i$ we have
   \begin{align}
   (e^{-tH_2} - e^{-tH_1}) f 
                =    \int_0^t e^{-(t-u)H_1}(V_1-V_2) e^{-uH_2} f \ du\ \ \ \forall\ \ f \in L^2(m).     \label{EU21}
   \end{align}
Since \eref{L291} has been proven for bounded $V$ in Section \ref{secnsh} we may use it to find that
 $ \|e^{-(t-u)H_1}\|_{2\to 2} \le \|e^{-V_1}\|_\nu^{t-u}$ for all $t-u \ge 0$,
 while 
 \beq 
  \|e^{-uH_2} f\|_{p_0} \le \| e^{-V_2}\|_\nu^u\ \|f\|_{p_0}      \label{EU22}
  \eeq
   for all $u \ge 0$. 
  Since $p_1^{-1} + p_0^{-1} = 1/2$ we have  
  \begin{align}
  \|(V_1-V_2) e^{-uH_2} f\|_2  
            \le \|V_1 - V_2\|_{p_1}   \| e^{-V_2}\|_\nu^u\ \|f\|_{p_0}\ \ \forall\ \ u \ge0. \label{EU23}
  \end{align}
  Hence
  \begin{align}
   \|(e^{-tH_2} - e^{-tH_1}) f\|_2 
    &\le    \int_0^t \|e^{-(t-u)H_1}\|_{2\to 2}\|(V_1-V_2) e^{-uH_2} f\|_2 \ du \notag \\ 
   &\le  \int_0^t   \|e^{-V_1}\|_\nu^{t-u}  \|V_1 - V_2\|_{p_1}   \| e^{-V_2}\|_\nu^u\  \|f\|_{p_0} du, \notag
\end{align}
which is \eref{EU19}.
\end{proof}

\begin{lemma}\label{lemesa2}  Suppose that $\|e^{-V}\|_\nu < \infty$ 
   and that $V \in L^{p_1}(m)$. For \linebreak 
   $k = 0,1,2, \dots$ let 
\begin{align}
V_k = (-k\vee V)\wedge k.    \label{EU40}
\end{align}
Define $H_k = \n^*\n + V_k$ and let
\begin{align}
S_k(t) = e^{-tH_k}, \ \ \ t \ge 0. \label{EU41}
\end{align}
Then the sequence $S_k(t)$ converges strongly in $L^2$  to a bounded positive 
operator $S(t)$ for each $t \ge 0$.
$S(\cdot)$ is a strongly continuous semigroup of  bounded positivity preserving operators on $L^2(m)$. 

For each $t \ge 0$, $S(t)$ extends uniquely  to a bounded operator on $L^q$  for  
 $q \in [q_0, 2]$ and restricts to a bounded operator on  $L^q(m)$ for  $q \in [2,p_0]$.
The extensions and restrictions form strongly continuous semigroups in these spaces.

      Denoting the extensions and restrictions by $S(t)$ we have, for all $f \in L^q$, 
\begin{align}
\|S(t) f\|_q 
     &\le \|e^{-V}\|_\nu^t\  \|f\|_q \ \ \text{for}\ \ q_0 \le q \le p_0\ \ and \ \ t \ge 0\ \ \text{and} \label{EU45}\\
\|S(t) f\|_p 
    &\le \|e^{-V}\|_\nu^t\  \|f\|_q \ \ \text{for} \ \ q_0 < q \le p <p_0\ \ \text{if} \ \ t \ge \tau(p) - \tau(q). \label{EU46}
\end{align}
\end{lemma}
      \begin{proof} By the monotone convergence theorem on $\{ V\le 0\}$ and 
 dominated convergence theorem on $\{V>0\}$  we have 
 $\lim_{k\to \infty} \|e^{-V_k}\|_\nu \to \|e^{-V}\|_\nu$. Moreover,
 since $-V_k \le -V$ wherever $V \le 0$ it follows that $0 \le e^{-V_k} \le e^{-V} + 1$ everywhere and
 therefore $\|e^{-V_k}\|_\nu \le \|e^{-V}\|_\nu +1$ for all $k$. Hence 
\begin{align}
      \int_0^t \|e^{-V_k}\|_\nu^{t-u} \|e^{-V_n}\|_\nu^u du \le   t\(\|e^{-V}\|_\nu +1\)^t.
\end{align}    
Apply \eref{EU19} to the potentials $V_k$ and $V_n$  
to find
\begin{align}
\|(S_k(t) - S_n(t))f\|_2 \le  t\(\|e^{-V}\|_\nu +1\)^t \|V_k - V_n\|_{p_1}  \|f\|_{p_0}.
\end{align} 
Since $\|V_k - V_n\|_{p_1} \to 0$ as $k,n \to \infty$, it follows that for fixed $f \in L^{p_0}$ and $T < \infty$ the sequence $S_k(t)f$ converges   in 
$L^2(m)$  uniformly on $[0, T]$ as $k\to \infty$.

Denote the limit by $\hat S(t)f$.  $\hat S(t)$ is 
a linear operator from $L^{p_0}(m) \to L^2(m)$ for each $t \in [0, T]$ and $\hat S(t)f$ is 
continuous in $t \in [0,T]$ into $L^2$ for each $f \in L^{p_0}$ and each $T>0$. 

For fixed $t>0$ and $f \in L^{p_0}$ there is a subsequence $k_j$   such that $S_{k_j}(t) f$ 
converges to $\hat S(t)f$  pointwise almost everywhere. By   \eref{L291}, which has already been proven
for bounded $V$ in Section \ref{secnsh}, we have $\|S_{k_j}(t) f\|_q \le \|e^{-V_{k_j}}\|_\nu^t\  \|f\|_q$.
 Apply Fatou's lemma on the left to find 
 \beq
\|\hat S(t) f\|_q \le \|e^{-V}\|_\nu^t\  \|f\|_q,\ \  q \in [q_0, p_0],\ \   f \in L^{p_0}  \label{EU47}
\eeq
    The same argument also shows that \eref{EU46}  holds for $\hat S(t) f$ when $f \in L^{p_0}$.

Note  for later use the uniform (in $k$) bound 
\begin{align}
\|S_{k}(t) f\|_q \  \le (\|e^{-V}\|_\nu +1)^t\  \|f\|_q\ \ \text{for all}\ \  q \in [q_0, p_0]    \label{EU48}
\end{align}
 and all $f \in L^q$, which follows from \eref{L291},  with $V$ replaced 
by the bounded function $V_k$,  
 and using the bound
$\|e^{-V_k}\|_\nu \le \|e^{-V}\|_\nu +1$.  

Since $L^{p_0}$ is dense in $L^q$ for each $q \in [q_0, p_0]$ we may, by virtue of \eref{EU47},
extend $\hat S(t)$ by continuity  in  $L^q$ norm  
to a bounded linear operator from $L^q$ into $L^q$, which we  denote by $S_q(t)$. 
\eref{EU45} and \eref{EU46}  hold for all $f \in L^q$ for this extended operator. The extensions are easily seen
to be consistent in the sense that if $q_0 \le q_1 \le q_2 \le p_0$ then $S_{q_1}$, restricted to $L^{q_2}$,
is $S_{q_2}$.  We will drop the subscript $q$ and just write $S(t)$, which now acts on each space $L^q$
as a bounded operator satisfying \eref{EU45} and \eref{EU46}.

      In case  $q\in [q_0, 2]$, the extended operator $S(t)$ is also a strong limit on all of $L^q$ of the 
 operators $S_k(t)$. Indeed, if  $g \in L^q$ and
 $\|f_n - g\|_q \to 0$ for some sequence $f_n \in L^{p_0}$ then, 
 \begin{align*}
 &\|S(t)g - S_k(t)g\|_q \\
 &\le \| S(t)g - \hat S(t)f_n\|_q + \| \hat S(t)f_n -S_k(t) f_n\|_q + \|S_k(t)f_n -  S_k(t)g\|_q \\
 &\le \| S(t)g - \hat S(t)f_n\|_q + \| \hat S(t)f_n -S_k(t) f_n\|_2 + \|S_k(t) \|_{q\to q} \|f_n - g\|_q. 
 \end{align*} 
 The first term on the right goes  to zero as $n\to \infty$ by the definition of $S(t)g$.
 The second goes to zero for each $n$ as  $k\to \infty$ by the definition of $\hat S(t)f_n$.
 The third term goes to 
 zero as $n\to \infty$ by the uniform bound \eref{EU48}.
 A standard argument completes the proof.
  
          The semigroup property        $S(t+s) = S(t) S(s)$        
  follows from the fact that the operators    $S_k(t)$ form semigroups
  and, for $q\in [q_0, 2]$, converge strongly in $L^q$,   boundedly in $k$, to $S(t)$ for each $t >0$.
  The semigroup equation also holds  of course when restricted to the spaces $L^p; p \in [2, p_0]$.
        
          Each operator $S_k(t)$ is positivity preserving because $e^{-t\n^*\n}$ is positivity preserving
 while the Trotter product formula 
 \begin{align}
 e^{-t(\n^*\n + V_k)} = \text{strong limit}_{n\to \infty}\( e^{-(t/n)(\n^*\n)} e^{-(t/n)V_k}\)^n
 \end{align} 
 is applicable in $L^2$, since $V_k$ is bounded. $ S(t)$, being a strong limit in $L^2$ of the operators $S_k(t)$,
  is therefore also positivity preserving.

       For the proof of strong continuity we must go back to the operators $\hat S(t)$. For each $f \in L^{p_0}$, 
$\hat S(t)f$ is continuous in $t$ as a function into $L^2$ and therefore as a function into
$L^q$ if $q_0 \le q \le 2$.   
 Hence if $g \in L^q$ with $q_0 \le q \le 2$  and $f_n \to g$ in $L^q$ for some sequence 
 $f_n \in L^{p_0}$   then, since $S(t) f_n = \hat S(t) f_n$, we have 
$S(t)g = \lim\ \text{in}\ L^q\  S(t)f_n$ 
and the convergence is uniform for $t \in [0,T]$ by \eref{EU45}. 
Therefore $S(\cdot)g$ is continuous into $L^q$ for each $g \in L^q$. $S(t)$ is therefore  a strongly 
continuous  semigroup on the spaces $L^q:  q \in [q_0,2]$ and therefore weakly continuous
on the dual spaces $L^p, \, p \in [2,p_0]$ because the operators $S(t)$ are symmetric.  
       But these spaces  are reflexive. 
  Consequently   $S(\cdot)$  is strongly continuous on $L^p: 2 \le p \le p_0$ 
  by the general theorem \cite[Theorem 1.6]{EN}.   
\end{proof}

  $S(t)$ has been constructed as a limit of the semigroups $e^{-tH_k}$. 
We should expect that $S(t) = e^{-tH}$ with $H = \text{closure of}\ \n^*\n +V$.  
The next lemma proves that this is the case.

      \begin{lemma} \label{lemesa3} Suppose that $\|e^{-V}\|_\nu < \infty$ and that $V \in L^{p_1}(m)$.
 Denote by  $H$  the infinitesimal generator of the strongly continuous 
 semigroup $S(t)$ in $L^2(m)$. Denote by  $H_0$  the self-adjoint Dirichlet form operator $\n^*\n$ for $m$.
 Then 
 \begin{align}
D(H)\cap L^{p_0} = D(H_0)  \cap L^{p_0}      \label{EU55}
\end{align}
and
\begin{align}
H= \text{closure of}\ (H_0 +V)\ \ \text{in}\ \ L^2(m).  \label{EU56a}
\end{align}
\end{lemma}
        \begin{proof}      
Choose $H_1 = H_k$ and $H_2  = H_0$ in \eref{EU21}.  Since $V_0 =0$ we have
\begin{align}
(e^{-tH_0} -e^{-tH_k})f =  \int_0^t e^{-(t-u)H_k} (V_k - 0) e^{-uH_0} f \ du, \ \ \ f \in L^{p_0}. \label{EU49}
\end{align}
Fix $f \in L^{p_0}(m)$. Since $p_1^{-1} + p_0^{-1} = 1/2$ and $V \in L^{p_1}(m)$ we see that
 $\| (V_k -0)e^{-uH_0}f - Ve^{-uH_0}f\|_2  \le \|V_k - V\|_{p_1} \| e^{-uH_0}f\|_{p_0} \to 0$ uniformly
 for $0 \le u \le t$ as $ k \to \infty$.
 Moreover,  by Lemma \ref{lemesa2}, $e^{-(t-u)H_k}$ converges strongly 
in $L^2$ and boundedly.
Therefore  $e^{-(t-u)H_k} (V_k - 0) e^{-uH_0} f$ converges in $L^2$ boundedly for $u\in [0,t]$ to
$e^{-(t-u)H} V e^{-uH_0} f$. We may take the
limit as $k \to \infty$ in \eref{EU49}  to find
\begin{align}
(e^{-tH_0} -e^{-tH})f =  \int_0^t e^{-(t-u)H} V e^{-uH_0} f \ du, \ \ \ f \in L^{p_0}.   \label{EU50c}
\end{align}
because, by the assumption of this lemma,  $S(t) = e^{-tH}$ on $L^2$ and therefore on $L^{p_0}$.
The integrand is a continuous function of $u$ into $L^2$. Multiply \eref{EU50c} by $t^{-1}$ and 
rearrange to find, for $f \in L^{p_0}$,
\begin{align}
t^{-1}(I -e^{-tH_0})f =  t^{-1}(I -e^{-tH})f  - t^{-1}\int_0^t e^{-(t-u)H} V e^{-uH_0} f \ du. \label{EU51}
\end{align}
If $f \in D(H) \cap L^{p_0}$ then  both
terms on the right side of \eref{EU51} converge in $L^2$ as $t\downarrow 0$, and therefore the limits 
on the right and left both exist. The limit on the right is $Hf - Vf$. Since the limit on the left exists
we know that $f \in D(H_0)$ and moreover the limit is $H_0f$.
Thus  $D(H) \cap L^{p_0} \subset D(H_0)$ and $(H-V)f= H_0 f$ for $f \in D(H)\cap L^{p_0}$. 
Therefore
\begin{align}
H = H_0+V \ \  \text{on}\ \   D(H) \cap L^{p_0}.  \label{EU56}
\end{align}
Similarly, if $f\in D(H_0) \cap L^{p_0}$ then 
the left side and second term on the right side of \eref{EU51} converge in $L^2$ and 
therefore $f \in D(H)$.  This proves \eref{EU55}.

Let $K$ be the closure of $H_0 +V$ in $L^2$.
$K$ is a closed symmetric operator on $L^2$. We wish to prove that $H= K$.
It suffices to show that there is a set $\D \subset D(H_0) \cap D(V)$  
 which is a core for $H$
and such that $H = H_0 + V$ on $\D$. 
For then $K \supset \text{closure of}\ \(\{H_0+V\}\Big|\D\) =H$ 
and  therefore $K^*\subset H \subset K$. Since $K^* \supset K$ it will then follow that $K = K^*= H$.

Let $a = 1 +\log \|e^{-V}\|_\nu$. Then \eref{EU45} shows that, for all $q \in [q_0, p_0]$ we have
$\|S(t) f\|_q \le e^{(a-1)t}\|f\|_q$. For $q\ge 2$ we may write this 
as $ \|e^{-tH}f\|_q \le  e^{(a-1)t} \|f\|_q$ because $L^q \subset L^2$. 
Therefore $ \|e^{-t(H+a)}f\|_q \le  e^{-t} \|f\|_q$ for $2 \le q \le p_0$. Hence
\begin{align}
\| (H+a)^{-1} f\|_q &= \| \int_0^\infty e^{-t(H+a)} f dt\|_q \le \int_0^\infty e^{-t}dt \|f\|_q = \|f\|_q.   \label{EU58}
\end{align}
For $q=2$ this shows that $\|(H+a)^{-1}\|_{2\to 2} \le 1$.
Let $\D = (H+a)^{-1} L^{p_0}$.  Then $\D \subset D(H)$ because  
$(H+a)^{-1} L^{p_0} \subset (H+a)^{-1} L^{2}$. Further, from \eref{EU58} with $q = p_0$ we see that 
$\D \subset L^{p_0}$.
Hence  
\beq
\D \subset D(H) \cap L^{p_0} = D(H_0)\cap L^{p_0} \subset D(H_0) \cap D(V),
\eeq
 From this and \eref{EU56}  it follows that $H = H_0 + V$ on $\D$.

It remains  to show that $\D$ is a core for $H$.    Suppose that $\phi \in D(H)$. 
    Since $(H+a) \phi \in L^2$ there is a sequence $f_n \in L^{p_0}$ which 
converges to $(H+a) \phi $ in $L^2$. The functions  $g_n \equiv (H+a)^{-1}f_n$ are in $\D$ and
$(H+a)g_n = f_n$, which   converge to $(H+a) \phi$ in $L^2$. But also $g_n \to \phi$ because
$\|(H+a)^{-1}\|_{2\to 2} \le 1$.
Therefore $\D$ is a core for $H$. This completes the proof of \eref{EU56a}.
\end{proof}

\bigskip
\noindent
\begin{proof}[Proof of Theorem \ref{thmesa2}]  It was shown in Lemma \ref{lemesa3} that the closure, $H$,
of $\n^*\n +V$ is the infinitesimal generator of the semigroup  $S(t)$ and that this semigroup 
extends and restricts to $L^q$ spaces as asserted in the statement of the theorem. 
The asserted inequalities \eref{EU9} and \eref{EU10} are restatements 
of \eref{EU45} and \eref{EU46} respectively.
\end{proof}

\begin{corollary} \label{cordom}  Under the hypotheses of Theorem \ref{thmesa2}
there holds
\begin{align}
Q(H) = Q(\n^*\n) \cap Q(|V|).   \label{EU60}
\end{align}
where $Q(A)$ denotes the form domain of a closed semi-bounded operator $A$.
\end{corollary}
     \begin{proof} If $\int_X |\n u|^2 dm < \infty$ and $\int_X |V|u^2 dm < \infty$ then 
  $u$ is clearly in the form domain of $H$. Therefore $Q(H) \supset Q(\n^*\n) \cap Q(|V|)$.
  The reverse containment requires specific information about $V$.
  
  Let $V_- = \sup(-V, 0)$. Choose  $\nu_1 \in (2c, \nu)$ and let $\ep  = (\nu/\nu_1) - 1$. Define 
a new potential by $W = V - \ep V_-$. On the set $\{ V\le 0\}$ we have $V_- = -V$. Therefore, on this set
\begin{align*}
-\nu_1 W & = -\nu_1 V + \ep \nu_1 V_-\\
   &= \nu_1(1+\ep) \(-V \) \\
   & = -\nu V.
 \end{align*}
 Since $W = V$ where $V >0$, the hypothesis   \eref{EU6}  ensures that $\int_X e^{-\nu_1 W} dm < \infty$.
By the Federbush semiboundedness theorem, \eref{L338}, applied to the potential $W$, we therefore have
\begin{align}
\n^*\n +(V -\ep V_-) \ge -b \equiv -\log \|e^{-(V-\ep V_-)}\|_{\nu_1}  >-\infty.     \notag
\end{align}
So      $ \ep V_- \le H + b$. Hence 
\begin{align}
 V_- \le \ep^{-1}(H + b).   \label{EU65}
\end{align}
Therefore
\begin{align} 
 H_0 + |V| = H_0 + V + 2V_-  = H + 2V_- \le  (1 + 2\ep^{-1})H + 2\ep^{-1} b.  \notag 
\end{align}
All of these inequalities hold upon taking the inner products $(\cdot\ u, u)$ 
with $u$ in the operator core  given by \eref{EU8}. Since the operator core is also a form core 
it follows that  $Q(H) \subset Q(H_0+ |V|)$, which is the right hand side of \eref{EU60} 
because $H_0$ and $|V|$ are  both non-negative.       
\end{proof}

\begin{remark} {\rm The equality \eref{EU60} will be needed 
 in the proof of Theorem \ref{thmA3}.
 It was needed in  similar contexts in \cite{GRW2001} and  \cite{Aida2001}, where 
 it was taken as a natural hypothesis.
 }
\end{remark}

 \subsubsection{Proof of non-standard  hyperboundedness} \label{secnsh2}

The proofs of Theorem \ref{thmns2} and Corollary \ref{corhb2} under the general conditions on $V$
specified there are largely consequences of the results in Sections \ref{secnsh} and  \ref{secesa}:
The qualitative assertions in Theorem \ref{thmns2} concerning essential self-adjointness of $\n^*\n +V$
and extendability of $e^{-tH}$ are proved in Theorem \ref{thmesa2}. 
The logarithmic Sobolev inequalities  \eref{L325} 
follow, by \cite[Theorem 2]{G1},  from the hyperboundedness inequalities  \eref{L289} asserted 
in Corollary \ref{corhb2} once the relation between the Sobolev coefficients $c_\nu(p)$ in 
Theorem \ref{thmns2} and the 
minimum time to boundedness $\tau(p) - \tau(q)$ in Corollary \ref{corhb2} are 
established, as they are in the Section \ref{secnsh} on bounded potentials.
The hyperboundedness inequalities  \eref{L289} and \eref{L291} were   proved in 
Theorem \ref{thmesa2} for the desired class of potentials.

\subsection{Existence and uniqueness of a  ground state} \label{secEU}

 \begin{theorem} \label{thmE1} $($Existence of a ground state$)$. 
Suppose that $m$ satisfies a  
logarithmic Sobolev inequality \eref{mt1}.
Assume that $\|e^{-V}\|_\nu < \infty$ for some $\nu > 2c$. Define $p_1$ as in \eref{EU7} and assume
that $V\in L^{p}$ for  some $p \ge p_1$. 
Then the closure, $H$, of  $\n^*\n + V$ is self adjoint,
bounded below, and the bottom of its spectrum is an eigenvalue of finite multiplicity.
\end{theorem}
   \begin{proof} We already know  from Theorem \ref{thmesa2} that $H$ is self-adjoint and bounded below,
   and that $e^{-tH}$ is positivity preserving for all  $t \ge 0$. 
   We need only show that the bottom of the spectrum of $H$ is an eigenvalue of finite multiplicity. 
   Referring to the notation in Lemma \ref{lemiv1},  
   choose a number $p_2 \in (2, p_0)$ and  a number $t \ge \tau(p_2)$.
   By  \eref{L289}, with $q = 2$,  $e^{-tH}$ is  bounded from $L^2$ to $L^{p_2}$.
   By  \cite[Theorem 1]{G1972}  
   the operator  norm $\|e^{-tH}\|_{2\to 2} \equiv :\mu$    is an eigenvalue of $e^{-tH}$ of finite multiplicity. 
   The spectral theorem shows that 
   $\l_0 \equiv -t^{-1}\log \mu$ is an eigenvalue of $H$ of finite multiplicity and $\l_0 = \inf\ spectrum\ H$.
   \end{proof}

 The techniques in the following theorem and lemma are distilled from
   \cite{GJ68,GJ70,Seg1970,G1972,SHk72}. 
   
\begin{theorem}\label{thmU2} $($Uniqueness of the ground state$)$.
   Let $m$ be a probability measure on some space $X$. 
 Suppose that $H_0$ is a non-negative self-adjoint operator on 
$L^2(m)$ and that $V$ is a potential in $L^2(m)$. Suppose that $H_0+V$ is essentially self-adjoint.
Denote its closure by $H$. Assume that
\begin{align}
&a.\ \text{The nullspace of $H_0$ is spanned by the constant functions.}     \label{U4} \\
&b.\ D(H)\cap L^\infty = D(H_0)\cap L^\infty.      \label{U5} \\
&c.\ \text{$e^{-tH}$ is positivity preserving for all $t >0$}  \label{U6}\\
&d.\ \text{$\l_0 \equiv \inf spectrum\ H$ is an eigenvalue.}      \label{U7}
\end{align}
\noindent
Then $\l_0$ has multiplicity 1 and belongs to an a.e. strictly positive eigenfunction.
\end{theorem}
 The proof depends on the following lemma.
 
\begin{lemma} \label{lemU2} 
Let $m$ be a probability measure on some space $X$. Suppose that $H_0$ is a non-negative self-adjoint operator on $L^2(m)$ and that
$V$ is a potential in $L^2(m)$. Suppose that $H_0 +V$ is essentially self-adjoint with closure $H$
and that conditions a. and b. of Theorem \ref{thmU2} hold.

If $f$ is a bounded measurable function such that its multiplication operator,  $M_f \equiv $ multiplication by $f$,
commutes with  $e^{-tH}$ for some $t >0$ then there is a real constant $C$ such that  $f= C$ a.e..
\end{lemma}
     \begin{proof} If $M_f$ commutes with $e^{-tH}$ for some $t>0$ then it commutes for all $t >0$ by the spectral theorem.
Since $M_f 1 = f$ we have 
\begin{align}
M_fe^{-tH} 1 = e^{-tH}f\ \ \ \forall t \ge0.     \label{U9}
\end{align}
Since $1$ is in the domains of $H_0$ and $V$ it is
in $D(H)$. Hence the left side is differentiable at $t = 0$ and therefore 
so also is the right side.  Thus $L^\infty \ni f \in D(H)$. 
By \eref{U5} we  have then $f \in D(H_0)\cap L^\infty$.  
Differentiating \eref{U9} at $t=0$  gives $M_fH1= Hf$. That is,
\begin{align}
f(H_0 + V) 1= (H_0 + V) f.      \label{U10}
\end{align} 
Since $H_0 1=0$ it follows from \eref{U10} that $H_0f = 0$. The lemma now follows from \eref{U4}. 
\end{proof}

\bigskip
\noindent
\begin{proof}[Proof of Theorem \ref{thmU2}]
Pick $t >0$. Any eigenfunction $\psi$ for $H$ with eigenvalue $\l$ satisfies $e^{-tH}\psi = e^{-t\l}\psi$.
Let $A = e^{-tH}$ and $\mu = e^{-t\l}$. Since $A$
is positivity preserving we have  $\mu |\psi| = |A\psi| \le A|\psi|$. 
Therefore $\mu (|\psi|, |\psi|)  \le (A|\psi|, |\psi|)$
and since $\mu = \sup \text{spectrum}\ A$, $|\psi|$ is also an eigenfunction of $A$ belonging to $\mu$.
So also are $|\psi| \pm \psi$. At least one of these is non zero.  
So $(A - \mu)\phi =0$ for some almost everywhere non-negative function $\phi$ which is not identically zero.

Let $E = \{ x: \phi(x) =0\}$ and let $f = \chi_E$. If $L^2 \ni g \ge 0$  then
\begin{align}
(A(fg), \phi) =  (f g, A\phi) = (f g, \mu \phi) =0\  \text{because} \ \ f\phi =0.
\end{align}
Since $A(f g) \ge 0$  and is orthogonal to $\phi$ it must be supported in  $E$. Therefore 
$A(f g)  = f A(f g)$ for all non-negative $g \in L^2$ and hence for all $g \in L^2$. Thus
$AM_f = M_f AM_f $,   where $M_f$ is the Hermitian projection consisting of multiplication by $f$. 
Take adjoints to find $M_fA = M_f AM_f = AM_f$.  So $M_f$ commutes  with $A$.
By Lemma  \ref{lemU2},  $f$ is constant and therefore a.e. equal to 0 or 1. 
  It can't be equal to 1 a.e. because $\phi > 0$ on a set of strictly positive measure. 
 Therefore $E$ is the empty set (up to a set of measure $0$).
So $\phi > 0 $ a.e..
      Thus either $|\psi| - \psi >0$ a.e., in which case $\psi = - |\psi|$ a.e, or else 
      $|\psi| + \psi >0$ a.e., in which case $\psi = |\psi|$ a.e..
    Thus any eigenfunction is either strictly positive a.e. or strictly negative a.e. Since 
 two such functions cannot be orthogonal  the eigenspace has dimension one.
\end{proof}

\bigskip
\noindent
\begin{proof}[Proof of Corollary \ref{corEU}] Since the hypotheses of Theorem \ref{thmns2} imply 
the hypotheses of Theorem \ref{thmE1} the latter theorem ensures that $H$ has an eigenvalue 
at the bottom of its spectrum of finite multiplicity. 
  Theorem \ref{thmU2} will show that the eigenvalue has multiplicity one and belongs to an
    a.e. positive eigenfunction once we verify the four conditions \eref{U4}-\eref{U7}. 

    We already know that \eref{U7} holds by Theorem \ref{thmE1}.
\eref{U6} was proven in Theorem \ref{thmesa2}. \eref{U5} follows from \eref{EU55}  by intersecting both sides
with $L^\infty$.    For the proof of \eref{U4} observe that in our case $H_0 =\n^*\n$ and the measure $m$
satisfies the logarithmic Sobolev inequality \eref{mt1}. 
The constant functions therefore span the nullspace,
and in fact there is a spectral gap,  by the Rothaus-Simon theorem  \cite{Rot1}, \cite{Simon1976}. For
a direct proof of the Rothaus-Simon theorem see \cite[Theorem 2.5 ]{G1993}  or \cite[Proposition 5.1.3]{BGL}.
This completes the proof of Corollary \ref{corEU}. 
\end{proof}

\subsection{Upper and lower bounds on $\|\psi\|_p$ for $p>0$}     \label{seculb}

\begin{theorem} \label{thmulb}
Assume that \eref{mt1} holds and that $\|e^{-V}\|_{L^\nu(m)} < \infty$ for some $\nu > 2c$. Suppose also
that $V \in L^{p_1}$ as in \eref{EU7}. Denote by $\psi$ the ground state of $\n^*\n +V$. Then
 \begin{align}
 \int_X \psi^2 \log\psi\ dm &\le c_\nu\log \|e^{\l_0 - V}\|_\nu  \label{L510}\\
 \|\psi\|_p &\le  \|e^{\l_0 - V}\|_\nu^{\tau(p)},\ \   2 \le p <p_0 \label{L511} \\
 \|\psi\|_r  &\ge    \|e^{\l_0-V}\|_\nu^{-\sigma} ,\ \ \  0 < r <2,  \label{L512}  
 \end{align}
 where $\tau(p)$ is given by \eref{L505}, $c_\nu$ is given by \eref{L329a},  and 
\begin{align}
\sigma = c_\nu(2r^{-1} -1).     \label{L513}
\end{align} 
 \end{theorem}
    \begin{proof} Since $\|\psi\|_2 = 1$ we have   $Ent_m(\psi^2) = 2 \int \psi^2 \log \psi\ dm$. 
    Choosing $u = \psi$ in \eref{L325h} we therefore find
    \begin{align}
 \int_X \psi^2 \log \psi\  dm &\le c_\nu\<(H+\log \|e^{-V}\|_\nu)\psi, \psi\> \\
 &= c_\nu(\l_0 + \log \|e^{-V}\|_\nu),
 \end{align}
 which proves \eref{L510}.
    
     For the proof of \eref{L511}   we apply \eref{L289} with $q=2$. Since $\tau(2) = 0$ and
  $H\psi = \l_0 \psi$ it follows from \eref{L289} that  
\begin{align}
e^{-t\l_0} \|\psi\|_p = \|e^{-tH}\psi\|_p \le \|e^{-V}\|_\nu^t \|\psi\|_2   \ \text{if}\ \ \ t \ge  \tau(p). \label{L294a}
\end{align}
Hence  $\|\psi\|_{p} \le  \| e^{(\l_0-V)}\|_\nu^t \ \ \ \text{if}\ \ \ t \ge \tau(p)$.  In particular \eref{L511} holds.

Since $(2-r)c_\nu = \sigma r$ it follows from \eref{L510} that 
$ (2-r) \int_X \psi^2 \log\psi\ dm \le \log \|e^{\l_0 - V}\|_\nu^{\sigma r}$
 and therefore 
\begin{align}
(r-2) \int_X \psi^2 \log\psi\ dm &\ge \log \|e^{\l_0 - V}\|_\nu^{-\sigma r}.     \label{L514}
\end{align}
Now $m_\psi \equiv \psi^2 m$ is a probability measure. Using Jensen's inequality we find
\begin{align}
\exp\((r-2) \int_X \psi^2 \log \psi dm\) &=\exp\(\int_X (\log\psi^{r-2}) dm_\psi\)  \notag\\
&\le\int_X \exp(\log\psi^{r-2}) dm_\psi       \notag\\
&= \int_X \psi^{r} dm.    \label{L515}
\end{align}
Combine this with \eref{L514} to find \eref{L512}.
\footnote{The author thanks Barry Simon for a remark in a letter of  August 5, 1973 
that led to the simple proof of \eref{L512}.}
\footnote{Inequalities similar to those in the proof of Theorem \ref{thmulb} 
 can be found in \cite[page 33]{Mat98}}
\end{proof}
 
\begin{remark} {\rm
Theorem \ref{thmulb} relies on hypercontractivity in the hypothesis to achieve a lower 
bound on $\|\psi\|_r$ in terms of $\|\psi\|_2$. Without some extra condition 
on $\psi$ beyond $\|\psi\|_2 =1$ such a lower 
bound cannot hold.  For example if $\psi =\epsilon^{-1} \chi_{[0,\epsilon^2]}$ in $L^2([0,1])$
then $\int_0^1\psi^2 dx = 1$ while $\int_0^1 \psi(x) dx = \epsilon$. So the $L^2$ norm does not control the
$L^1$ norm in the absence of some further condition, such as hypercontractivity. 
}
\end{remark}

\section{The product of moments $\|\psi\|_r \|\psi^{-1}\|_s$} \label{secpm} 

\subsection{The moment product theorem }    \label{secmpt}

 In the previous section we were concerned with the existence and uniqueness
of the ground state $\psi$ and its growth properties, as measured by the $L^p$
 norms $\|\psi\|_{L^p(m)}$ with $  p >0$. The key determiner of these properties was
 the behavior of the negative part 
 of the potential, as measured by $\|e^{-V}\|_{L^\nu(m)}$. In the present section we will be concerned with 
 the decay behavior of $\psi$  where  it is small, as measured by the norms $\|\psi^{-1}\|_{L^s(m)}$. 
 The key determiner
 of this behavior will be the positive part of the potential, as measured by $\|e^{V}\|_{L^\kappa(m)}$.
 The maximum value of $s$ for which we can establish such bounds is given in the following notation
 and theorem.
 
\begin{notation}\label{notkappa3}{\rm For $\ka >0$ let
 \beq
 b_\ka = \sqrt{1 +(2c/\kappa)}.     \label{s98f} 
 \eeq
  The quadratic equation 
\begin{align}
t^2 -(2\kappa/c) (t+1) = 0   \label{W850}
\end{align} 
has two solutions of opposite sign, $s_0>0$ and $-r_0<0$,  given by  
\begin{align}
   s_0 &=(\kappa/c)\(  b_\ka +1 \)\ \ \text{and} \ \  
            r_0=  (\kappa/c)\( b_\ka -1 \) <1,                                 \label{W851g}
 \end{align} 
 in accordance with the quadratic formula: 
 \beq
 t =(1/2)\{2\ka/c \pm \( (2\ka/c)^ 2+4(2\ka/c)\)^{1/2} \}= (\ka/c)(1 \pm b_\ka).
 \eeq
 The assertion that $r_0<1$  in \eref{W851g}  follows directly from  \eref{W850}, 
 which shows that $t+1 >0$ for any solution,  and in particular for $t= -r_0$. 
      For later use note that\eref{W851g}  is equivalent to 
 \begin{align}
 s_0 = \frac{2}{b_\ka -1},\ r_0 = \frac{2}{b_\ka +1}\ \ \text{and also to} \ \  2s_0^{-1} +1 = b_\ka = 2r_0^{-1} -1.    \label{W851k}
 \end{align}
 The quadratic function in \eref{W850} factorizes as
 \begin{align}
 (2\ka/c)(t+1) - t^2 = (s_0-t)(t+r_0).    \label{W851f}
 \end{align} 
 With $c_\nu$ defined as in \eref{L329a}, define
 \begin{align}
 \ell(t) &= \frac{c}{2b_\ka}\log \frac{t + b_\ka c_\nu}{t- b_\ka c_\nu}, \ \ \text{for}\ \ t > b_\ka c_\nu. \label{W852}
 \end{align}
 For later use note that 
 \begin{align}
 &\ell(t) - t \ \ \text{has a unique zero point in} \ \  (b_\ka c_\nu, \infty)\ \  \text{and} \label{W854}\\
 &\ell(t) + t\ \   \text{has a unique minimum point in} \ \ (b_\ka c_\nu, \infty).    \label{W854a}
 \end{align}
 The  statement  \eref{W854} follows from the fact that $\ell(t)$ is strictly decreasing from $\infty$ down to
 zero  on the given interval while $t$ strictly increases to $\infty$ on the interval.
 \eref{W854a} follows from the computation
  \beq 
 \ell'(t) = \frac{-c c_\nu }{t^2 - (b_\ka c_\nu)^2},\ \ t > b_\ka c_\nu,    \label{W854c}
 \eeq
 which shows that the derivative $\ell'(t) + 1$ is strictly increasing from
 $-\infty$ to $1$ on the interval $\{t > b_\ka c_\nu\}$.
}
\end{notation}

\begin{theorem} \label{thmmp1}  $($Moment product theorem$)$. 
Assume the hypotheses of Theorem \ref{thmesa2} hold.
Let $\ka >0$ and define $r_0$ and $s_0$ 
as in Notation \ref{notkappa3}. Assume further that $\|e^V\|_\ka < \infty$. Suppose that 
\begin{align}
0<r< r_0\ \ \ \text{and}\ \ \  0 < s < s_0.    \label{W852a}
\end{align}
Let
\begin{align}
\sigma &=(2r^{-1} - 1) c_\nu \ \ \ \text{and}\ \ \  a = (2s^{-1} +1) c_\nu.  \label{W852b}
\end{align}
Then 
\begin{align}
\|\psi\|_r \| \psi^{-1}\|_s \le   \| e^{V-\l_0}\|_\kappa^{\ell(a) + \ell(\sigma)}     \label{W753g}
\end{align}
\end{theorem}
The proof will be given in the next four subsections.

\begin{remark} \label{rempr} {\rm We will see by example in Section \ref{secgp} that the peculiar 
restriction on $r$ and $s$ ensuing from the quadratic equation \eref{W850}
 is not  an artifact of the proof.
}
\end{remark}

\subsection{Aida's identity}

\begin{theorem}\label{thmA3}  {\rm  $($Aida, \cite[Equ. (3.26)]{Aida2001}$)$.}
Assume the hypotheses of Theorem \ref{thmesa2}. 
Denote by $\psi$ the ground state for $H \equiv \n^*\n +V$ $($closure$)$.   Informally, 
\begin{align}
\n^*\n \psi + V \psi = \l_0 \psi.     \label{sch}
\end{align}
Let 
\begin{align}
F = - \log \psi        \label{W6}
\end{align}
and let $v :\R \to \R$ be a bounded $C^1$ function with bounded derivative. 
Then, writing $v(F)$ for the composition of  $v$ with $F$, we have
\begin{align}
\int_X (v'(F) + v(F)) |\n F|^2 dm = \int_X v(F) (V -\l_0) dm. \ \  \text{Aida's identity.}      \label{W30a}
\end{align}
In particular,
\begin{align}
\ \ \ \int_X |\n F|^2 dm  = \int_X (V - \l_0) dm.    \label{W31}
\end{align}
\end{theorem}
        \begin{proof} 
  Since $\psi$ is in the domain of $H$ it is also in the form domain of $H$.  
  By\eref{EU60} it is therefore in the form domain of $\n^*\n$. That is, $\n \psi \in L^2(m)$.
  Let $\ep >0$ and define $\psi_\ep = \psi +\ep$ and $F_\ep = - \log \psi_\ep$.
   Since $\psi \ge 0$, $\psi_\ep^{-1}$ is bounded and $\psi/\psi_\ep \le 1$. 
  Moreover $\n F_\ep = -\psi_\ep^{-1} \n \psi$,  which is 
  in $L^2(m)$. 
   
    Suppose that $v:\R \to \R$ is $C^1$ and $v$ and $v'$ are bounded.
Let
\begin{align} 
w =v(F_\ep) e^{F_\ep} = v(F_\ep) \psi_\ep^{-1}.
\end{align}
Then $w$ is bounded and also
\begin{align}
\n w &= v'(F_\ep) e^{F_\ep} \n F_\ep +v e^{F_\ep} \n F_\ep \notag \\
&= (v' +v) e^{F_\ep} \n F_\ep,
\end{align}
which is in $L^2(m)$ because $(v'+v)$  and $e^{F_\ep}$ are bounded.  
 So $w\in Q(\n^*\n)\cap Q(|V|)$.
We may therefore compute the inner product of $w$ with both sides of the Schr\"odinger 
equation $H\psi = \l_0\psi$
as follows.
\begin{align*}
\l_0\int_X \psi w\, dm &= \int_X (H\psi) w\, dm \\
&=\int_X(\n^*\n \psi) w\, dm + \int_X V \psi w \, dm \\
&=\int_X(\n \psi)\cdot (\n w)\, dm + \int_X V\psi  w \, dm\\
&=\int_X (\n \psi)\cdot \((v' +v) \psi_\ep^{-1} \n F_\ep \) dm+ \int _X V\psi  w \, dm\\
&= -\int_X  (v' +v) \n F_\ep\cdot \n F_\ep dm + \int_X V\psi w \, dm.
\end{align*}
Therefore  
\begin{align}
\int_X  (v' +v) |\n F_\ep|^2  dm = \int_X (V-\l_0)v(F_\ep) (\psi/\psi_\ep) dm,    \label{W33}
\end{align}
where we have written $v$ for $v \circ F_\ep$.
Consider first the case $v \equiv 1$. Since 
$|\n F_\ep|^2 = \psi_\ep^{-2} |\n \psi|^2 \uparrow \psi^{-2}|\n \psi|^2 = |\n F|^2$ 
as $\ep \downarrow 0$, we may apply the monotone 
convergence theorem on the left and the dominated convergence theorem on the right to find
\begin{align}
\int_X |\n F|^2  dm = \int_X (V-\l_0) dm.
\end{align}
This proves \eref{W31}. Furthermore, since  $|(v' +v) |\n F_\ep|^2| \le \(\sup( |v + v')|\) | \n F|^2$, we can apply
the dominated convergence theorem to the left side of \eref{W33} (as well as on the right) to find \eref{W30a}.
  \end{proof}
\begin{remark} {\rm We will need different regularizations to carry out computations based 
on Aida's identity. The regularization of $F$ given by $\psi \to \psi+\ep$, which we used 
in the preceding proof was already used by Aida \cite{Aida2001}. 
}
\end{remark}
  
\begin{remark}\label{remWKB} {\rm (WKB equation).  
Aida's identity \eref{W30a} can be informally  derived, as in Remark \ref{remAWKB}, from the 
 WKB equation,
\begin{align}
\n^*\n F+ | \n F|^2 = V - \l_0, \ \ \      \text{ (WKB)}    \label{W7}
\end{align}
where $F$ is given by \eref{W6}.  
\eref{W7} itself follows from our form of the Schr\"odinger 
equation \eref{sch}
with the help of the product rule for the $m$ divergence operator $\n^*$, defined in \eref{div1}, namely,
 \begin{align}
\n^*(f \alpha ) &= f\n^* \alpha - \alpha \cdot \n f \ \ \ \    \text{$($product rule$)$},  \label{gs5}
\end{align}
where $\alpha$ is a vector field on $X$ and $f$ is a real valued function. The product rule follows
readily from the definition \eref{div1}.
To derive \eref{W7} observe that  \eref{W6} gives 
 $\n \psi = - \psi \n F$ and therefore    $\n^* \n \psi = - \n^* (\psi \n F) = \n \psi \cdot \n F - \psi \n^*\n F$.
It follows then from \eref{sch} 
that  $ (\l_0 - V)\psi =   \n \psi \cdot \n F - \psi \n^*\n F$.
Divide this equation by $\psi$ to find \eref{W7}.
}
\end{remark}

\begin{remark}\label{remAWKB} {\rm  (Aida's identity from WKB).
Proceeding informally, let $g$ be a ``differentiable'' real valued function on $X$.
Multiply \eref{W7}  
by $g$ and integrate over $X$ to find
\begin{align}
\int_X g(x)(V(x) - \l_0) dm(x) &= \int_X\( ( \n^*\n F)g + g|\n F|^2\) dm(x) \notag \\
&=  \int_X\( (\n F)\cdot (\n g) + g|\n F|^2\) dm(x).   \label{W26}
\end{align}
Let $g(x) = v(F(x))$. Then $\n g (x) = v'(F(x)) \n F(x)$. Insert this into \eref{W26} to find \eref{W30a}.
The integration by parts in \eref{W26} needs to be justified. This was done 
in our  proof of Theorem \ref{thmA3}. 
}
\end{remark}

\subsubsection{Examples} \label{secA2}

We will derive information about  the $L^p$ norms $\|\psi^{\pm 1}\|_p$ by first applying the 
logarithmic Sobolev inequality \eref{mt1} to compositions $w\circ F$, with $F = -\log \psi$. 
For this we will need bounds for integrals of the form $\int_X u(F(x)) |\n F|^2 dm$ because they 
 appear as the energy term in \eref{mt1}. Aida's identity, \eref{W30a}, allows us to express 
such an integral directly in terms of the potential $V$. To carry this procedure out it is necessary to find a function $v$ such that  
\beq
v'(s) + v(s) = u(s), \ \ \     s \in \R    \label{W41b}
\eeq
when $u$ is given, as we can see from Aida's identity.  
Depending on the choice of $u$, we  will derive  different 
entropy bounds in Section \ref{secEbound}
and then, via Herbst's method, norm bounds in Section \ref{secmombd}.

    In the following examples for Theorem \ref{thmA3} we ignore the previous boundedness 
    restrictions on $v$ and $v'$ because 
these restrictions can be removed once more information about $V$ is available. 

    \begin{example} \label{Ex1} {\rm Suppose that  $u(s) = e^{as}$ for some real $a\ne -1$. Then we 
may take $v(s) =(1+a)^{-1} e^{as}$ as a solution to \eref{W41b}. 
\eref{W30a} then shows that
\beq
\int_X e^{aF(x)}  |\n F|^2 dm(x) = (1+a)^{-1} \int_X e^{aF} (V -\l_0) dm.    \label{W35}
\eeq
This simple example, with $a+1 >0$,  underlies our main estimates. We will need
to truncate this function $v$ at first  to justify some technical steps.
}
\end{example}

\begin{example} \label{Ex10} 
{\rm        (A general class of examples). Suppose that $u:\R \to \R$ is a non-negative 
 continuous  function such that
\begin{align}
\int_{-\infty}^0 e^{r} u(r) dr <\infty.     \label{W39}
\end{align}
Define
\begin{align}
v(s) = e^{-s}\int_{-\infty}^s e^r u(r) dr,\ \ \ -\infty < s < \infty .    \label{W40}
\end{align}
Then
\beq
v'(s) + v(s) = u(s)                  \label{W41}
\eeq
and, by \eref{W30a}, 
\begin{align}
\int_X u(F) |\n F|^2 dm = \int_X v(F) (V -\l_0) dm.      \label{W42}
\end{align}
Of course in each application of this identity one must verify the integrability of both sides for the given 
functions $u$ and $v$.
}
\end{example}

\begin{example} \label{Ex3.6c} {\rm  $($Aida, \cite[Equ. (3.27)]{Aida2001}$)$. For any real number $a$ there holds
\begin{align}
\int_{F \ge a} |\n F|^2 dm \le \int_{F\ge a} (V - \l_0) dm       \qquad       (Aida)        \label{W43a}
\end{align}
and
\begin{align}
\int_{F > a} |\n F|^2 dm \le \int_{F> a} (V - \l_0) dm.       \qquad       (Aida)        \label{W43b}
\end{align} 
    \begin{proof} Choose $\ep >0$ and let $v$ be a smooth non-decreasing function 
on $\R$ which is zero on $(-\infty, a -\ep]$ and one on   $[a,\infty)$. Then
 \begin{align}
 \int_{F \ge a} |\n F|^2 dm &\le \int_X v(F(x))  |\n F|^2 dm \notag \\
 &\le \int_X (v + v')  |\n F|^2 dm  \notag \\
 &=\int_X v(F(x)) (V -\l_0)dm     \notag\\
 &=\int_{a-\ep < F < a} v(F(x))(V- \l_0) dm + \int_{F\ge a} (V - \l_0) dm  \label{W43b}
 \end{align}
 Since $|v| \le 1$ and $V$ is integrable we can let $\ep\downarrow 0$ and find that the first 
 term on the right of  \eref{W43b} goes to zero. This proves \eref{W43a}.
 
 Use \eref{W43a} and the dominated convergence theorem twice to find 
\begin{align*}
\int_{F> a}|\n F|^2 dm &=\lim_{n\to \infty} \int_{F\ge a+(1/n)}|\n F|^2 dm 
&\le \lim_{n\to \infty} \int_{F\ge a+(1/n)}(V-\l_0) dm \\
&=\int_{F >a} (V-\l_0)dm.
\end{align*} 
 \end{proof} 
}
\end{example}

Note: The set $\{ F= a\}$ could be a set of strictly positive measure. But, interestingly, one always has 
$\int_{F= a} |\n F|^2 dm =0 $. We will not need this fact.

          \subsection{Entropy bound from Aida's identity}  \label{secEbound}

The simple identity in Example \ref{Ex1} underlies the method of this section. 
We will combine variants of it with the logarithmic Sobolev inequality \eref{mt1} to derive bounds
on $L^p(m)$ norms of $1/\psi$.

To avoid technical problems we need to use first a bounded truncation 
of $F$, denoted $\hat F$ in the next theorem. 
 We will remove the truncation in Section \ref{secpfmp}. 
 
 If, for some real valued function $f$ on $X$, one puts $u = e^{f/2}$ into the 
 logarithmic Sobolev inequality \eref{mt1},
 one finds
 \begin{align}
Ent_m(e^f) \le (c/2)  \int_X |\n f|^2 e^f dm.     \label{W60}
\end{align}
This is actually equivalent to \eref{mt1}. It is 
convenient to use this form of the 
logarithmic Sobolev inequality.

        \begin{theorem} \label{thmeb2} $($Entropy bound$)$.   Assume that the  logarithmic 
        Sobolev inequality \eref{W60} holds for $m$.
       Assume that  the hypotheses of Theorem \ref{thmesa2} hold and also that $\|e^V\|_\ka < \infty$ for 
   some $\ka >0$. Denote by   $s_0$ and $r_0$  the roots of the quadratic equation 
   defined in Notation \ref{notkappa3}.
Let
\beq
\eta = \log \int_X e^{\kappa(V-\l_0)} dm.\ \ \ \ \ \ \ \ \ \ \    \label{W746}
\eeq
 Suppose that $\phi:\R \to \R$ is bounded, smooth and 
\beq
0 \le \phi' \le 1.              \label{W53}
\eeq
 Let
\beq
\hat F(x) = \phi (F(x)) .            \label{W54}
\eeq
Then
\begin{align}
Ent_m(e^{t\hat F} )&\le \frac{t^2}{(s_0-t)(t+r_0)}  \eta  E(e^{t\hat F})
  \ \  \  \text{if} \ t \in (-r_0, s_0). \label{W801f}  
\end{align}
\end{theorem}

Note: $\phi$ is intended to be  a bounded approximation to the identity function $\phi(s) = s$. 
It will later be taken to be 
a smooth approximation to $ \phi_n(s) := (-n)\vee(s\wedge n)$.

   \begin{lemma} \label{lemeb3}  If $\phi: \R \to \R$ is smooth, bounded 
and $0 \le \phi' \le 1$ then, for $\hat F  = \phi \circ F$, we have
\begin{align}
Ent_m(e^{t\hat F}) \le \frac{ct^2}{2(1+t)} \int_X e^{t\hat F} (V- \l_0) dm \ \  \text{if}\ \ 1+t >0.  \label{W59g}
\end{align} 
\end{lemma} 
    \begin{proof}
Insert $f(x) = t\phi(F(x))$ into the logarithmic Sobolev inequality \eref{W60} to find
\begin{align}
Ent_m(e^{t \hat F}) &\le (c/2)\int_X e^{t\hat F} |\n (t \phi\circ F)|^2 dm  \notag\\
&=(ct^2/2) \int_X  e^{t\hat F} \phi' (F(x))^2 |\n F|^2 dm.         \label{W73}
\end{align}
Let      
\beq
u(s) = e^{t\phi(s)} \phi'(s)^2\ \ \ \text{ and}  \ \ v(s) = (1+t)^{-1} e^{t\phi(s)}.  \label{W74}
\eeq
We will show that
\beq
u(s) \le v(s) + v'(s).      \label{W59}
\eeq 
Since $1+t >0$ and $(1+t)(v(s) + v'(s)) =e^{t\phi(s)}\( 1+ t \phi'(s)\)$,  we have
\begin{align*}
(1+t) u(s) &= e^{t\phi(s)}\((1+t)  \phi'(s)^2\) \\
& \le  e^{t\phi(s)}\((1+t)  \phi'(s)\) \\
& =     e^{t\phi(s)} \( \phi'(s) + t \phi'(s)\) \\
 & \le e^{t\phi(s)} \( 1+ t\phi'(s)\) \\
 &= (1+t) (v(s) + v'(s)).
\end{align*}
Divide by $1+t$ to find \eref{W59}.
From \eref{W59} and Aida's identity  \eref{W30a} we find 
\begin{align}
\int_X  e^{t\phi(F)} \phi'(F)^2 |\n F|^2 dm &=   \int_X u(F)  |\n F|^2 dm  \notag\\
&\le  \int_X(v(F) + v'(F))  |\n F|^2 dm        \notag\\
&=  \int_X v(F) (V- \l_0) dm \notag\\
&=\frac{1}{(1+t)} \int_X e^{t\phi(F)} (V- \l_0) dm.  \label{W59f}
\end{align}
Combine this with \eref{W73}, using $\hat F = \phi \circ F$,  to find \eref{W59g}.
\end{proof}

\bigskip
\noindent
\begin{proof}[Proof of Theorem \ref{thmeb2}] 
 From Young's inequality   \eref{BG500c} we have
\begin{align}
 \int_X e^{t\hat F} \kappa(V- \l_0) dm &\le 
 Ent_m(e^{t\hat F}) + \(\log \int_X e^{\kappa(V-\l_0)} dm \) E(e^{t\hat F})   \notag\\ 
  &=   Ent_m(e^{t\hat F}) + \eta  E(e^{t\hat F}).    \label{W59y}
\end{align}
Note that $t+1 >0$ if $t \in (-r_0, s_0)$ because $r_0 <1$, by \eref{W851g}.
From \eref{W59g}  and \eref{W59y} we find
 \begin{align}
Ent_m(e^{t\hat F})  &\le  \frac{ct^2}{2\kappa(1+t)}\( Ent_m(e^{t\hat F}) + \eta  E(e^{t\hat F}) \)  \notag
\end{align}
and therefore
\begin{align}
 \(1-  \frac{ct^2}{2\kappa(1+t)}\)Ent_m(e^{t\hat F}) 
                         \le   \frac{t^2}{2(\kappa/c)(1+t)}\eta  E(e^{t\hat F}). \label{W801a}
 \end{align}
 But
  \begin{align}
 \(1-  \frac{ct^2}{2\kappa(1+t)}\) &= 1 -\frac{t^2}{2(\kappa/c)(1+t)} 
  = \frac{2(\kappa/c)(1+t) - t^2}{2(\kappa/c)(1+t)}  .    \notag
\end{align}
Insert this into \eref{W801a} and cancel denominators to find
\begin{align}
 \( 2(\kappa/c)(1+t) - t^2\)Ent_m(e^{t\hat F}) \le  t^2\eta  E(e^{t\hat F}). \label{W801b}
 \end{align}
 The coefficient of  $Ent_m(e^{t\hat F}) $  factorizes by \eref{W851f} into  $(s_0 - t)(t+r_0)$, 
 which is strictly positive for  $ t \in (-r_0, s_0)$. We may therefore divide by it to find \eref{W801f}.
 \end{proof}

\subsection{Moment bound from entropy: Herbst's method}    \label{secmombd}

Herbst's method for deriving bounds on the moments $E(\psi^{-s})$ consists 
first in deriving bounds on the ratios  $Ent_m(\psi^{-s})/ E(\psi^{-s})$, which we 
have already done in Section \ref{secEbound}
for the truncated versions of $\psi^{-s}$. (Recall $\psi^{-s} = e^{sF}$.) 
 Second, one expresses the derivative $(d/ds)\(s^{-1}\log E(\psi^{-s})\)$ in terms of this ratio
and then integrates 
 the resulting differential inequality.
     
             In the many applications,  
     \cite{Car79,AMS94,ASt94, AS94, BG99,DS84,Hino1997,Led1995,Rot8,GR98,Ust92,Ust93,Ust96,vanHandel}, 
 of this method, however, 
 one needs information about the initial condition at $s=0$ in order to derive information 
 at time $s$  from the differential inequality.   In our setting  
 this initial condition takes the form of an assumption on $E(\log \psi)$, which we cannot use
 because  the only size hypothesis available to us is the normalization 
 condition $E(\psi^2) = 1$. Instead, we will  continue the differential inequality  
 through the apparent singularity at $s=0$ and use for initial condition
 the value of $E(\psi^r)$ for some $r >0$. We will thereby derive an upper bound  
 for $\int_X \psi^{-s} dm$ in terms of a lower bound  for  
$\int_X \psi^r dm$.  The lower bound has already been derived in Section \ref{seculb}.
Further discussion of the impracticality of using the initial condition at $s =0$ is given 
in Remark \ref{reminit2}.

The next lemma carries out Herbst's method in the form we need for passing through 
the apparent singularity. 
We abstract this step in  Herbst's method by
replacing the truncated function $\hat F$ by a general bounded measurable function $g$.
Various forms of the identity \eref{H2} figure in many of the 
applications of Herbst's method  in the papers listed above.

\begin{lemma}\label{lemH2} 
 Let $g:X\to R$ be a  bounded measurable function.  Define 
 \begin{align}
 \|e^{g}\|_t = E(e^{tg})^{1/t} \ \ \text{for}\ \ \ t\ne0, \ \ \ t \in \R. \label{H1}
 \end{align}
  Then 
 \begin{align}
(d/dt)\log \|e^g\|_t  =\frac{Ent_m(e^{tg})}{t^2 E(e^{tg})}, \ \ \ t \ne0. \label{H2}
\end{align}
Moreover
\begin{align}
\lim_{t\downarrow 0} \log \|e^g\|_t  &=\int_X g\, dm =\lim_{t\uparrow 0} \log \|e^g\|_t .  \label{H3}
\end{align}
The singularity on the right hand side of \eref{H2}  at $t=0$  is removable in the sense that 
the right side  extends  to a continuous function on $\R$. 

Suppose that $\beta$ is a 
continuous function on an interval $(-r_0, s_0)$ including $0$ such that
\begin{align}
\frac{Ent_m(e^{tg})}{t^2 E(e^{tg})}  \le \beta(t),\ \ \ 0 \ne t \in (-r_0, s_0).    \label{H4}
\end{align}
Then
\begin{align}
\|e^{-g}\|_r \|e^g\|_s \le e^{\int_{-r}^s \beta(t) dt}\ \ \ \text{when}\ 0 <r < r_0\ \text{and}\ 0 < s < s_0. \label{H5}
\end{align}
\end{lemma}
        \begin{proof}  Let $w(t) = E(e^{tg})$. If $t \ne0$ then
 \begin{align}
 w'(t) &=  E(e^{tg} g) = (1/t)E(e^{tg} \log e^{tg})  \notag\\
 &=(1/t) \(Ent_m(e^{tg}) + w(t) \log w(t)\).         \label{H6}
 \end{align} 
 Therefore
 \begin{align*}
 &(d/dt) \log \|e^g\|_t = (d/dt) \((1/t) \log (w(t)\) \\
 &=(1/t)w^{-1}w'  -(1/t^2)\log w(t)   \\
 & = (1/t^2) w(t)^{-1}   \(Ent_m(e^{tg}) + w(t) \log w(t)\)      -(1/t^2)\log w(t) \\
 &= (1/t^2)w(t)^{-1} Ent_m(e^{tg}).
 \end{align*}
 This proves \eref{H2}.
 
      If $t >0$ then $\|e^g\|_t = E((e^g)^t)^{1/t} \to \exp(\int_X g\, dm)$  as $t\downarrow 0$ by   
      \cite[Page 71, Problem 5]{Ru}.  
       This proves the first equality in \eref{H3}.
      If $t<0$ let $s = -t$. Then $\|e^g\|_t =  E(((e^{-g})^s)^{-1/s} \to \(\exp \int (-g) dm\)^{-1} = \exp \int g\, dm$
       as   $s \downarrow 0$.  This proves the second equality. 
       
       Concerning the removability of the  singularity in \eref{H2} observe that for small $t$ we have
\begin{align}
&Ent_m(e^{tg}) =E(e^{tg} tg) - E(e^{tg}) \log E(e^{tg}) \notag \\
&=   E( tg + t^2g^2 +o(t^2))      \notag\\
&\ \ \ \ \ \   -  \(1 + E(tg) +O(t^2)\)\log\(1+E(tg) +(t^2/2)E(g^2)  +o(t^2)\)  \notag \\
&  = tE(g) + t^2E(g^2) +o(t^2)                 \notag\\
&\ \ \ \ \ \  -  \(1 + E(tg) +O(t^2)\)\(E(tg) +(t^2/2)E(g^2) -E(tg)^2/2 +o(t^2)\) \notag \\
&= t^2E(g^2) +o(t^2)   -  \((t^2/2)E(g^2) -E(tg)^2/2 +o(t^2)\)  \notag\\
&\ \ \ \ \ \  - \( E(tg) +O(t^2)\)\(E(tg) +(t^2/2)E(g^2) -E(tg)^2/2 +o(t^2)\)          \notag\\
&  = (t^2/2) E(g^2)    -(t^2/2) E(g)^2      +o(t^2).         \notag 
\end{align}
Divide by $t^2$ to see that the right hand side of \eref{H2} has a common limit from the left and the right.

Now suppose that \eref{H4} holds. Writing $u(t) = \log E(e^{tg})^{1/t}$ for $t \ne 0$ we find from \eref{H2} 
and \eref{H4} that $du/dt \le \beta(t)$ for $t \ne 0$. Therefore, taking into account \eref{H3}, we have  
$u(s) - u(-r) = u(s) - u(0_+) +\(u(0_-) - u(-r)\) \le  \int_0^s\beta(t) dt + \int_{-r}^0 \beta(t) dt 
= \int_{-r}^s \beta(t)dt$.
Take the exponential of this inequality  to find
\begin{align}
e^{u(s)} e^{-u(-r)} \le \exp  \int_{-r}^s \beta(t)dt.    \notag
\end{align}
Since $e^{u(s)} = \|e^{g}\|_s$ and $e^{-u(-r)} = \|e^{-g}\|_r$ the inequality \eref{H5} is proved.
\end{proof}

\begin{remark} \label{reminit} {\rm
The proof of Lemma \ref{lemH2} shows that one can bound $\|e^g\|_s$ and $\|e^{-g}\|_r$  
 separately:
One has $u(s) - u(0_+) \le \int_0^s \beta(t) dt$, which gives 
\beq
\|e^g\|_s \le e^{E(g)} e^{\int_0^s \beta(t) dt}, \ \ \ 0 < s < s_0.   \label{H10}
\eeq
Similarly $u(0_-) - u(-r) \le \int_{-r}^0 \beta(t) dt$, which gives 
$e^{E(g)} \|e^{-g}\|_r \le  e^{\int_{-r}^0 \beta(t) dt}$.  \eref{H5} follows by multiplying these two inequalities
and canceling $e^{E(g)}$.

        It is the inequality \eref{H10}, with untruncated $g$, which is usually used in the application 
        of Herbst's method. 
         Information  about $E(g)$ is available  in these applications 
 and is sometimes taken as a hypothesis.  But in our application $g$ is a 
  truncated version of $\log\psi$. We have no useful information about $E(\log \psi)$.
  The usefulness of the product inequality \eref{H5} relies on the fact that $E(g)$ does not appear.
  See Remark \ref{reminit2} for further discussion of our case.
}
\end{remark}

\begin{remark} \label{remHerbsthist}  {\rm In many of the classical applications of the inequality \eref{H2}
one assumes that  a logarithmic Sobolev inequality, such as \eref{mt1}, holds and 
that $C\equiv \sup_X |\n g| < \infty$. In this case one has the simple entropy bound 
\begin{align*}
  Ent_m(e^{tg}) &\le 2c \int |\n e^{tg/2}|^2 dm \\
  &=2c(t/2)^2 \int e^{tg} |\n g|^2 dm \\
  & \le (ct^2/2)     E(e^{tg}) C^2,
  \end{align*}
  from which it follows that one can take $\beta(t) = cC^2/2$ for all $t \ne0$ in \eref{H4}.
 By \eref{H10} we have then
 \begin{align}
 E(e^{sg}) \le e^{sE(g)} e^{s^2cC^2/2}\ \ \text{for}\  s >0\ \text{and therefore for all}\ \ s \in \R.
 \end{align}
 One can remove the boundedness assumption on $g$ in this inequality while maintaining the bound
 $|\n g| \le C$ on its gradient. Such knowledge of the Laplace transform of $g$ can be used to deduce
 other bounds on functions of $g$. See for example \cite[page 100]{AMS94}. 
      The identity \eref{H6}, which is equivalent to \eref{H2}, was the form originally used by 
      Herbst (cf. \cite[Corollary 3.4]{AMS94}).
 It works well for the case $g(x) = a x^2$ on $\R$. \eref{H2} was already used by 
 van Handel \cite{vanHandel} in the study of sub Gaussian measures.

 For our application of \eref{H2} we will need to use  the entropy bound \eref{W801f}, which produces 
 the choice of $\beta(t)$ given in \eref{W756} and which is singular at the endpoints of the interval $(-r_0, s_0)$.
}
\end{remark}

\subsection{Proof of the moment product theorem} \label{secpfmp}

The next lemma  proves Theorem \ref{thmmp1} for a truncated version of $\psi$.
As  before, we write $\psi = e^{-F}$.

\begin{lemma}\label{lemptm}   $($Product of truncated moments$)$.
Assume the hypotheses and notation of Theorem \ref{thmmp1}. 
 Denote  by $\hat F$ the truncated function defined in \eref{W54}. Then
 \begin{align}
\|e^{-\hat F}\|_r \| e^{\hat F}\|_s \le   \| e^{V-\l_0}\|_\kappa^{\ell(a) + \ell(\sigma)}  .  \label{W753ga}
\end{align}
\end{lemma}
      \begin{proof} 
Choose $g = \hat F$ in Lemma \ref{lemH2}. By \eref{W801f} we have
\begin{align}
\frac{Ent_m(e^{t\hat F})}{t^2 E(e^{t\hat F})} \le \beta(t),\ \   0 \ne t \in (-r_0, s_0),    \label{W755}
\end{align}
where
\begin{align}
\beta(t) =  \frac{\eta}{(s_0 -t)(t+r_0)}, \ \ \ t \in (-r_0, s_0).     \label{W756}
\end{align}
It follows from \eref{H5} that 
\begin{align}
\|e^{-\hat F}\|_r \| e^{\hat F}\|_s \le e^{\int_{-r}^s \beta(t) dt}  \ \ 
                              \text{when}\ 0 <r < r_0\ \text{and}\ 0 < s < s_0.        \label{W550c}
\end{align}
It remains only to compute that the right side of  \eref{W550c} is equal to the right side of \eref{W753ga}.
We isolate this computation in the following sublemma.
\end{proof}

\begin{sublemma}  \label{sublem5}
\begin{align}
\exp\(\int_{-r}^s  \frac{\eta}{(s_0 -t)(t+r_0)} dt\) =  \| e^{V-\l_0}\|_\kappa^{\ell(a) + \ell(\sigma)} 
\end{align}
\end{sublemma}
    \begin{proof}
From the definition \eref{W746} of $\eta$ we have
$e^{\eta y} = \|e^{V - \l_0}\|_\ka^{\ka y}$
for any real number $y$. Thus we need to show that
\begin{align}
\int_{-r}^s \frac{\ka}{(s_0-t)(t+r_0)} dt = \ell(a) + \ell(\sigma).    \label{W853a}
\end{align}
From \eref{s98f} and \eref{W851g} we see that
 $s_0+ r_0= 2(\ka/c) b_\ka$.
 Therefore
\begin{align} 
\int_{-r}^s \frac{\ka}{(s_0-t)(t+r_0)} dt  
&=  \ka (s_0+ r_0)^{-1} \int_{-r}^s \(\frac{1}{(s_0-t)} + \frac{1}{(t+r_0)} \)dt     \notag\\
&=    (c/2b_\ka)\log\frac{t+r_0}{s_0-t} \big|_{-r}^s .              \label{W853b} 
\end{align}
We want to rewrite this  in terms of the quantities $a$ and $\sigma$ defined in \eref{W852b} because
they will appear explicitly in the defective logarithmic Sobolev inequality \eref{gs805} - \eref{gs807b}. 
To this end we have, using \eref{W851k},
\begin{align}
\log\frac{t+r_0}{s_0-t} \Big|_{-r}^s  & = \log \(\frac{ r_0^{-1} +t^{-1} }{t^{-1} -s_0^{-1}  }\ \frac{r_0}{s_0}\)
  \Big|_{-r}^s                \notag \\
&= \log \(\frac{ r_0^{-1} +t^{-1} }{t^{-1} -s_0^{-1}  } \) \Big|_{-r}^s     \notag \\
&= \log \(\frac{ (2r_0^{-1} -1)  +(2t^{-1} +1)}{(2t^{-1} +1) -(2s_0^{-1} +1)  } \) \Big|_{-r}^s  \notag \\
&=\log \(\frac{ b_\ka  +(2t^{-1} +1)}{(2t^{-1} +1) -b_\ka  } \) \Big|_{-r}^s     \notag \\
&= \log \(\frac{ b_\ka  +(2s^{-1} +1)}{(2s^{-1} +1) -b_\ka  } \) 
     - \log \(\frac{ b_\ka  +(-2r^{-1} +1)}{(-2r^{-1} +1) -b_\ka  } \)     \notag \\
&= \log \(\frac{ b_\ka  +(2s^{-1} +1)}{(2s^{-1} +1) -b_\ka  } \)  
+  \log \(\frac{ (2r^{-1} -1) + b_\ka}{-b_\ka +(2r^{-1} -1) } \)  \notag  \\
& = \log\frac{a + b_\ka c_\nu}{a - b_\ka c_\nu} + \log\frac{\sigma + b_\ka c_\nu}{\sigma - b_\ka c_\nu}. 
                                                              \label{W853c}
\end{align}
This, together with  \eref{W853b} and the definition \eref{W852}, proves \eref{W853a}.
\end{proof}

\bigskip
\noindent
\begin{proof}[Proof of Theorem \ref{thmmp1}] 
 We will  choose a sequence $\phi_n$ of functions, each of which satisfies
the conditions in Theorem \ref{thmeb2} for $\phi$ and such that $\phi_n(s)$ converges
to $s$ in a suitable sense.
Taking $\hat F = \phi_n\circ F$ in \eref{W753ga}, we will show that the limit yields \eref{W753g}.

For an integer $n \ge 1$ the function $\R \ni y \mapsto f_n(y) \equiv (-n)\vee(y \wedge n)$ is 
linear on $[-n,n]$ and  constant outside this interval. 
Choose a smooth nondecreasing function $\phi_n$ which agrees with $f_n$ outside 
the two intervals $\{|s - (\pm n)| < 1/2\}$, satisfies $0 \le \phi_n' \le 1$ and lies below $f_n$
for positive $y$ and above $f_n$ for negative $y$. Clearly such functions exist.
Then  $0 \le \phi_n(y) \uparrow y$ for $y \ge0$ and $0 \ge \phi_n(y)\downarrow y$ for $y \le 0$.
The  functions $ F_n(x) := \phi_n(F(x))$ then converge monotonically  
upward  on $\{x: F(x) \ge 0\}$ and downward on $\{ x: F(x) <0\}$.  For $s >0$
the sequence $\int_Xe^{s F_n} dm$ therefore converges to $\int_X e^{s F} dm$ by applying the 
monotone convergence theorem over the first set and the dominated convergence 
theorem over the second set. 
 Similarly,  for $r >0$ the sequence $\int_X e^{-r F_n} dm$ converges by applying these two 
 theorems to the opposite sets. 
 
     Choose $\phi$ in Theorem \ref{thmeb2}  to be $\phi_n$. The left side of \eref{W753ga}
     is then $\| e^{-F_n}\|_r \|e^{F_n}\|_s$, which converges to $\| e^{-F}\|_r \|e^{F}\|_s$ as $n\to \infty$.
 The right side of \eref{W753ga} 
 is independent of $n$ and the inequality therefore holds in the limit. 
 Since $\psi = e^{-F}$ and $\psi^{-1} = e^F$, \eref{W753g} follows. 
\end{proof}

\bigskip

\begin{remark}\label{reminit2}  {\rm Remark \ref{reminit}, together with the limiting
procedure of the previous proof, shows, informally, that
\begin{align}
\|\psi^{-1}\|_s &\le  \| e^{V-\l_0}\|_\ka^{\ell(a)} \exp{\int  F dm}, \ \ \ 0< s< s_0 \ \ \text{and}  \label{W860}\\
\|\psi\|_r  &\le  \| e^{V-\l_0}\|_\ka^{\ell(\sigma)}  \exp{-\int F dm},  \ \ \ 0<r < r_0   \label{W861}
\end{align}
On the one hand, the two exponential factors are finite because
\begin{align}
2\int (-F) dm &= \int \log \psi^2 dm \le  \int \psi^2 dm =1\ \ \ \text{and} \notag \\
s\int F dm &= \int \log \psi^{-s} dm \le \int \psi^{-s} dm < \infty
\end{align}
if $0 < s < s_0$, by Theorem \ref{thmmp1}.
But these inequalities are not useful for us because  we do not 
have good control over the size of the exponential factor in \eref{W860}. 
Some bounds on $\pm \int_X F dm$ are derived by Aida in \cite[Lemma 3.3, Part (4)]{Aida2001}.
 He requires only that a Poincar\'e inequality hold for $m$.
}
\end{remark}

\begin{remark}\label{remnonu} {\rm The bound in the moment product inequality \eref{W753g}
depends on  $\|e^V\|_\ka , \ka$ and $\l_0$, but
only uses the  condition $\|e^{-V}\|_\nu < \infty$ for the 
 purpose of showing $\n^*\n + V$ is essentially self-adjoint and that a unique 
 ground state exists. The boundary values 
 $r_0, s_0$ depend only on $\ka$. The inequality \eref{W753g} therefore holds without any specific
 assumption on $e^{-V}$ if the essential self-adjointness and existence of a unique ground state 
 can be shown by some other method. The equation \eref{W853c} shows that the exponent 
 of $\|e^{V-\l_0}\|_\ka$ in \eref{W753g} depends only on $c, \ka$ and on $r,s$ but not on $\nu$.
}
\end{remark}

\section{$L^p$ bounds on the inverse of the ground state} \label{seclp} 

\subsection{The controlling functional of $V$} \label{secM}

 The upper bound \eref{W753g} on the product of moments is 
 dominated by a power of $\|e^{V-\l_0}\|_\ka$ while a lower bound on $\|\psi\|_r$  
 is dominated by a power of   $\|e^{\l_0-V}\|_\nu$, as in \eref{L512}. The ground state eigenvalue
 appears in both sets of estimates. We will see that when combining these estimates so as to get
 a bound on  $\|\psi^{-1}\|_s$ it is possible to arrange these two factors in a product so that the eigenvalue
 $\l_0$ cancels. As a result the following functional of $V$ appears naturally
 in almost all of the estimates.

      \begin{notation} \label{notM} {\rm Let
 \begin{align}
 M = \|e^V\|_{L^\ka(m)} \| e^{-V}\|_{L^\nu(m)}.  \label{M1}
\end{align}
$M$ depends on $\ka, \nu$ and $V$.
$M$ has the following general properties for any $\ka >0, \nu >0$ and $a \in \R$.
\begin{align}
\|e^{V-a}\|_\kappa  \|e^{a - V}\|_\nu  &=M. \label{M2} \\
M &\ge 1. \label{M3}
\end{align}
\eref{M2} holds because the  constant factors $e^{-a}$ and $e^a$ cancel.
For the proof of \eref{M3} observe that for any $p >0$ we have
 \begin{align}
 1 = \( \int e^{pV/2} e^{-pV/2} dm\)^2 \le \int e^{pV} dm \int e^{-pV} dm. \notag
 \end{align}
 Therefore $\|e^V\|_p \|e^{-V}\|_p \ge 1$. Choose $ p = \min (\ka, \nu)$. 
 If, say, $p=\nu$ then we have
 $1 \le \|e^V\|_\nu \|e^{-V}\|_\nu  \le  \|e^V\|_\ka \|e^{-V}\|_\nu$ by H\"older's inequality. A similar argument 
 holds if $p =\ka$. This proves \eref{M3}. 
 }
\end{notation}

       \begin{lemma}  $($Upper and lower bounds on $\l_0)$. 
Assume that the logarithmic Sobolev inequality \eref{mt1} 
 holds. Then
\begin{align}
e^{-\l_0} &\le \|e^{-V}\|_\nu, \ \ \ \nu  \ge 2c  \ \ \ (\text{Federbush})\  \ \label{s1}   \\
e^{\l_0 }&\le \|e^V\|_\kappa , \ \ \ \ \ \kappa >0\ \ \ \ \  (\text{Aida}) \label{s2} \\
e^{t \l_0} &\le e^{t\int_X V dm}, \ \ \  t \ge 0     \label{s2a}        \\
\|e^{V-\l_0}\|_\kappa &\le  M,     \ \ \    \ \ \ \ \ \ \  \nu  \ge 2c          \label{s3} \\
\|e^{\l_0 - V}\|_\nu  & \le   M,    \ \ \  \  \ \ \ \ \ \ \kappa >0\           \label{s4} \\
\|e^{V-\l_0}\|_\kappa  \|e^{\l_0 - V}\|_\nu  &=M, \qquad\ \   \forall\  \ka >0, \nu > 0   \label{s5}
\end{align}
\end{lemma}
                \begin{proof} 
The Federbush semi-boundedness theorem, see Remark \ref{remfed5}, 
asserts that $e^{-\l_0} \le \|e^{-V}\|_{2c}$ because $ \l_0 = \inf \{(H_0 +V)u,u): \|u\|_2 =1\}$. 
\eref{s1} now follows from H\"older's inequality.
From \eref{W31} we find that   
$\l_0 = \int_X V dm -\int_X|\n F|^2 dm \le \int_X V dm$. Therefore
\begin{align}
\l_0 \le \int_X V dm,        \label{s6}
\end{align}
from which \eref{s2a} follows.
But also    $\kappa \l_0 \le \int \kappa V dm \le \log \int e^{\kappa V} dm$
     by Jensen's inequality. Hence $\l_0 \le \log \| e^V\|_\kappa$ from which \eref{s2} follows.   
  In view of \eref{s1} we have 
  $\|e^{V -\l_0}\|_\ka = \|e^V\|_\ka e^{-\l_0} \le    \|e^V\|_\ka \|e^{-V}\|_\nu = M$, giving \eref{s3}.
   \eref{s4} follows similarly from \eref{s2}. The identity \eref{s5} is a special case of \eref{M2}.
\end{proof}

\subsection{Upper bound on $\int \psi^{-s} dm$ for $s > 0$} \label{secub}

The following is a corollary of Theorem \ref{thmmp1}.
  \begin{corollary} \label{corub1}  
  $($Upper bound on $\|\psi^{-1}\|_s)$.  Assume the hypotheses and notation of Theorem \ref{thmmp1}. 
  Suppose that $\sigma > b_\ka c_\nu$. Then
\begin{align}
\|\psi^{-1}\|_{L^s(m)} 
\le \| e^{V-\l_0}\|_\kappa^{\ell(a) + \ell(\sigma)} \|e^{\l_0-V}\|_\nu^{\sigma}, \ \ \  0 < s < s_0,    \label{W710a}
\end{align}
where 
 $a$ is given by \eref{W852b}.
In particular
\begin{align}
\|\psi^{-1}\|_{L^s(m)} \le M^{\ell(a) + \ell(\sigma) +\sigma}   .    \label{W710b}
\end{align}
If $ 0< s <\min\{s_0, 2\}$ then
\begin{align}
\int_X e^{s|\log \psi|} dm < \infty \ \  \text{and}\ \ Ent_m(e^{s|\log \psi|}) < \infty .          \label{W712}
\end{align}
\end{corollary}
     \begin{proof} Given $\sigma >b_\ka c_\nu$, define $r$ by  \eref{W852b}. Combine  \eref{W753g}  and \eref{L512}  to find \eref{W710a}. 
Use \eref{s3} and \eref{s4} to derive \eref{W710b} from \eref{W710a}.
 
For the proof of \eref{W712}    observe that  for $0 < s < s_0$, 
\eref{W710a} implies that $\int e^{s(-\log \psi)}dm < \infty$.
On the other hand if $0 < s \le 2$ then
$\int e^{s\log\psi} dm = \int \psi^s dm \le \|\psi\|_2^{s/2} =1$. Since 
$\int e^{s|\log \psi|} dm \le \int e^{s\log \psi} dm + \int e^{s(-\log \psi)} dm $ the first 
    assertion in \eref{W712} follows. 
   The second assertion  follows by choosing a slightly larger $s$ in the first assertion. 
\end{proof}

\begin{remark} \label{remsigmas} {\rm The bound \eref{W710b} arises from bounding each of the two
factors in \eref{W710a} by $M$ to a power, using \eref{s3} and \eref{s4}. But there is a loss in using 
\eref{s3} and \eref{s4} separately instead of using the 
combined product, as in \eref{s5}, where possible. If, given
$s$, one chooses $\sigma$ suitably then the two powers on the right side of \eref{W710a} can be made
equal and  the $\l_0$ independent bound \eref{s5} can be used.
For the proof of existence of such a $\sigma$ observe that the definition \eref{W852} shows that 
$\ell$ is strictly decreasing on the interval $(b_\ka c_\nu, \infty)$ and $-\ell$ is strictly increasing with 
range $(-\infty, 0)$. 
  Consequently  the function   $\sigma \to \sigma - \ell(\sigma)$ is strictly increasing on 
  this interval and has range $(-\infty, \infty)$.  
  Given $s \in (0, s_0)$, there is therefore a unique number $\sigma_s$ in this interval such that
  $ \sigma_s - \ell(\sigma_s) = \ell(a)$. With this choice of $\sigma$ we then have 
   \begin{align}
  \| e^{V-\l_0}\|_\kappa^{\ell(a) + \ell(\sigma_s)} \|e^{\l_0-V}\|_\nu^{\sigma_s}  
  =\(\| e^{V-\l_0}\|_\kappa\|e^{\l_0-V}\|_\nu\)^{\sigma_s} = M^{\sigma_s}      \label{W710c}
  \end{align}
  and therefore 
  \begin{align}
\|\psi^{-1}\|_{L^s(m)}  \le M^{\sigma_s}.   \label{W710d}
\end{align} 
Although this is a sharper bound than \eref{W710b} when $a$ and $\sigma$ do not have to be specified,
it may be difficult in applications to control $\sigma_s$.
}
\end{remark}

 \subsection{$V$ is large where $\psi$ is small}   \label{secbig}

  Suppose that $\ka >0$ and $s_0$ is defined as in Notation \ref{notkappa3}. 
 Theorem \ref{thmmp1} shows that if $\|e^V\|_\ka < \infty$ then $\|\psi^{-1}\|_s <\infty$ for all $s < s_0$.
  Contrapositively, if  $\|\psi^{-1}\|_s =\infty$ for some  $s < s_0$ then $\|e^V\|_\ka =\infty$.    
  We will show that a stronger contrapositve holds. Namely, if $\|\psi^{-1}\|_s =\infty$ 
  for some  $s < s_0$ then  $\int_{\psi < \delta} e^{\ka V} dm = \infty$ for all $\delta >0$.
   This is a quantitative version of the statement  that $V$ is large where $\psi$  is small. 
   A qualitative version, such as  
   ``$V$ is unbounded where $\psi^{-1}$ is unbounded",   does not hold
   in our context, nor  for  a Schr\"odinger operator $-\Delta + V$ acting in $L^2(\R^n, dx)$. 
   The latter is well known.
   We will describe in Example \ref{exunb}, a bounded potential in our context
   for which $\psi$ and $\psi^{-1}$ are both unbounded.

          The proof of the strong contrapositive inequality is a consequence of 
          the following local moment product theorem.

\begin{theorem} \label{thmlocmp2} $($A local moment product theorem$)$.
Suppose that the hypotheses of Theorem \ref{thmesa2} hold. Let $\ka >0$ and define $r_0$ and $s_0$
as in Notation \ref{notkappa3}.
Let $\delta >0$. Define
\begin{align}
\psi_\delta(x) = \min(\psi(x), \delta).     \label{big40}
\end{align}
If $0 < r < r_0$ and $0<s < s_0$ then
\begin{align}
\|\psi_\delta\|_r \| \psi_\delta^{-1}\|_s 
          \le \| e^{(V-\l_0)\chi_{\psi \le \delta}}\|_\ka^{\ell(a) + \ell(\sigma)}. \label{big42}
\end{align}
\end{theorem}

The proof depends on the following lemma, which  is a small variant of  Lemma \ref{lemeb3}. 
           \begin{lemma} \label{lemeb5}  Let $\phi: \R \to \R$ be a smooth bounded function 
           which is zero on $(-\infty, b]$ and 
 such that $0 \le \phi' \le 1$ everywhere.
Let $F = - \log \psi$ and define $\hat F =\phi\circ F$. 
Then 
 \begin{align}
 Ent_m(e^{t\hat F})\le\frac{ct^2}{2(1+t)} \int_{F \ge b} e^{t\hat F}  (V- \l_0) dm.  \label{W59k}
\end{align}
 Define $r_0$ and $s_0$ as in Notation \ref{notkappa3} 
 and let 
\begin{align}
\eta_b = \log \int_X e^{\kappa  (V- \l_0)\chi_{F\ge b}} dm.    \label{big54}
\end{align}
Then  
\begin{align}
Ent_m(e^{t\hat F} )&\le \frac{t^2}{(s_0-t)(t+r_0)}  \eta_b  E(e^{t\hat F})
  \ \  \  \text{if} \ t \in (-r_0, s_0), \label{W801g}  
\end{align}
\end{lemma}
            \begin{proof}  Let $h$ be a non-decreasing smooth function on $\R$ which is $0$ on 
$ (-\infty, b-\ep)$ and $1$ on $[b, \infty)$. Define  
\beq
u(s) = e^{t\phi(s)} \phi'(s)^2\ \ \ \text{ and}  \ \ v_1(s) = (1+t)^{-1} e^{t\phi(s)}h(s).   \label{W74b}
\eeq
Then 
\begin{align}
u(s) \le v_1(s) + v_1'(s)       \label{W74c}
\end{align}
 because on $[b, \infty)$,  $u$ and $v_1$ are equal to the functions 
  $u$ and $v$, respectively, given in \eref{W74}. Therefore \eref{W74c} holds over this 
  interval by virtue of \eref{W59}.  
  For $s\le b-\ep$ both sides of \eref{W74c} are zero, while for $b-\ep < s < b$, $u(s)$ is zero while
 $(1+t) (v_1(s) + v_1'(s)) = (1+t)e^{t\phi(s)}(h(s) + h'(s)) \ge 0$.
 So \eref{W74c} holds everywhere.
  
As in the derivation of \eref{W59f},  we then find 
\begin{align}
\int_X  e^{t\phi(F)} \phi'(F)^2 |\n F|^2 dm &\le  \int_X v_1(F) (V- \l_0) dm  \notag\\
&=\frac{1}{1+t}\int_X e^{t\phi(F)} h(F) (V-\l_0) dm.          \notag
\end{align}
From \eref{W73} it follows that  $Ent_m(e^{t\hat F})\le\frac{ct^2}{2(1+t)} \int_X e^{t\hat F} h(F) (V- \l_0) dm$.
Since $e^{t\hat F}$ is bounded and $V - \l_0$ is integrable we can let $\ep \downarrow 0$ and conclude
from  the dominated convergence theorem  that 
\begin{align}
Ent_m(e^{t\hat F})\le\frac{ct^2}{2(1+t)} \int_X e^{t\hat F} (V- \l_0)  \chi_{F \ge b} dm, \label{big55}
\end{align}
which is \eref{W59k}.
      The proof of \eref{W801g} follows from \eref{big55} 
      the same way that \eref{W801f}  follows from \eref{W59g}. 
      One need only replace $\eta$ by $\eta_b$ in \eref{W59y}.
\end{proof}

\bigskip
\noindent
\begin{proof}[Proof of Theorem \ref{thmlocmp2}] (9/30/22, 10/4/22, 7/23/23) 
Let
\begin{align} 
\beta_b(t)  =  \frac{\eta_b}{(s_0 -t)(t+r_0)}, \ \ \ t \in (-r_0, s_0).
\end{align}
Then, from \eref{W801g}, we find
\begin{align}
\frac{Ent(e^{t \hat F})}{t^2 E(e^{t\hat F})} \le \beta_b(t)
\end{align}
 As in the derivation of \eref{W550c} it follows that 
$\| e^{-\hat F}\|_r \| e^{\hat F}\|_s \le e^{\int_{-r}^s \beta_b(t) dt }$.
Since $e^{\eta_b} = \|e^{(V-\l_0)\chi_{F \ge b}}\|_\ka^\ka$,
we find, by Sublemma \ref{sublem5} 
\begin{align}
\| e^{-\hat F}\|_r \| e^{\hat F}\|_s \le \| e^{(V - \l_0)\chi_{F \ge b}}\|_\ka^{\ell(a) + \ell(\sigma) }.  \label{big58}
\end{align}
Let $g(y) = 0\vee (y-b)$ for $ y \in \R$. Choose a sequence $\phi_n$ of smooth  functions, each of which 
is bounded, such that $\phi_n = 0$ on $(-\infty, b]$, $0 \le \phi_n' \le 1$, and such 
that $\phi_n(y) \uparrow g(y)$. 
Such a sequence is easily seen to exist. In Lemma \ref{lemeb5} choose $\phi = \phi_n$. Let
$F_b = g \circ F$. Then $\phi_n \circ F \uparrow F_b$ on $X$. Replacing $\hat F$ by $\phi_n\circ F$ in
\eref{big58} we can apply the dominated convergence theorem on the first factor on the left and the 
monotone convergence theorem on the second factor to find
\begin{align}
\|e^{-F_b}\|_r \|e^{F_b}\|_s \le  \|e^{(V- \l_0)\chi_{F\ge b}}\|_\ka^{\ell(a) + \ell(\sigma)}. \label{big44}
\end{align}
Given $\delta >0$ choose $b$ so that $e^{-b} = \delta$. We claim that
\begin{align}
e^{-F_b(x)} = \delta^{-1}\psi_\delta(x)\ \  \forall x \in X.    \label{big60}
\end{align}
Indeed $F_b(x) = 0 \vee(F(x) - b)$. So if $F(x) < b$ then $F_b(x) = 0$.  
      So $e^{-F_b(x)} = 1$. But $\psi(x) =e^{-F(x)} >  e^{-b} = \delta$.
 So \eref{big60} holds by the definition \eref{big40}. 
On the other hand, if $F(x) \ge b$ then 
     $\psi(x) = e^{-F(x)}\le e^{-b}  = \delta$. 
     So $\psi_\delta(x) = \psi(x)=  e^{-F(x)} = e^{-F_b(x)}e^{-b} = \delta e^{-F_b(x)}$. 
     This proves \eref{big60}.  Moreover $\{\psi \le \delta\} = \{F\ge b\}$.
     
     Therefore we may write \eref{big44} as $\| \delta^{-1} \psi_\delta\|_r \|\delta \psi_\delta^{-1}\|_s \le
       \| e^{(V- \l_0)\chi_{\psi \le \delta} }\|_\ka^{\ell(a) + \ell(\sigma)}$, 
       which is \eref{big42} after canceling $\delta$.     
\end{proof}

\begin{corollary} \label{corvlps2}  $(V$ is large where $\psi$ is small$)$.
 Given $\kappa >0$, define $s_0$ as in Notation \ref{notkappa3}. Suppose that $0<s < s_0$.
 If
\begin{align}
\int_X \psi^{-s} dm = \infty              \label{big71}   
\end{align}
then 
\begin{align}
\int_{\psi \le\delta} e^{\ka V} dm  = \infty\ \ \ \text{for all} \ \   \delta >0.   \label{big72}  
\end{align}
\end{corollary}
   \begin{proof} 
Let $\delta >0$.
 Choose a number $r \in (0, r_0)$. Then
$\|\psi_\delta\|_r > 0$ because $\psi_\delta > 0$ a.e..  From \eref{W852} we see that   $\ell(t) >0$ 
for all allowed $t$ and therefore $\ell(a) + \ell(\sigma) >0$. 
Since $\psi_\delta^{-1} - \psi^{-1}$ is bounded, it follows from \eref{big71} that  $\|\psi_\delta^{-1}\|_s = \infty$.
The local moment product formula \eref{big42}  shows then that  
$ \|e^{(V - \l_0)\chi_{\psi \le\delta}}\|_\ka^\ka =\infty$. 
That is, 
$\int_{\psi > \delta} 1 dm + \int_{\psi \le \delta} e^{\ka(V-\l_0)} dm = \infty$.  \eref{big72} follows.
\end{proof}   

\begin{remark} {\rm   If one takes $s$ as given in the condition \eref{big71} then the condition on $\ka$
that ensures ``largeness'' in the sense of \eref{big72} is 
\beq
\ka/c > \frac{ s^2}{ 2(s+1)} .       \label{big70}  
\eeq
Indeed $\ka/c = s_0^2/2(s_0+1)$ by \eref{W850}. The condition \eref{big70} is therefore
equivalent to $s < s_0$ for $s >0$ because the right side of \eref{big70} is increasing.

}
\end{remark}

\bigskip
 Even if the potential is bounded, neither the ground 
state $\psi$ nor its inverse $1/\psi$ 
  need be bounded, even in the presence of \eref{mt1}.
 Here is a simple example.

\begin{example}\label{exunb} {\rm (Bounded $V$ but unbounded $\psi$ and $\psi^{-1}$).
Take $m$ to be the Gauss measure $dm= (2\pi c)^{-1/2} e^{-x^2/(2c)} dx$. It is known that $m$ satisfies the logarithmic  Sobolev inequality   \eref{mt1}, \cite{G1}. Let
\begin{align}
\psi(x) = Z^{-1}
\begin{cases} &(1+x^2), \ \ \ x >1 \\
     & (1+x^2)^{-1}, \ \ \  x < -1 \\
     & \text{smooth and $>0$ on} \ \  \ [-2,2]
\end{cases}
\end{align}
Let $F = - \log \psi$. Then outside the interval $[-1,1]$ we have 
$F(x) = -\delta \log (1+x^2)$, where $\delta = sgn\ x$. Therefore, for $|x| >1$ we find
$F'(x) =-2\delta x/(1+x^2)$ and $F''(x) = -2\delta \(\frac{1-x^2}{1+x^2}\)/(1+x^2)$.
The definition \eref{div1} shows that for our measure $m$ we have  
$\n^*v(x) = -v'(x) +  c^{-1}x v(x)$ for any smooth vector field $v$ on $\R$.
 Therefore $|\n F|^2 + \n^*\n F = (F')^2 - F'' + c^{-1} xF'$.
We  take this to be our potential. Explicitly, we have then, for $|x| >1$,
\begin{align}
V=  4 x^2/(1+x^2)^2  +2\delta \(\frac{1-x^2}{1+x^2}\)/(1+x^2)   -2\delta c^{-1}x^2/(1+x^2).
\end{align}
By the WKB equation \eref{W7} the ground state for $\n^*\n +V$ is $\psi$.
$V$ is bounded on $\R$ but $\psi$ and $\psi^{-1}$ are both unbounded. Theorem \ref{thmM}
is applicable to this example and therefore the ground state measure $\psi^2 dm$ satisfies
a logarithmic Sobolev inequality.
}
\end{example}

\section{Defective LSI   for the ground state measure} \label{secDLSI}

\subsection{The ground state transformation} \label{secgs}

 In the previous sections we established properties of the Schr\"odinger 
operator $\n^*\n +V$  and its heat semigroup
 in the spaces $L^p(m)$. We also established properties of the ground
state $\psi$ and its inverse $1/\psi$ in the spaces $L^p(m)$.  
The ground state measure associated to $\psi$ is the 
probability measure $m_\psi$ defined by
\begin{align}
dm_\psi : = \psi^2 dm.
\end{align}
In the present section we will relate the Schr\"odinger operator $\n^*\n + V$ to the Dirichlet form operator
$\hat H$ for $m_\psi$.
$\hat H$  acts densely in $L^2(m_\psi)$.  Define
\begin{align}
U: L^2(m_\psi) &\to L^2(m)\ \ \ \ \ \text{by} \\
Uu &= u \psi, \ \ \ \ u \in L^2(m_\psi)
\end{align}
The identity $\int_X |u\psi|^2 dm = \int_X |u|^2 dm_\psi$ shows that the map  $U$ is unitary.

 Denote  by $H$ the closure of $\n^*\n +V$ in $L^2(m)$. 
 In the next lemma we will make 
a computation, frequently made in this context, which shows that $U^{-1}(H - \l_0)U = \hat H$,
and which at the same time exhibits the quantities which need to be estimated 
for proving invariance of intrinsic hypercontractivity.
This computation is sketched in \cite[Section 4]{G1}, derived and used in \cite{KS85} 
and derived again in many similar contexts. We make no effort to identify the domains of 
operators in this partly informal computation or to justify some of the technical steps 
because in the cases of interest the final identities will be easily justifiable.

\begin{lemma} \label{lemgs1}  Let $\psi =e^{-F}$ be a strictly positive $($a.e.$)$ function 
with $\|\psi\|_{L^2(m)} = 1$. 
The adjoint of $\n$ with respect to the measure $m$, defined in \eref{div1}, is denoted $\n^*$.
If $u$ is bounded and $|\n u|$ is  in $L^2(m_\psi)$ then
\begin{align}
\int_X |\n (u\psi)|^2 dm =\int_X|\n u|^2\, dm_\psi -\int_X u^2(\n^*\n F + |\n F|^2) dm_\psi.          \label{gs709F}
\end{align}
In particular, if $\psi$ is the ground state of $\n^*\n+V$ then
\begin{align}
\int_X |\n (u\psi)|^2 dm =\int_X|\n u|^2\, dm_\psi + \int_X u^2(\l_0 -V) dm_\psi          \label{gs709G}
\end{align}
and
\begin{align}
U^{-1}(H- \l_0) U = \hat H        \label{gs709H}
\end{align}
\end{lemma}
       \begin{proof} Since $\n\psi = -\psi \n F$ the product rule gives  
       $\n (u\psi) = (\n u)\psi + u \n \psi =  (\n u - u\n F)\psi$.
    The product rule \eref{gs5} for $\n^*$ implies 
 $\n^* (e^{-2F} \n F)= (\n^*\n F +2 |\n F|^2)e^{-2F}$. Hence
 \begin{align}
\int_X |\n (u\psi)|^2 dm &=\int_X |\n u - u\n F) |^2 \psi^2 dm        \notag\\
&   = \int_X \(|\n u|^2 + u^2 |\n F|^2 - 2 u\n u \cdot \n F\)\psi^2 dm                \label{gs1000} \\
&= \int_X\(|\n u|^2 + u^2 |\n F|^2) dm_\psi  -\int_X \n u^2 \cdot e^{-2F} \n F dm        \notag \\
&= \int_X\(|\n u|^2 + u^2 |\n F|^2) dm_\psi -\int_X  u^2 \n^*\( e^{-2F} \n F\) dm  \notag\\
& =  \int_X\(|\n u|^2 + u^2 |\n F|^2) dm_\psi  -\int_X u^2\( \n^* \n F + 2 |\n F|^2\) dm_\psi, \notag
 \end{align}   
 which proves \eref{gs709F}. 
     If $\psi$ is the ground state of $\n^*\n +V$ then \eref{gs709G} follows from \eref{W7}. 
    
     The left side of \eref{gs709G}  is $((\n^*\n) Uu, Uu)_{L^2(m)}$.  
 The right side is $(\hat H u, u)_{L^2(m_\psi)} + ((\l_0 - V) Uu, Uu)_{L^2(m)}$.
 Therefore 
 \beq
 ((\n^*\n) Uu, Uu)_{L^2(m)} + ((V- \l_0) Uu, Uu)_{L^2(m)} =(\hat H u, u)_{L^2(m_\psi)}. \notag
 \eeq
 Hence $(U^{-1}(H- \l_0) U u, u)_{L^2(m_\psi)} = (\hat H u, u)_{L^2(m_\psi)}$.
 Since $H$ and $\hat H$ are both symmetric \eref{gs709H} holds.
  \end{proof}

\begin{corollary} \label{corgs2} Suppose that $m$ satisfies the logarithmic 
Sobolev inequality \eref{mt1} :
\begin{align}
Ent_m(f^2) \le 2c \int_X |\n f|^2 dm                 \label{LSc3}
\end{align}
Then the ground state measure satisfies the inequality
\begin{align}
Ent_{m_\psi}(u^2) \le  2c\int_X |\n u|^2 dm_\psi + \int_X u^2(2c(\l_0 -V) + 2F) dm_\psi \label{gs720}
\end{align}
\end{corollary}
      \begin{proof} Putting 
      $f = u\psi =ue^{-F}$ we have 
 \begin{align}
 \int_X u^2 \log u^2 dm_\psi  &= \int f^2 \log (f^2e^{2F}) dm  \notag\\
  & = \int f^2 \log f^2 dm + \int f^2 2F dm.        \notag
 \end{align}
  Since $\|f\|_{L^2(m)}^2 = \|u\|_{L^2(m_\psi)}^2 $ we therefore have
  \begin{align}
  Ent_{m_\psi}(u^2) = Ent_m(f^2) + \int f^2 2F dm.     \label{gs724}
  \end{align}
  Combine this with \eref{LSc3} and then use \eref{gs709G} to find
  \begin{align}
   Ent_{m_\psi}(u^2) &\le 2c\int_X|\n(u\psi)|^2 dm +\int u^2 2F dm_\psi \label{gs725a} \\
   &=  2c\int_X|\n u|^2\, dm_\psi +  \int_X u^2(2c(\l_0 -V)+2F) dm_\psi.  \label{gs725}
   \end{align}
\end{proof}

\begin{remark} \label{remearly} {\rm  Many of the early approaches to 
the derivation  of a DLSI from a perturbation
of either $V$ or $F$ hinge on estimating the last integral in \eref{gs720}  
We  compare some of these approaches in Section \ref{secwp}. 
}
\end{remark}

\begin{remark} \label{remearly2}{\rm  If one assumes only that a defective logarithmic 
Sobolev inequality  holds, namely 
\beq
Ent_m(f^2) \le 2c \int_X |\n f|^2 dm   + D \|f\|_{L^2(m)}^2,
\eeq
 instead of \eref{LSc3}, then the ground state transformation yields, instead of \eref{gs720},
the inequality 
\begin{align}
Ent_{m_\psi}(u^2) &\le  2c\int |\n u|^2 dm_\psi       \notag \\
&+ \int_X u^2(2c(\l_0 -V) + 2F) dm_\psi 
                                + D\|u\|_{L^2(m_\psi)}^2 . \label{gs720b}
\end{align}
If one knew  that  $F- cV$ were 
 bounded above  then \eref{gs720b} would show
that $m_\psi$ also  satisfies  a DLSI. 
The discussion in Section  \ref{secwp} includes a history of  conditions that relate $F$ and $V$
in  such a pointwise manner. 
Such pointwise conditions do not  fall within the purview of this paper.
}
\end{remark}

\begin{remark}\label{remrosen2}{\rm  We can borrow a bit of the kinetic energy
from \eref{gs720} and shift it to the last term in \eref{gs720} to derive a condition on $\log \psi$ 
ensuring a DLSI:
Using the the relation $f = u\psi$ and the identity \eref{gs709G} 
 we find
\begin{align*}
&2c \int_X |\n (u\psi)|^2 dm  = 2(c+a)  \int_X |\n (u\psi)|^2 dm  - 2a \int_X |\n f|^2 dm\\
&=2(c+a)\(\int_X|\n u|^2\, dm_\psi + \int_X (\l_0 -V)f^2 dm\) -  2a \int |\n f|^2 dm\\
& =2(c+a)\int_X|\n u|^2\, dm_\psi  -2\int_X\( a |\n f|^2 +(c+a)(V - \l_0)f^2\) dm.
\end{align*}
Since, by \eref{gs725a}, $Ent_{m_\psi}(u^2) \le 2c\int_X |\n (u\psi)|^2 dm + 2\int_X F f^2 dm$,
we have
\begin{align}
Ent_{m_\psi}(u^2)   &  \le 2(c+a) \int_X|\n u|^2\, dm_\psi    \notag   \\
& + 2\int_X\Big\{ F f^2- a |\n f|^2 -(c+a)(V-\l_0) f^2\Big\}dm. \label{gs740b}
\end{align}

Suppose then  that there is a number $b$ such that $-\log \psi$ satisfies the form inequality
\beq
-\log \psi \le \{a \n^*\n +(c+a)(V-\l_0)\} +b      \label{gs741b}
\eeq
in $L^2(m)$.  Then  line \eref{gs740b} is at most $2b\| f\|_{L^2(m)}^2$ and we have the DLSI
\beq
Ent_{m_\psi}(u^2)  \le 2(c+a) \int_X|\n u|^2\, dm_\psi   + 2b \|u\|_{L^2(m_\psi)}^2      \label{gs742}
\eeq
This is a perturbation version of Rosen's lemma, \cite{DS84}. In practice one proves or assumes that
$-\log \psi \le  (c+a)(V-\l_0)  +b$, which implies \eref{gs741b} and is slightly more general than the condition
in Remark \ref{remearly2} but is still a pointwise condition.  Example \ref{exunb} shows how easily
this condition can fail even though the perturbed measure is hypercontractive. In that example $-\log \psi$
is unbounded above and below while $V$ is bounded.
}
\end{remark}

\subsection{The defective logarithmic Sobolev inequality} \label{secdlsi}

In the following two theorems we derive a defective logarithmic Sobolev inequality
for the ground state measure $m_\psi$ using progressively stronger conditions 
on the potential $V$. In the first theorem we assume  
that $\|e^{-V}\|_{L^\nu(m)} < \infty$ under the usual condition that $\nu > 2c$. 
We describe the defect partly in
terms of $\|\psi^{-1}\|_{L^s(m)}$ to illustrate how this quantity plays a central role. 
In the second theorem
we add on the hypothesis that $\|e^{V}\|_{L^\ka(m)} < \infty$ and use the 
bounds on $\|\psi^{-1}\|_{L^s(m)}$ derived in Section \ref{seclp}. The constants $c_\nu $ and $b_\ka$ that
occur repeatedly are defined in \eref{L329a} and \eref{s98f} respectively.
 
\begin{theorem} \label{thmDLSI2} 
Assume the hypotheses of Theorem \ref{thmesa2}.  
Suppose that
\begin{align}
a > c_\nu\ \ \ \text{and let}\ \ \  s =   \frac{2 c_\nu}{a - c_\nu} .\ \ \ 
\text{Equivalently,}  \ \ \   a =   (1 +\frac{2}{s})c_\nu.    \label{gs803} 
\end{align}
Assume that $\|\psi^{-1}\|_s < \infty$. Then
\begin{align}
Ent_{m_\psi} (f^2) &\le 2a \int_X |\n f|^2 dm_\psi   \notag \\
 & \ \ \ \ \ +2 \| f\|_{L^2(m_\psi)}^2\Big\{ \log \(\|\psi^{-1}\|_s\|e^{\l_0 - V}\|_\nu^a\) \Big\} 
 \label{gs805}
\end{align}
\end{theorem}

     \begin{theorem} \label{thmDLSI3} 
   In addition to the hypotheses of Theorem \ref{thmesa2} 
assume that       $\|e^V\|_\ka < \infty$ for some $\ka >0$. 
Suppose that 
\begin{align}
a > c_\nu b_\ka \ \ \ \text{and}\ \ \  \sigma > c_\nu b_\ka.       \label{gs803b}
\end{align}
Then
\begin{align}
Ent_{m_\psi} (f^2) &\le 2a \int_X |\n f|^2 dm_\psi   \notag \\
& \ \ \ \ \ +2 \| f\|_{L^2(m_\psi)}^2  \log\(\|e^{\l_0 - V}\|_\nu^{a+\sigma} \, 
\| e^{V-\l_0}\|_\kappa^{\ell(a) + \ell(\sigma)} \).    \label{gs806}
  \end{align}
  In particular, if $a = \sigma = t$, the unique point at which $\ell(t) = t$, $($cf. \eref{W854}$)$
  then the right side is independent of $\l_0$ and  there holds
  \begin{align}
Ent_{m_\psi} (f^2) &\le 2a \int_X |\n f|^2 dm_\psi   
 +2 \| f\|_{L^2(m_\psi)}^2  \log M^{2a}, \label{gs807}
  \end{align}  
with $M$ defined in \eref{M1}. 

 For arbitrary $a$ and $\sigma$ 
in the allowed range $(c_{\nu} b_{\ka}, \infty)$ there holds
 the $\l_0$ independent bound
\begin{align}
Ent_{m_\psi} (f^2) &\le 2a \int_X |\n f|^2 dm_\psi  +2 \| f\|_{L^2(m_\psi)}^2 \log M^{a+ \ell(a) +\sigma+\ell(\sigma)} .    \label{gs807b}
  \end{align}
  \end{theorem}
 
  Note that the lower bound, $c_\nu b_\ka$,   required of 
 $a$ and $\sigma$ in \eref{gs803b} depends  only on $c, \nu$ and $ \ka$. 
The proofs depend on the following lemma.

\begin{lemma} \label{lemDLS5} If $\|u\|_{L^2(m_\psi)} < \infty$ and   $Ent_{m_\psi}(u^2) < \infty$ then 
\begin{align}
 &\int_X u^2 F dm_\psi  
   \le  \frac{1}{2+s}\(Ent_{m_\psi}(u^2) + s\|u\|_{L^2(m_\psi)}^2 \log \|\psi^{-1}\|_{L^s(m)}\). \label{gs609b1}\\
  & \int_X u^2\(\nu(\l_0 -V) +2F\)  dm_\psi                      \notag \\
  &\qquad \qquad \ \ \ \  \le     \(Ent_{m_\psi}(u^2) 
                                  + \nu\|u\|_{L^2(m_\psi)}^2 \log\| e^{\l_0 - V}\|_{L^\nu(m)}\). \label{gs607b1}  \\
\ \notag \\
 & \int_X u^2\(2c(\l_0 -V) + 2F\) dm_\psi \le\(1 -(c/a)\)  Ent_{m_\psi}(u^2)     \label{gs605b1a}\\ 
&\ \ \ \ \ \ \ \ \ +(2c/a)\Big\{a\log\| e^{\l_0 - V}\|_{L^\nu(m)} 
        + \log \|\psi^{-1}\|_{L^s(m)} \Big\} \|u\|_{L^2(m_\psi)}^2  .     \notag
\end{align}
\end{lemma}
          \begin{proof} Apply Young's inequality \eref{BG500c} to find
\begin{align}
 (2+s)\int_X u^2 F dm_\psi &=\int_X u^2 \{(2+s)F\}  dm_\psi    \notag\\
 &\le Ent_{m_\psi}(u^2) + \|u\|_{L^2(m_\psi)}^2 \log \int_X e^{(2+s)F} dm_\psi  \notag\\
 &= Ent_{m_\psi}(u^2) + \|u\|_{L^2(m_\psi)}^2 \log \int_X e^{s F} dm    \notag\\  
  &= Ent_{m_\psi}(u^2) + s\|u\|_{L^2(m_\psi)}^2 \log \|\psi^{-1}\|_{L^s(m)},   \notag 
\end{align}
proving \eref{gs609b1}. To prove  \eref{gs607b1} apply Young's inequality again to find
\begin{align}
\int_X &u^2\(\nu(\l_0 -V) +2F\)  dm_\psi  \notag \\
      &  \le  Ent_{m_\psi}(u^2) + \|u\|_{L^2(m_\psi)}^2 \log \int e^{\nu(\l_0 -V) +2F} dm_\psi \notag \\
&=  Ent_{m_\psi}(u^2) + \|u\|_{L^2(m_\psi)}^2 \log\int_X e^{\nu(\l_0 -V)} dm  \notag\\
&= Ent_{m_\psi}(u^2) + \nu\|u\|_{L^2(m_\psi)}^2 \log \|e^{\l_0 -V}\|_\nu.  
\end{align}
For the proof of \eref{gs605b1a} we can apply \eref{gs609b1} and \eref{gs607b1} after 
decomposing the left side    of \eref{gs605b1a} as
\begin{align}
&\int_X u^2\(2c(\l_0 -V) + 2F\) dm_\psi       \notag  \\
&= \frac{2c}{\nu}\int_X u^2\(\nu(\l_0 -V) +2F\)  dm_\psi  +(1- (2c/\nu))\int_X u^2\, 2F dm_\psi \notag \\ 
& \le \frac{2c}{\nu}\(Ent_{m_\psi}(u^2) 
                                  + \nu\|u\|_{L^2(m_\psi)}^2 \log\| e^{\l_0 - V}\|_{L^\nu(m)}\) \notag  \\
 &+ (1- (2c/\nu)) \frac{2}{2+s}\(Ent_{m_\psi}(u^2) + s\|u\|_{L^2(m_\psi)}^2 \log \|\psi^{-1}\|_{L^s(m)}\).  \label{gs606b}
\end{align}
The definition \eref{gs803} gives $s/(2+s) = c_\nu a^{-1}$ and therefore $2/(2+s) = 1 - c_\nu a^{-1}$.
The combined coefficient of $Ent_{m_\psi}(u^2) $  in the last two lines is therefore 
$(2c/\nu) + (1- (2c/\nu)) (1 - c_\nu a^{-1}) =1-ca^{-1}$, since  $(1-(2c/\nu)) c_\nu = c$. This agrees
with the coefficient of $Ent_{m_\psi}(u^2)$ in \eref{gs605b1a}. 

       The coefficient of   $\log \|e^{\l_0- V}\|_\nu$ in \eref{gs606b} is clearly in agreement with that
in  \eref{gs605b1a}.
       
       The coefficient of $\|u\|_{L^2(m_\psi)}^2 \log \|\psi^{-1}\|_{L^s(m)}$   in \eref{gs606b} is 
 $2  (1- (2c/\nu)) s/(2+s)= 2  (1- (2c/\nu)) c_\nu a^{-1} = 2c/a$, giving agreement with \eref{gs605b1a}.
 \end{proof}

\bigskip
\noindent
\begin{proof}[Proof of Theorem \ref{thmDLSI2}] Combining \eref{gs720} with \eref{gs605b1a} we find
\begin{align*}
&Ent_{m_\psi}(u^2) \le  2c\int |\n u|^2 dm_\psi + \int_X u^2(2c(\l_0 -V) + 2F) dm_\psi \\
&\qquad  \qquad \ \ \  \le 2c\int |\n u|^2 dm_\psi   
+\(1 -(c/a)\)  Ent_{m_\psi}(u^2)    \\  
& \ \ \ \ \ \ \ \ \ 
+\(2c \log\| e^{\l_0 - V}\|_{L^\nu(m)} 
        + (2c/a) \log \|\psi^{-1}\|_{L^s(m)} \) \|u\|_{L^2(m_\psi)}^2 .
\end{align*}
Transfer the term $\(1 -(c/a)\)  Ent_{m_\psi}(u^2)$ to the left side and multiply by $a/c$ to find 
\begin{align*}
&Ent_{m_\psi}(u^2) \le    2a\int |\n u|^2 dm_\psi  \\
&+2\(a \log\| e^{\l_0 - V}\|_{L^\nu(m)} 
        +  \log \|\psi^{-1}\|_{L^s(m)} \) \|u\|_{L^2(m_\psi)}^2,
        \end{align*}
        which is \eref{gs805}.
\end{proof}

\bigskip
\noindent
\begin{proof}[Proof of Theorem \ref{thmDLSI3}]
We will bound the factor $ \|\psi^{-1}\|_{L^s(m)}$ in \eref{gs805} using the bound  \eref{W710a}, given in 
Corollary \ref{corub1}. To apply Corollary \ref{corub1} we must show that $s$,  
defined in \eref{gs803} is at most $s_0$. Using 
\eref{gs803b} and \eref{W851k}, we find
\begin{align}
s&<\frac{2c_\nu}{c_\nu b_\ka - c_\nu}= \frac{2}{b_\ka - 1} = s_0. 
\end{align}
We also need to verify that $r$, defined in \eref{W852b}, lies in $(0, r_0)$. 
But the condition \eref{gs803b} for $\sigma$ gives $(2r^{-1} - 1) > b_\ka$ and 
therefore $r < 2/(b_\ka +1) = r_0$, by \eref{W851k}. 
Insert the bound on $\|\psi^{-1}\|_{L^s(m)}$ from \eref{W710a} into \eref{gs805} to find \eref{gs806}. 

  By \eref{W854} there is a unique point $t \in (b_\ka c_\nu, \infty)$ such that $t= \ell(t)$.
If $a= t$ then $ \|e^{\l_0 - V}\|_\nu^{a} \,\| e^{V-\l_0}\|_\kappa^{\ell(a)} =  
\(\|e^{\l_0 - V}\|_\nu \,\| e^{V-\l_0}\|_\kappa\)^a = M^a$.
Choosing $a = \sigma = t$ we see then that \eref{gs807} follows from \eref{gs806}. 

If $a$ and $\sigma$ are chosen arbitrarily in the allowed range $(c_\nu b_\ka, \infty)$ then
\eref{gs807b} follows from   \eref{s3}, \eref{s4} and \eref{gs806}. 
\end{proof}

\begin{remark} \label{remparam} {\rm  The parameters $a$ and $\sigma$ in \eref{gs806}
 are at our disposal as long as both are chosen greater than $c_\nu b_\ka$. 
 We saw that  if $a = \sigma = t$, with $t$ chosen to make $\ell(t) = t$ as in \eref{W854}, then the bound
 \eref{gs806} reduces to the $\l_0$ independent bound \eref{gs807}. 
 But if  we choose $a = \sigma = t$, with $t$ chosen  to minimize $\ell(t) + t$, as in \eref{W854a} then the 
 estimate in the $\l_0$ independent defect given in \eref{gs807b} would be minimized.
The choice of some special values of $a$ and $\sigma$, as well as the behavior of these special 
values as $\ka\downarrow 0$ and $\nu \downarrow 2c$, may be of significance in some applications.
}
\end{remark}

\subsection{Cases: $V$ is bounded below, or above, or both} \label{seccases}

If $V$ is bounded below or above then one can let  $\nu\uparrow \infty$ or $\ka \uparrow \infty$, respectively,
in the formulas of the preceding sections,  
 giving some clarifying simplifications.

Throughout this section we assume that the logarithmic Sobolev inequality \eref{mt1} holds.
    
\subsubsection{$V$ bounded below} \label{secbelow}
If $V$ is bounded below then all of the significant  quantities in Theorem \ref{thmns2} and 
Corollary \ref{corhb2} have limits  as $\nu \uparrow  \infty$. The following corollary shows that
the Sobolev coefficient, $pc_\nu(p)$in \eref{L325} converges to the classical one, 
the interval of validity $(q_0,p_0)$ converges to $(1,\infty)$ and  
the minimum time $\tau(p) - \tau(q)$ 
to boundedness converges to Nelson's shortest time.

\begin{corollary} \label{corbelow} $(V$ bounded below$)$. Suppose that the logarithmic 
Sobolev inequality  \eref{mt1} holds.    
 Assume that $V$ is bounded below and that $V \in L^{p_1}(m)$ for some $p_1 > 2$. 
 Then  $\n^*\n +V$ is essentially self-adjoint. Its closure $H$ has a unique positive $($a.e.$)$ 
 ground state $\psi$.      There holds
\begin{align}
 &Ent_m(|u|^p)  
 \le \frac{c p^2}{2(p-1)} \int_X \<(H -(\inf V)) u, u_p\> dm\ \  \text{if}\ \ p\in (1, \infty).  \label{L325below}
 \end{align} 
 Furthermore
 \beq
\| e^{-tH}\|_{q\to p} \le e^{- t\inf\, V}         
 \ \text{if} \ \ \ e^{-t/c} \le \sqrt{\frac{q-1}{p-1}}, \ \ \ 1 < q < p < \infty.   \label{L289a}
\eeq
In particular the time to boundedness, \eref{L505}, reduces to Nelson's classical time to contraction,
which is determined by $\tau_0(p) = (c/2) \log (p-1), 1 < p <\infty$. 
\end{corollary}
   \begin{proof} If  $V$ is bounded below we have $\lim_{\nu\uparrow \infty} \|e^{-V}\|_\nu \to e^{-\inf V}$.
 Moreover $p_0\uparrow \infty$ and $q_0\downarrow 1$ when $\nu \uparrow \infty$, as we see from \eref{L341d} and \eref{L313qi}. 
  From \eref{L505} we see that $\tau(p)$ 
  converges to $\tau_0(p) :=(c/2)\log (p-1)$.
  If $t >(c/2)\log (p-1) -(c/2)\log (q-1)$ then, by \eref{L289},  $ \| e^{-tH}\|_{q\to p} \le \|e^{-V}\|_\nu^t $ 
  holds for large enough $\nu$, leaving aside for a moment the technical issue of self-adjointness of $H$.
  Therefore  $ \| e^{-tH}\|_{q\to p} \le e^{- t\inf\, V}$ if $t> (c/2)\log\frac{p-1}{q-1}$ and also 
  if $t\ge (c/2)\log\frac{p-1}{q-1}$ by strong continuity of $e^{-tH}$.  This proves \eref{L289a}.
  It may be of use to note monotonicity:  $\tau(p)- \tau(q) \downarrow \tau_0(p) - \tau_0(q)$ 
  as $\nu \uparrow \infty$ by Corollary \ref{corhb2}.

   Concerning the Sobolev coefficient in \eref{L325below}
   observe first that $\nu/p_0 \to c/2$ 
   as $\nu \to \infty$, as we see from   \eref{L313q}. 
In view of  \eref{L329} we therefore have $ c_\nu(p) = \frac{\nu}{p_0 - p} \frac{p}{p-q_0} \to (c/2) \frac{p}{p-1}$. 
 If    $1 < p < \infty$ then $ q_0 < p < p_0$ for sufficiently large $\nu$. 
   Keeping $p$ and $u$
   fixed, we may take the limit in the inequality \eref{L325} as $\nu \to \infty$.  
   Since $p  c_\nu(p) \to \frac{cp^2}{2(p-1)}$ while $\log\|e^{-V}\|_\nu \to - \inf V$,   
   the inequality \eref{L325}     goes over to \eref{L325below} as $\nu \to \infty$.
   
        Since $V \in L^{p_1}$ for some $p_1 > 2$ the hypothesis \eref{EU7} of Theorem \ref{thmesa2} holds for 
 some large enough finite  $p_0$.  $\n^*\n + V$ is essentially self-adjoint on its domain and
 all the conclusions of Theorem \ref{thmesa2} hold.
   \end{proof}

\begin{corollary} \label{corpolbelow}   
$(Polynomial\ growth\ of\ \|\psi\|_p)$.  If the hypotheses of Corollary \ref{corbelow} hold then
\begin{align}
\|\psi\|_p \le (p-1)^{(c/2) \sup (\l_0 - V)}, \ \ \ \ p \ge 2.    \label{L804}
\end{align}
In particular 
\beq
\psi \in \cap_{p < \infty} L^p(X,m). \label{L804b}
\eeq
Further,
\begin{align}
\|\psi\|_r \ge e^{-\sigma \sup(\l_0 - V)},\ \  0 < r < 2, \ \ \ \ \sigma =  c(2r^{-1} - 1).     \label{L808}
\end{align}
\end{corollary} 
     \begin{proof} 
  Using \eref{L289a} with $q = 2$ and $p>2 $ we find that   
  \begin{align}
  \|\psi\|_p &= \|e^{t\l_0}e^{-tH} \psi\|_p  \notag\\
  & \le e^{t\l_0} \| \psi\|_2  e^{-t\inf V}\ \ \ \text{if}\ \ e^{-t/c} \le \sqrt{\frac{1}{p-1}}   \label{L805}\\
  &= \|\psi\|_2 e^{t\sup(\l_0 - V)}  .  \notag  
  \end{align}
  Take 
  \beq
  t = \tau_0(p) := (c/2) \log (p-1),   \label{L806}  
  \eeq
   which is the optimal value allowed in \eref{L805}. 
  Since $\|\psi\|_2 = 1$ we find \eref{L804}.
 
  For the proof of \eref{L808} observe that as $\nu\uparrow \infty$,  
  $c_\nu \to c$ by \eref{L329a} 
  while the  right side of  \eref{L512} converges to the right side of \eref{L808}.
  \end{proof}

\subsubsection{$V$ bounded above}  \label{secabove}

\begin{corollary} \label{corpolabove} 
 $($Polynomial growth of $\|\psi^{-1}\|_s$ $).$ 
 Assume that \eref{mt1} holds, that $\|e^{-V}\|_\nu < \infty$ for some $\nu > 2c$ and that
  $V$ is bounded above. Then 
$\|\psi^{-1}\|_s$ has at most polynomial growth as $ s\uparrow \infty$. In particular
\begin{align}
\| \psi^{-1}\|_s \le (1+s)^{(c/2)\sup (V - \l_0)} \(\|e^{\l_0-V}\|_\nu^{3c_\nu} 2^{(c/2)\sup (V - \l_0)}\), \ 0 < s <\infty. 
\label{L819}
\end{align}
\end{corollary}
\begin{proof} Since $V$ is bounded above we have $\| e^V\|_\ka < \infty$ for $0 < \ka \le \infty$.
In Theorem \ref{thmmp1} all of the significant quantities, $b_\ka$, $\ell_\ka$,
$r_0$ and $s_0$ have limits as $\ka \uparrow \infty$, which we can summarize as follows.
\begin{align}
b_\infty &= 1, \ \    r_0= 1,\ \  s_0 =\infty, \ \ \ell_\infty(t) 
                  =  (c/2) \log \frac{t+ c_\nu}{t-c_\nu},\ t > c_\nu .                     \label{L820}
 \end{align}
 The first of these identities follows from the definition \eref{s98f}, the second and third from \eref{W851g},
 while the fourth  follows from the definition \eref{W852}.

 Since $s_0 = \infty$  
 the moment  product theorem, Theorem \ref{thmmp1}, shows that  
 $\|\psi^{-1}\|_s < \infty$ for all $s < \infty$. 
To make use of the   moment product inequality \eref{W753g}  observe that $a$ is given in terms of $s$
by \eref{W852b}. Thus 
\begin{align}
\ell_\infty(a) &=(c/2) \log \frac{(2s^{-1} +1)c_\nu+ c_\nu}{(2s^{-1} +1)c_\nu-c_\nu},\ \ \ 0 < s < \infty \notag\\
&=(c/2) \log (1+s), \ \ \ 0 < s <\infty.
\end{align}
Similarly we have
\begin{align}
\ell_\infty(\sigma) &= (c/2) \log\frac{(2r^{-1} - 1) + 1}{(2r^{-1} - 1) -1} \notag\\
&= -(c/2) \log(1-r).
\end{align}
Inserting these values into \eref{W753g} we find
\begin{align}
\|\psi\|_r \| \psi^{-1}\|_s &\le   \| e^{V-\l_0}\|_\infty^{(c/2) (\log (1+s) - \log(1-r))}    \notag \\
&=  \frac{(1+s)^{(c/2)\log\| e^{V-\l_0}\|_\infty}}{(1-r)^{(c/2)\log\| e^{V-\l_0}\|_\infty}}, 
                         \ \ 0 < s < \infty, \ 0 < r <1 \notag\\
& =\frac{(1+s)^{(c/2)\sup (V - \l_0)}}{(1-r)^{(c/2)\sup (V - \l_0)}}, 
                         \ \ 0 < s < \infty, \ 0 < r <1      \label{L822}
\end{align}
Choose $r = 1/2$ for simplicity. \eref{L512} and \eref{L513}  show 
that $\|\psi\|_{1/2} \ge  \|e^{\l_0-V}\|_\nu^{-3c_\nu}$.
Therefore
\begin{align}
\| \psi^{-1}\|_s &\le  (1+s)^{(c/2)\sup (V - \l_0)}
 \(\frac{ \|e^{\l_0-V}\|_\nu^{3c_\nu}}{(1/2)^{(c/2)\sup (V - \l_0)}}\) .
\end{align}
This proves \eref{L819}.
\end{proof}

\subsubsection{$V$ bounded.}

When $V$ is bounded several expressions take a simpler form.
Observe first that
\begin{align}
\inf V \le \l_0 &\le \sup V.  \label{L830}\\
\sup (V - \l_0) &\le \sup V - \inf V = Osc (V).         \label{L831}\\
\sup (\l_0 - V) &\le \sup V - \inf V = Osc (V).         \label{L832}
\end{align}

\begin{corollary} \label{corallp} 
 $($Bounds on $\|\psi^{\pm1}\|_p$ when $V$ is bounded$)$.
Assume that $V$ is bounded. Then
\begin{align}
\|\psi\|_p &\le (p-1)^{(c/2) Osc(V)}, \ \ \ \ p \ge 2.    \label{L833a} \\
\|\psi\|_r &\ge   e^{-\sigma Osc (V)},      \ \ \  0<r< 2, \ \ \sigma = c(2r^{-1} -1).   \label{L834a}\\
\|\psi\|_r \|\psi^{-1}\|_s &\le \(\frac{1+s}{1-r}\)^{(c/2)Osc(V)} , \ \ \   0 < s <\infty, \ \ 0 < r <1 .          \label{L835a}\\
\|\psi^{-1}\|_s &\le  \(\frac{1+s}{1-r}\)^{(c/2)Osc(V)} e^{\sigma Osc (V)} ,        
\ \ \ 0 < s < \infty,       \ \ 0 < r <1 .        \label{L836a}
\end{align}
\end{corollary}
      \begin{proof} 
 \eref{L833a} follows from \eref{L804}  and \eref{L832}.      
 \eref{L834a} follows from \eref{L808} and \eref{L832}. 
 \eref{L835a} follows from \eref{L822} and \eref{L831}. 
       \eref{L836a} follows from \eref{L835a} and \eref{L834a}.       
\end{proof}

\begin{remark}  {\rm 
We saw in Example \ref{exunb} that $\psi$ and $\psi^{-1}$ can be unbounded even if $V$ is bounded.
But Corollary \ref{corallp} shows that the $L^p$ norms  of $\psi$ and $\psi^{-1}$ can only grow
 polynomially in $p$ when $V$ is bounded.
}
\end{remark}

   \begin{corollary} \label{coreF2} If $V$  is bounded then
\begin{align}
\int_X e^{b F^2} dm < \infty\ \ \  \forall\ \ \ b< \infty,      \label{L841}
\end{align}
where $F = - \log \psi$.
\end{corollary}
    \begin{proof} The polynomial growth conditions in  Corollary \ref{corallp}  imply that there are constants
$C_1, C_2$ independent of $t$ such that
\begin{align}
\|\psi^{\pm 1} \|_t \le C_1(1+t)^{C_2}\ \ \ \ \ \ t > 0:    \label{L840}
\end{align}
Indeed, put $r = 1/2$ in \eref{L836a} to derive \eref{L840}  for  $\| \psi^{-1}\|_t$, while
 \eref{L833a}  gives \eref{L840} for $\|\psi\|_t$ in case $t \ge 2$.    Use $\|\psi\|_ t \le \|\psi\|_2$
 for $0 < t <2$.
 
     We can write \eref{L840}  in terms of $F$ in the equivalent form
\beq
\int_X e^{tF(x)} dm(x) \le C_1^{|t|} e^{C_2|t| \log (1+|t|)}, \ \ \ t \in \R.
\eeq
Suppose that $b >0$. In the identity 
$
e^{by^2/2} = (2\pi b)^{-1/2} \int_{-\infty}^\infty e^{ty} e^{-t^2/(2b)} dt
$
insert $y = F(x)$ and take expectation to find
\begin{align}
\int_X e^{bF(x)^2/2} dm(x) &= (2\pi b)^{-1/2} \int_{-\infty}^\infty \int_Xe^{tF(x)}dm(x)e^{-t^2/(2b)} dt \notag\\
&\le  (2\pi b)^{-1/2} \int_{-\infty}^\infty C_1^{|t|} e^{C_2|t| \log(1+ |t|)} e^{-t^2/(2b)} dt  \notag\\
& < \infty.     \notag
\end{align}
\end{proof}

\begin{remark}\label{remtv} {\rm (Two variants of Corollary \ref{coreF2}).  The two inequalities in \eref{L840} correspond to the two conditions:
$V$ is bounded below or above, from which they were derived. If just one of these two conditions holds
then \eref{L841} can be replaced by
\begin{align}
\int_{F \ge 0} e^{bF^2/2} dm &< \infty\ \ \ \text{if $V$ is bounded above.}  \label{L801}\\
\int_{F \le 0} e^{bF^2/2} dm &< \infty\ \ \ \text{if $V$ is bounded below.}   \label{L801a}
\end{align}

One need only start with the inequality  
$
e^{by^2/2} \le 2  (2\pi b)^{-1/2} \int_0^\infty e^{t|y|} e^{-t^2/(2b)} dt
$
and proceed as in the proof.  For example on the set $\{ F \le 0\}$ we have 
$e^{t|F(x)|} = e^{-tF(x)} = \psi(x)^t$ for $t >0$ and therefore, using \eref{L840} with $\|\psi^{+1}\|_t$ we have
\begin{align}
\int_{ F \le 0} e^{bF(x)^2/2} dm(x)
&\le  2(2\pi b)^{-1/2} \int_0^\infty C_1^{|t|} e^{C_2|t| \log(1+ |t|)} e^{-t^2/(2b)} dt < \infty , \notag
\end{align}
A similar argument, using the bound \eref{L840} for $\|\psi^{-1}\|_t$, 
 gives the same bound for $\int_{F \ge 0} e^{bF^2} dm$.
}
\end{remark}

\begin{corollary}  \label{corDLSIbV}  $($DLSI for bounded $V$ $)$. Assume that \eref{mt1} holds and 
that $V$ is bounded. Let $\psi$ denote the normalized ground state for $\n^*\n +V$. Then
\begin{align}
Ent_{m_\psi}(f^2) \le 2a \int_X |\n f|^2 dm_\psi 
 + \{2D_{a,\sigma}  Osc(V)\} \| f\|_{L^2(m_\psi)}^2   \label{L855}
\end{align}
for any $a > c$ and  $\sigma > c$, where
\begin{align}
D_{a, \sigma} &= a + \sigma + \ell_0(a) + \ell_0(\sigma)\ \ \ \text{and}  \label{L856}\\
\ell_0(t) &= (c/2)\log \frac{t+c}{t-c},\ \ \ t > c .      \label{L857}
\end{align}
In particular, choosing $a = \sigma =2c$ we have
\begin{align}
Ent_{m_\psi}(f^2) \le 4c \int_X |\n f|^2 dm_\psi  + \{2c (4 + \log 3) Osc(V)\} \|f\|_{L^2(m_\psi)}^2. \label{L858}
\end{align}
\end{corollary}
    \begin{proof} Since $V$ is bounded we can let $\nu$ and $\ka$ increase to infinity in \eref{gs806}.
First observe that $c_\nu\downarrow c, a_\nu \uparrow 1$ as $\nu\uparrow \infty$ 
and $b_\ka \downarrow 1$ as 
$\ka \uparrow \infty$ by their definitions \eref{L329a}, \eref{L341d} and \eref{s98f}.  Thus if $a>c$ and 
$\sigma >c$ then \eref{gs803b} holds for large enough $\nu$ and $\ka$ and we may  apply 
 \eref{gs806}.
 We have $\|e^{\l_0 - V}\|_\infty = e^{\sup(\l_0 - V)}$ and $ \|e^{V- \l_0}\|_\infty = e^{\sup(V-\l_0)}$.
Further, $\ell(t) $, defined in \eref{W852} goes over to $\ell_0(t)$, defined in \eref{L857}, as 
$\nu\uparrow \infty$ and $\ka \uparrow \infty$.
Therefore \eref{gs806} goes over to
\begin{align}
&Ent_{m_\psi} (f^2) \le 2a \int_X |\n f|^2 dm_\psi   \notag \\
& +2 \| f\|_{L^2(m_\psi)}^2 \( (a+\sigma) \sup(\l_0 - V)  +(\ell_0(a) + \ell_0(\sigma))\sup (V - \l_0)\).
                                                                  \label{gs806a}
  \end{align}    
  \eref{L855} now follows from \eref{gs806a}, \eref{L831} and \eref{L832}. 
  
  Now $\ell_0(2c) = (c/2)\log 3$.  Hence, for $a = \sigma = 2c$ we have $D_{a,\sigma} = 4c + c\log 3$. 
This proves \eref{L858}. 
  \end{proof}

\begin{remark}\label{remrec} {\rm  (Recovery of \eref{mt1}.) In case $V=0$ (or constant) 
the ground state for $ \n^*\n + V$ is the constant $1$. Therefore $m_\psi = m$. Moreover
$Osc(V)= 0$. Hence \eref{L855} reduces to   $Ent_m (f^2) \le 2a \int_X |\n f|^2 dm $, which is valid
for all $a >c$. Taking the limit $a\downarrow c$ yields the original LSI, \eref{mt1}, again.
}
\end{remark}

\begin{remark}\label{remhb} {\rm (Invariance of DLSI).
 In their fundamental paper  \cite{DS84}, Davies and Simon were interested mainly
in intrinsic ultracontractivity. They proved that intrinsic ultracontractivity is invariant under perturbation
of the Sch\"odinger potential by a bounded potential. They raised the question as to whether intrinsic
hyperboundedness (referred to as intrinsic hypercontractivity at that time) 
 was also invariant under perturbation by a bounded potential. 
  At the infinitesimal level this amounts to asking whether, for  a bounded potential $V_1$,  
  the ground state measure for 
$-\Delta + V_0 + V_1$ satisfies a DLSI  when the ground state measure $\psi_0^2 dx$ 
for $-\Delta + V_0$ does.

 We have answered this affirmatively in this 
paper but only when the defect for  $\psi_0^2 dx$ is zero. In this case the
 perturbed ground state measure also has defect zero. 
 As to whether perturbation of a DLSI yields
another DLSI under some conditions on the perturbing potential is still an open question.
}
\end{remark}

\section{Spectral gap} \label{secsg}
  In this section we will prove that the Dirichlet form operator for $m_\psi$
has a spectral gap and that it can be estimated  from below by a function of $c, \ka, \nu$ and $M$. 

An irreducible Dirichlet form operator will have a spectral gap under some mild qualitative conditions, 
but the size of the gap might not be quantifiable by these conditions.
 Feng-Yu Wang has shown, \cite{Wang2014a},
that if the Dirichlet form is irreducible - which ours is - and satisfies a              
defective Poincar\'e inequality  then there is a spectral gap.  
 Miclo, \cite{Mic2015}, has shown that if the Dirichlet form is irreducible
and satisfies a defective logarithmic Sobolev inequality then it also has a spectral gap.
It does not appear that the methods used in these papers will yield a quantitative lower bound on the gap.

\subsection{Small perturbations: Wang's method}
Feng-Yu Wang, \cite[Corollary 1.2]{Wang2004},  
 showed that if the defect in a defective logarithmic 
Sobolev inequality is sufficiently small then there is a spectral gap that  can be quantitatively estimated. 
 In our context his theorem shows  that if, for some probability measure $\mu$ on a 
 Riemannian manifold $X$, there holds 
 \begin{align}
 Ent_\mu(f^2) \le 2C_1 \int_X|\n f|^2 d\mu + C_2 \|f\|_2^2     \label{T12}
 \end{align}
 with $C_2 < \log 2$ then $\n^*\n$ has a spectral gap that can be estimated by 
 \begin{align}
 Gap\ \n^*\n \ge  \frac{ \log(3 - 4b)}{C_1\log3},\ \ \ \text{where}\ b =\sqrt{(1 - e^{-C_2})/2}. \label{T14}
 \end{align}

 Wang first proved an equivalent,  exponentiated version of this Corollary. 
 It asserts that if 
  \begin{align}
  \|e^{-t\n^*\n}\|_{L^2 \to L^4}^4 \le A               \label{T15}
  \end{align}
   with $1 \le A <2$ for some $t >0$ then $\n^*\n$ has a spectral gap. 
 This extends a theorem of 
 B. Simon \cite[Theorem 2]{Simon1976}, who showed $\n^*\n$ has a spectral gap if $A=1$ for some $t >0$. 
 Miclo, \cite[Proposition 11]{Mic2015}, gave an example of a  
 similar form that showed that if one only knows \eref{T15} holds for some $A\ge2$ then, 
 although there is still a spectral gap, 
 there can be  no quantitative bound on the spectral gap dependent only on $A$ and $t$ 
 (or equivalently, $C_1,C_2$ and $t$).
 His example strongly suggests that the same is true in our context.

\bigskip 
In our setting, the number $M\equiv  \|e^V\|_\ka \|e^{-V}\|_\nu $, introduced in Section \ref{secM}, controls
the defect, as we see in \eref{gs807b}. Combining this estimate of the defect with Wang's theorem 
gives  the following consequence.

\begin{theorem} \label{thmws}  Suppose that
\beq
\(a + \sigma +\ell(a)  + \ell(\sigma)\)\log M< (1/2) \log 2    \label{T5}
\eeq
for some $a$ and $\sigma$ in the allowed range $(c_\nu b_\ka, \infty)$.
 Then there is a constant $C_3$ such that
 \begin{align}
 Ent_{m_\psi}(f^2)    \le C_3\int_X|\n f|^2 dm_\psi .
 \end{align}
 \end{theorem}
        \begin{proof} The inequality \eref{gs807b} shows that the measure $m_\psi$ satisfies
        a defective logarithmic Sobolev inequality with defect $2 (a + \sigma +\ell(a)  + \ell(\sigma))\log M$
        for any choice of $a$ and $\sigma$ in the allowed range $(c_\nu b_\ka, \infty)$. The condition \eref{T5}
        shows therefore that Wang's condition on the defect is satisfied.
            From Wang's  gap,  \eref{T14},  and \eref{gs807b} we can compute $C_3$ via  
            Rothaus'  tightening theorem. 
            See  \cite[Proposition 5.1.3]{BGL}  for an efficient exposition of this method or 
            Proposition \ref{propbgl}  below.
\end{proof}

\begin{remark} {\rm For any choice of $a$ and $\sigma$ there are clearly potentials $V$ for 
which \eref{T5} is satisfied because $\log M =0$ when  $V=0$.
}
\end{remark}

\begin{example}\label{exwang}{\rm  In case $V$ is bounded we can use the estimate
\eref{L858} for the defect. We find that Wang's criterion for a spectral gap holds if
\begin{align}
 Osc(V)< \frac{\log 2}{ 2c (4 + \log 3) }.
 \end{align}
 A lower bound for the spectral gap can be computed from \eref{T14} with $C_1 = 2c$ and
 $C_2 = \{2c(4 + \log 3\} Osc(V)$.
}
\end{example}

 \subsection{General perturbations: Aida's method} \label{seclargeV}

 We will prove that the Dirichlet form operator for the probability measure $\psi^2 dm$
has a spectral gap at zero.
 This will allow us, in Section \ref{sectightening},  to remove 
the defect in \eref{gs807b}.    
 
The Deuschel-Holley-Stroock theorem, \cite{HS1987,DS1990}, asserts that if  the logarithmic 
Sobolev inequality  \eref{mt1} 
holds for a probability measure  $m$ and if $w$ is a strictly positive  weight, which is bounded and 
bounded away from zero, then the  logarithmic  Sobolev inequality  
\begin{align}
Ent_{wm}(u^2) \le 2c_1\int_X |\n u|^2 wdm\ \  \text{holds with}\ \  c_1 = c\frac{\sup w}{\inf w} .  \label{sg202}
\end{align}
 See e.g. \cite[Proposition 5.1.6]{BGL} for an efficient proof. 
 Similarly, if $m$ satisfies a Poincar\'e inequality, 
 $Var_m(f) \le \gamma \int |\n f|^2 dm$, then
 \begin{align}
 Var_{wm}(f) \le \gamma_1\int |\n f|^2 wdm, \ \ \  \gamma_1 = \gamma \frac{\sup w}{\inf w}.\label{sg203}
 \end{align}
 The latter follows easily from the inequalities 
 $Var_{wm}(f) \le \int (f - \int f dm)^2 wdm \le (\sup w) Var_m(f)$ and 
 $\int |\n f|^2 dm \le (1/\inf w)\int |\n f|^2 wdm$.

  The density $\psi^2$ for the ground state measure is typically neither bounded  
 nor bounded away from zero. 
 The DHS theorem was consequently  inapplicable in Section \ref{secDLSI} for proving a DLSI.
 The inequality \eref{sg203} is similarly inapplicable for proving that $m_\psi$ satisfies a Poincar\'e inequality,
 even  though $m$ does.

  Aida  developed a method in \cite{Aida2001}    
  for proving that $\psi^2 dm$ satisfies a Poincar\'e inequality even when $\psi$ is not bounded 
  nor bounded away from zero.
  He
 decomposed the space $X$ into
 the three regions $\{ \psi < \ep\}, \{ \ep \le  \psi \le K\}$ and $ \{\psi >K\}$ and used the idea 
 behind the DHS  theorem on the middle region. 
 He established bounds associated to the remaining two regions
 in terms of  the quantities defined below in Notation \ref{notA3}.
 In this section we will derive Aida's lower bound on the spectral gap in terms of these quantities. 
 We will also  make use  
 of the already established defective logarithmic Sobolev inequality \eref{gs807b}, which  
 was  unavailable to Aida at 
 the time of his paper  \cite{Aida2001}.
 In the next section we will show how our assumptions
 on the potential $V$ allow us to make quantitative estimates of these quantities, 
 thereby producing a quantitative estimate of the spectral gap.

 $\psi$ need not be a solution to the Schr\"odinger equation in this section.

 \begin{notation} \label{notA3}  {\rm Suppose that $\psi$ is an a.e.  strictly positive function in $L^2(m)$
 with $\|\psi\|_{L^2(m)} = 1$. 
   Define
\begin{align}
dm_\psi = \psi^2 dm\ \ \ \    \text{and}\ \ \ F = - \log \psi.     \label{505}
\end{align}  
 Let $0 < \ep < 1 < K < \infty$. Define
 \begin{align}
 A_\ep &= m(\psi \le \ep). \label{sg10}\\ 
 B_\ep & = \int_{\psi \le\ep}|\n F|^2 dm.\label{sg11}\\
 C_K &= \int_{\psi > K} \psi^2 dm. \label{sg12}
 \end{align}
 }
 \end{notation}

 \begin{theorem} \label{thmsg2}  $($Aida's Theorem$)$.  
 Referring to Notation \ref{notA3},  
 assume that 
\beq
\int_X |\n F|^2 dm < \infty.    \label{506}
\eeq
 Suppose that for some numbers $ \gamma >0, B >0, D \ge 0$ and for all real valued functions $u$ of finite energy there holds
 \begin{align}
 \int_X\(u^2 - \<u\>_m^2\)dm 
                               &\le \gamma \int_X|\n u|^2 dm\ \ \ \ \  \label{503f}
                               \ \  \ \text{and} \\
Ent_{m_\psi}(u^2) &\le B \int_X |\n u|^2 dm_\psi + \|u\|_2^2 D . \label{504}    
 \end{align}
Then there is a number $\gamma_1$               
 such that
 \begin{align}
 \int_X\(u^2 - \<u\>_{m_\psi}^2\)dm_\psi  
             &\le \gamma_1 \int_X|\n u|^2 dm_\psi.    \label{A363}
 \end{align} 
 If, for some $\ep > 0$ and $K > 1$,
 there holds 
 \begin{align}
 \(2K^2\( 2\gamma B_\ep +A_\ep \)+4C_K\)  e^{ 12(D + e^{-1})}  &\le 1/3   \label{A360}
 \end{align}
 then one may choose 
 \begin{align}
 \gamma_1 \le  B + 8\gamma\(\frac{K}{\epsilon}\)^2 .   
                       \label{A365b}
 \end{align}
 Such an $\ep$ and $K$ always exist.
 \end{theorem}
 
 The proof depends on the following  lemmas. The invariance of a weak Poincar\'e inequality 
 under  
 perturbation of a measure by insertion of a density was proven by  Rockner and Wang 
 \cite[Theorem 6.1]{RW2001}  and by Aida \cite[Lemma 2.2]{Aida2001}. 
 The next lemma is a form of Aida's weak Poincar\'e inequality 
 in the case that the unperturbed measure satisfies a Poincar\'e inequality, which is 
 the only case that we need.

\begin{lemma} \label{lemA1} $($Weak Poincar\'e inequality$)$. Referring to  
Notation \ref{notA3} again,
assume that \eref{506} and \eref{503f} hold. Suppose that $u$ is bounded and has finite energy.
  Let $0 < \epsilon < 1 < K < \infty$.
Then
\begin{align}
\| u - \<u\>_{m_\psi} \|_{L^2(m_\psi)}^2 &\le \gamma\(\frac{2K}{\epsilon}\)^2 
\int_X |\n u|^2  dm_\psi     +\zeta \|u\|_\infty^2                                          \label{520e}
\end{align}
where
\begin{align}
\zeta =2K^2\( 2\gamma B_\ep   
                        +A_\ep \)+4C_K.                       \label{422e1}
\end{align}
For any $\delta >0$ there exists $\epsilon$ and $K$ such that $\zeta < \delta$.
\end{lemma}

\begin{sublemma}\label{lemA2b}  Referring to  
Notation \ref{notA3} again, for any bounded real  \linebreak 
valued function $u$ on $X$ we have
\begin{align}
\| u - \<u\>_{m_\psi} \|_{L^2(m_\psi)}^2  
           \le K^2 \| u - \<u\>_m\|_{L^2(m)}^2 + 
           4C_K \|u\|_\infty^2.          \label{A100}
\end{align}
\end{sublemma}
 \begin{proof} For any real number $a$ we have
            \begin{align*}
  \| u - \<u\>_{m_\psi} \|_{L^2(m_\psi)}^2  &\le \| u - a\|_{L^2(m_\psi)}^2  \\
  &\le K^2\int_{\psi \le K} (u-a)^2 dm + \(\int_{\psi >K} \psi^2 dm \)\|u-a\|_\infty^2 \\
  &\le K^2\int_X (u-a)^2 dm + C_K \|u-a\|_\infty^2.
\end{align*}
Choose $a = \<u\>_m$. Then $ \|u-a\|_\infty \le 2\|u\|_\infty$ and \eref{A100} follows.
\end{proof}

 \begin{sublemma} \label{lemA3}   
 Assume that the hypotheses of Lemma \ref{lemA1} hold.
  Then   
\begin{align}
\|u - \<u\>_m\|_{L^2(m)}^2  \le  \frac{4\gamma}{\epsilon^2}\int_X |\n u|^2 \psi^2 dm    
+  \Big\{4\gamma  B_\ep  +2 A_\ep\Big\}  \|u\|_\infty^2        \label{191}
\end{align}
\end{sublemma}
     \begin{proof}  
     We need to use a regularized version of the function  $[0,\infty) \ni t \to min(t, 1)$.         
Let  $ 0 < \delta < 1/2$ and let $f$ be a smooth non-decreasing real valued 
function on $[0,\infty)$ such that  $f(t)= t$ for $0\le t \le 1-\delta$ and $f(t) = 1$ for $ t \ge 1+\delta$ and such that $f'(t) \le 1$ everywhere.  In the end we will let $\delta \downarrow 0$.

 Let $\phi(t) = f(t/\ep)$. 
Then $\phi'(t) \le \ep^{-1}$ everywhere, $\phi(t) \le t/\ep$ everywhere, and, when $t\ge\ep(1+\delta)$ 
we have $\phi(t)=1$ and $ \phi'(t) =0$.
 Let 
\beq
\chi(x) = \phi(\psi(x)).           \label{201e}
\eeq
Since  $\|u - \<u\>_m\|_2^2 \le \|u - a\|_2^2$ for any real number $a$ we have
\begin{align} 
\|u - \<u\>_m\|_2^2 &\le \|u - \<u\chi\>_m\|_2^2   \notag\\
&=\| (u- u\chi) + (u\chi - \<u\chi\>_m) \|_2^2       \notag \\
&\le 2 \|u\chi - \<u\chi\>_m\|_2^2 + 2\|u(1-\chi)\|_2^2   \notag\\
&\le 2 \|u\chi - \<u\chi\>_m\|_2^2 + 2m(\psi <\epsilon(1+\delta)) \|u\|_\infty^2,            \label{205}
\end{align}
wherein we have used in the last line the fact that  $1- \chi =0$ wherever $\psi \ge\ep(1+\delta).$
From the Poincar\'e inequality   \eref{503f}  we find
\begin{align*}
 \|u\chi - \<u\chi\>_m\|_2^2 &\le \gamma \int_X |\n (u \chi)|^2 dm \\
 &= \gamma \int_X\( |u \n \chi+ \chi \n u|^2 dm  \notag\\
&\le 2\gamma \int_X |\n u|^2 \chi^2 dm +2 \gamma\int_X u^2 |\n \chi|^2 dm \\
&\le 2\gamma \int_X |\n u|^2 (\psi/\ep)^2 dm + 2\gamma \|u\|_\infty^2 \int_X  |\n \chi|^2 dm.
\end{align*}
Now $\n \chi = \phi'(\psi)\n \psi = -\phi'(\psi) \psi \n F$. Therefore 
$|\n \chi|^2 = \phi'(\psi)^2 \psi^2 |\n F|^2 \le \ep^{-2} \psi^2 |\n F|^2$ wherever $\psi < \ep(1+\delta)$ and 
is zero elsewhere. Therefore
\begin{align}
 \|u\chi - \<u\chi\>_m\|_2^2  \le \frac{2\gamma}{\ep^2} \int_X |\n u|^2 \psi^2 dm 
 + 2\gamma \|u\|_\infty^2 (1+\delta)^2 \int_{\psi < \ep(1+\delta)} | \n F|^2 dm \notag.
  \end{align}
  Insert this bound into \eref{205} to find
  \begin{align}
  \|u - \<u\>_m\|_2^2 &\le \frac{4\gamma}{\ep^2} \int_X |\n u|^2 \psi^2 dm 
 + 4\gamma \|u\|_\infty^2 (1+\delta)^2 \int_{\psi < \ep(1+\delta)} | \n F|^2 dm   \notag\\
 &+ 2m(\psi <\epsilon(1+\delta)) \|u\|_\infty^2 .\notag
  \end{align}
  We can now let $\delta\downarrow 0$ and use the dominated convergence 
  theorem on the second term  to find \eref{191}.
 \end{proof}

 \bigskip
\noindent
\begin{proof}[Proof of Lemma \ref{lemA1}] Insert \eref{191} into \eref{A100}  
 to find \eref{520e}. To prove  the last assertion of the lemma choose $K$ so large that $4C_K < \delta/2$.
 Then choose $\ep$ so small that the first term in $\zeta$ is also $< \delta/2$. These choices can be made
 because $\psi \in L^2(m)$ while $\psi >0$ a.e. and \eref{506} holds.
\end{proof}

 \begin{lemma} \label{lemA4}  $($Truncation of $u)$.    
  Let $\psi$ be a non-negative function satisfying $\int_X\psi^2 dm = 1$.
 Let $u \in L^2(m_\psi)$ and assume that $\int_X u\, dm_\psi =0$. 
 For all $R > 0$ define $u_R = (u\wedge R)\vee (-R)$.
  Then
\begin{align}
\|u\|_{L^2(m_\psi)}^2 \le \|u_R -\<u_R\>_{m_\psi} \|_{L^2(m_\psi)}^2 + 2 \int_{|u| > R} u^2  dm_\psi\ \  
                                                             \label{605b}
\end{align}
\end{lemma}
       \begin{proof}  Writing $m_\psi = \mu$ for ease in reading we have
 \begin{align}
\Big| \|u\|_{L^2(\mu)}^2 -&\|u_R - \< u_R\>_\mu \|_{L^2(\mu)}^2 \Big|    \notag
 = \Big| \int_X (u^2 - u_R^2) d\mu + \(\int_X u_R d\mu\)^2  \Big|.
 \end{align}
 But  $u - u_R =0$ wherever $|u| \le R$ and $|u - u_R| \le |u|$ everywhere. Therefore 
 $ \Big|\int_X (u^2 - u_R^2) d\mu\Big| = \int_{|u| >R} (u^2 -R^2) d \mu \le \int_{|u| > R} u^2 d\mu$. 
 Further, since $\int_X u d\mu = 0$, it follows that 
 $\(\int_X u_R d\mu\)^2  =  \(\int_X (u_R -u) d\mu\)^2 =  \(\int_{|u| >R} (u_R -u) d\mu\)^2 
 \le \int_{|u| >R} |u|^2 d\mu$. 
\end{proof}

\begin{lemma} \label{lemA5} Assume that the hypotheses of Theorem \ref{thmsg2}  hold.
Suppose that $\|u\|_{L^2(m_\psi)} =1$ and that $R >1$. Then 
\begin{align}
\int_{|u| \ge R} u^2 dm_\psi \le \frac{1}{\log R^2}\(B \int_X |\n u|^2 dm_\psi +  D +e^{-1}\). \label{gs216d}
\end{align}
\end{lemma}
   \begin{proof} Since $s\log_+ s \le s\log s + e^{-1}$ for all $s \ge 0$ 
   we have, in case  $\|u\|_{L^2(m_\psi)} =1$,
\begin{align}
\int_X  u^2 \log_+ u^2 dm_\psi &\le Ent_\psi(u^2) + e^{-1} 
\le   B \int_X |\n u|^2 dm_\psi +  D +e^{-1}  .            \notag
\end{align}
Therefore if $R >1$  then
\begin{align}
\log R^2 \int_{|u| \ge R} u^2 dm_\psi &\le \int_X u^2\log_+ u^2 dm_\psi
&\le B \int_X |\n u|^2 dm_\psi +  D +e^{-1} .     \notag
\end{align}
\end{proof}

\bigskip
\noindent  
\begin{proof}[Proof of Theorem \ref{thmsg2}]
Suppose that $\|u\|_{L^2(m_\psi)}^2=1$ and $\int_X u\ dm_\psi = 0$. Using first \eref{605b}, 
then  \eref{520e}  and\eref{gs216d} we find
\begin{align}
1&=\|u\|_{L^2(m_\psi)}^2   \notag \\
&\le \|u_R  -\<u_R\>_{m_\psi}\|_{L^2(m_\psi)}^2 
+2\int_{|u| > R} u^2  dm_\psi \notag \\
&\le \Big\{\gamma\(\frac{2K}{\epsilon}\)^2 \int_X |\n u|^2  dm_\psi     +\zeta R^2\Big\}    \notag\\
&\qquad \qquad\qquad+ \Big\{ \frac{2}{\log R^2}\(B \int_X |\n u|^2 dm_\psi +  D +e^{-1}\)\Big\}    \notag \\
&= \Big\{\gamma\(\frac{2K}{\epsilon}\)^2 +\frac{B}{\log R}       
              \Big\}  \int_X |\n u|^2 dm_\psi
+ \Big\{\zeta R^2    +\frac{D+e^{-1}}{\log R}       
         \Big\}                                                       \label{gs280b}
\end{align}
It suffices to show that the second expression in braces in \eref{gs280b} can be made 
less than $1/2$ by choosing $R, K$ and $\ep$ suitably. 

 We may choose $R>1$ so that
\beq
\frac{(D+e^{-1})}{\log R}  = 1/6.                                   \label{gs281} 
\eeq
From \eref{422e1} we find $\zeta R^2 =2K^2\( 2\gamma B_\ep +A_\ep \)R^2 +4C_K R^2$. 
Choose $K \ge 1$ so large that $4C_K R^2 \le 1/6$. 
Then choose $\ep$ so small that 
\begin{align}
2K^2\( 2\gamma B_\ep    +A_\ep \)R^2 \le 1/6.   \notag
 \end{align}
 Then $\zeta R^2 \le 1/3$. From the definitions of $C_K, A_\ep$ and $B_\ep$ it's clear 
 that such $\ep$ and $K$  exist. Since $R^2 = e^{12(D+ e^{-1})}$, \eref{A360} is satisfied.
Inserting these bounds into \eref{gs280b} we find
\begin{align}
1 \le & \Big\{\gamma\(\frac{2K}{\epsilon}\)^2 + \frac{B}{6(D + e^{-1})} \Big\} \int_X |\n u|^2 dm_\psi   
                         + 1/2        \notag
\end{align}
and therefore 
\begin{align}
1 \le  \Big\{2\gamma\(\frac{2K}{\epsilon}\)^2 + \frac{B}{3(D + e^{-1})} \Big\}  \int_X |\n u|^2 dm_\psi. \notag
\end{align}
Thus we have a spectral gap and we may take 
\begin{align}
\gamma_1 =  \Big\{2\gamma\(\frac{2K}{\epsilon}\)^2 + \frac{B}{3(D + e^{-1})} \Big\}.  \label{A340}
\end{align}
in \eref{A363}.   
But $3(D + e^{-1}) >1$ because  $D \ge 0$. The second term in
\eref{A340} is therefore at most $B$. \eref{A365b} now follows. 
\end{proof}

\subsection{Bounds on Aida's  spectral gap} \label{secbasg}

\subsubsection{The distribution of $\psi$} \label{secdist}

The three quantities that determine most of the estimates needed in Aida's bound on the spectral gap
are given in Notation    \ref{notA3}.
We will give bounds on these three quantities in terms
of the given data $c,\nu, \kappa, \|e^{-V}\|_\nu$ and $\|e^{V}\|_\kappa$. There are parameters at  our disposal
whose choice of values may be significant in some applications. But we will use values which 
keep our bounds simple and  serve the purposes of this paper. 
      
      In order to apply Theorem  \ref{thmDLSI3} we need to choose $a$ and $\sigma$ satisfying 
      \eref{gs803b}. We will take  $a = \sigma = 2c_\nu b_\ka$, which will simplify some formulas.

 \begin{theorem} \label{thmdist}  Let
\begin{align}
a = \sigma = 2c_\nu b_\ka, \ \ s_1 =  (b_\ka - (1/2))^{-1} \ \ \text{and}\ \ \ 
\alpha_1  &= a + (c \log 3)/b_\ka             \label{Adist0}
 \end{align}
Then 
\begin{align}
A_\ep &\le (\ep M ^{\alpha_1})^{s_1},         \label{Adist1c} \\
B_\ep&\le A_\ep^{1/2} \Big\| V +  \log \|e^{-V}\|_{2c}\Big\|_2\ \ \ \ \ \ \ \text{and} \label{Adist2c} \\
 B_\ep &\le A_\ep^{1/2} \ka^{-1} M^\ka                          \label{Adist2c3}\\
C_K &\le \(M^{ \frac{\log 3}{(2c)^{-1} - \nu^{-1}} }/K^2\)^{a_\nu/(2-a_\nu)} . \label{Adist3c}
\end{align}
\end{theorem}
    \begin{proof} [Proof of \eref{Adist1c}] 
If  $\epsilon = e^{-b}$ then $m(\psi \le \epsilon)  = m(F \ge b)$. 
From  Chebyshev's inequality we find  $m(F\ge b)  e^{sb} \le \int_X e^{sF} dm$  for all $s \ge 0$ 
and  therefore, by \eref{W710b}, 
 we have
\begin{align}
m(\psi \le \epsilon) &\le e^{-sb}  \int e^{sF} dm   
= \epsilon^s  \|\psi^{-1}\|_s^s           \notag \\
&\le  \epsilon^s M^{s\(\ell(a) + \ell(\sigma) +\sigma\)}\ \ \  \text{for}\ \ s \in (0,s_0)  \label{Adist1d} 
\end{align}
whenever  $a > c_\nu b_\ka$, $\sigma > c_\nu b_\ka$ and $a = (2s^{-1} +1) c_\nu$, as in \eref{W852b}. 
The choice for $a$ and $s_1$ given in \eref{Adist0} is consistent with this link between 
$a$ and $s_1$ because $(2s_1^{-1} +1) c_\nu =(2b_\ka - 1 +1)c_\nu = 2 b_\ka c_\nu =a$. 
Put $t =a$  in \eref{W852} to find 
\begin{align}
\ell(a) =(c\log 3)/(2b_\ka)          \label{Adist5}
\end{align}
when $a = 2c_\nu b_\ka$.  Therefore 
       $\ell(a) + \ell(\sigma) + \sigma =  (c\log 3)/b_\ka + 2c_\nu b_\ka= \alpha_1$,
 which, inserted into \eref{Adist1d}  gives \eref{Adist1c}.
\end{proof}

\bigskip
\noindent
\begin{proof}[Proof of \eref{Adist2c} - \eref{Adist2c3}]
For the proof of \eref{Adist2c} use the Federbush semi-boundedness theorem
\eref{s1} to find  $- \l_0 \le \log \|e^{-V}\|_{2c}$ and insert this into \eref{W43a}. We get
\begin{align*}
\int_{F \ge b} |\n F|^2 dm &\le \int_{F \ge b} (V +   \log \|e^{-V}\|_{2c}) dm \\
&\le m(F \ge b)^{1/2}  \Big\| V +  \log \|e^{-V}\|_{2c}\Big\|_2.
\end{align*}
Choose $b $ so that  $\epsilon = e^{-b}$ again. Since $\{ F \ge b\} = \{ \psi \le \ep\}$, 
\eref{Adist2c} follows. 

The inequality \eref{W43a} together with Young's inequality give 
\begin{align}
\int_{F \ge b} |\n F|^2 dm &\le \int_{F \ge b} (V-\l_0)  dm   \notag \\
&\le {\color{red} \ka^{-1}}Ent_m(\chi_{F\ge b}) + m(F \ge b) \ka^{-1}\log\int_X e^{\ka(V-\l_0)} dm  \label{Adist2c2}
\end{align}
Since $\chi_F \log \chi_F = 0$ we have
 $Ent_m(\chi_{F\ge b})  =  -m(F \ge b)\log(m(F \ge b)) \le m(F \ge b)^{1/2} $
because  $-t^{1/2} \log t \le 2/ e <1$ for $0 \le t <1$.  
The second term in \eref{Adist2c2} is $ m(F \ge b)\log\|e^{V-\l_0}\|_\ka$,
which, in view of \eref{s3} is at most   $m(F \ge b)^{1/2} \log M$. Choose $b$ again 
so that $e^{-b} = \ep$ to find 
\begin{align}
B_\ep&\le A_\ep^{1/2} (\ka^{-1}+ \log M).          \label{Adist2c1}
\end{align}
 Since $M^\ka \ge 1$, we have $( \ka^{-1}  + \log M)  =\ka^{-1}( 1 + \log M^\ka)  
 \le  \ka^{-1} M^\ka$, from which \eref{Adist2c3} follows.
\end{proof}

The proof of \eref{Adist3c} depends on the following lemma, which 
 implements  a standard method of getting $L^p$ bounds from hyperboundedness.

\begin{lemma} \label{lemdistK} Let $K >0$ and  $2 <p <p_0$. Then
\begin{align}
C_K &\le  \|e^{\l_0 -V}\|_\nu^{p \tau(p)}/K^{p-2},                      \label{Adist3b}
\end{align}
\end{lemma}
     \begin{proof}
       For any $K>0$  and $p >2$ we have  $\psi^2 \le \psi^p/K^{p-2}$  wherever 
 $\psi \ge K$. 
 Apply the hyperboundedness inequality  \eref{L289}  with $q=2$ to find
 \begin{align*}
 \int_{\psi\ge K} \psi^2 dm 
& \le \int_{\psi \ge K} K^{2-p} \psi^p dm \\
 &\le K^{2-p} \int_X \psi^p dm \\
 &= K^{2-p} \|e^{t\l_0}e^{-tH} \psi\|_p^p \\
 &\le K^{2-p} \(e^{t\l_0}\|e^{-tH}\|_{2\to p} \|\psi\|_2\)^p\ \ \\
 &\le   K^{2-p} \(e^{t\l_0} \| e^{-V}\|_\nu^t\)^p   \qquad \qquad        \ \ \forall t \ge  \tau(p) \\ 
 &=  K^{2-p}  \|e^{\l_0-V}\|_\nu^{tp}   \qquad\qquad \qquad    \ \ \ \forall t \ge  \tau(p).  
 \end{align*}
 Now  choose $t=\tau(p)$ to find \eref{Adist3b}.
\end{proof}

\bigskip
\noindent
\begin{proof}[Proof of \eref{Adist3c}] We will choose a special value of $p$ in \eref{Adist3b} 
that makes the dependence of the exponents on $\nu$ simple and explicit.
 Define $p$ by
\beq 
 p^{-1} = (1/4) + (1/2) p_0^{-1}. \label{A424c}
 \eeq
      Since $p_0 >2$ we have $(1/2)p_0^{-1} < 1/4$ and therefore
 \begin{align}
p_0^{-1}  = (1/2)p_0^{-1} + (1/2)p_0^{-1} <  1/4 +(1/2)p_0^{-1}  < 1/2. \notag
\end{align}
Hence $p_0> p >2$.   To evaluate $\tau(p)$, observe that, in view of \eref{L313b}, one has
\begin{align*}
q_0^{-1} - p^{-1}  &= 1-p_0^{-1} - \(1/4 +(1/2)p_0^{-1} \) = 3\(1/4 -(1/2)p_0^{-1}\) \\
&= 3(p^{-1} - p_0^{-1}\).
\end{align*}
It follows from \eref{L505} that  $\tau(p) =  \frac{c}{2a_\nu} \log 3 .$
From the expression \eref{L313qi} for $p_0^{-1}$ we find
 $p^{-1} = (1/4)(2- a_\nu)$. Therefore
\begin{align}
p\tau(p) = \frac{2c\log 3}{(2-a_\nu)a_{\nu}},
\end{align}
while $p-2 = \frac{4}{2-a_\nu} - 2 = \frac{2a_\nu}{2-a_\nu}$.
Inserting these values into \eref{Adist3b} we find
\begin{align}
C_K &\le  \|e^{\l_0 -V}\|_\nu^{ \frac{2c\log 3}{(2-a_\nu)a_{\nu}}}/K^{\frac{2a_\nu}{2-a_\nu}} .\label{A425}
\end{align}
It will be useful to write this in terms of $K^2$.
From \eref{L341d} we see that
\begin{align}
\frac{2c}{a_\nu^2} =  \frac{1}{(2c)^{-1} - \nu^{-1}} .       \notag
\end{align}
So we may rewrite 
\begin{align}
C_K \le \(\| e^{\l_0 -V}\|_\nu^{ \frac{\log 3}{(2c)^{-1} - \nu^{-1}} }/K^2\)^{a_\nu/(2-a_\nu)} .
\end{align}
Finally, the bound \eref{s4}  gives \eref{Adist3c}.
\end{proof}

\begin{remark} \label{remCK} {\rm  The choice of $p$ given by  \eref{A424c}   
simplifies $\tau(p)$ in \eref{Adist3b} and gives the simple form \eref{Adist3c} of the bound on $C_K$. 
 But  one can also simplify $\tau(p)$ by choosing $p$ so that 
 $(1/2) - p^{-1} = y^{-1}((1/2) - p_0^{-1})$,
for some $y >1$. (This reduces to \eref{A424c} when $y=2$.) 
 In this case one finds
$\tau(p) = (c/(2a_\nu)) \log x$, where $x=(y+1)/(y-1) $. The resulting bound on $C_K$ 
 is given by
\begin{align}
C_K  &\le  \(\|e^{\l_0 -V}\|_\nu^{\frac{cy\log x}{2a_\nu^{2}} }/K\)^{2 \frac{a_\nu}{y-a_\nu}} \notag \\
&=  \(\|e^{\l_0 -V}\|_\nu^{\frac{cy\log x}{a_\nu^{2}} }/K^2\)^{ \frac{a_\nu}{y-a_\nu}} .      \label{A471}
\end{align}
In some applications it might be useful to choose $y$ large. But in this paper the estimate
\eref{Adist3c} serves our purposes.
}
\end{remark}

\begin{remark} \label{remA3.3,4} {\rm   The estimates of $A_\ep, B_\ep$ and $C_K$ that we gave in 
Theorem \ref{thmdist}  depend on $\|e^{-V}\|_\nu$ and $\|e^V\|_\ka$.
It may be desirable for future
applications to avoid use of $\|e^V\|_\ka$, because a bound on this, although almost necessary
for bounds on $\|\psi^{-s}\|_{L^1(m)}$, as we see in Corollary \ref{corvlps2}, do not seem to be anywhere near necessary
for establishing a defective logarithmic Sobolev inequality, as examples show. It is  possible, however,
to get bounds on $A_\ep, B_\ep$ and $C_K$ just in terms of   $c,\nu,\|e^{-V}\|_\nu$ and $\|V\|_p$ for  
any $p > 1$.
The key steps in one such 
 procedure have been carried out by Aida in \cite[Lemma 3.3, Part (4)]{Aida2001}.
}
\end{remark}

\subsubsection{Aida's spectral gap}   \label{secAsg}

 Theorem \ref{thmsg2}, gives a bound,  \eref{A365b},  on the coefficient $\gamma_1$
in Poincar\'e's inequality \eref{A363} for the ground state measure $m_\psi$.
        The bound depends on the choice of a region $\{ \ep < \psi < K\}$, outside of which  the contributions
        to the Poincar\'e inequality are well controlled  by the energy.
        $\ep$ and $K$ must be chosen so as to satisfy the inequality \eref{A360}.
 In this section we will use the bounds on the distribution of $\psi$,
  derived in Theorem \ref{thmdist}, to make a choice of  $\ep$ and $K$ satisfying \eref{A360}, 
 from which we can derive 
  a quantitative bound on the   Poincar\'e coefficient $\gamma_1$ in terms
of the given data $c,\nu, \kappa$ and $M\equiv \|e^{-V}\|_\nu \|e^V\|_\kappa$.

\begin{theorem} \label{thmspecbd} Under the  hypotheses of Theorem \ref{thmM} there exists
a number $\gamma_1$ such that 
\begin{align}
 \int_X\(u^2 - \<u\>_{m_\psi}^2\)dm_\psi  
             &\le \gamma_1 \int_X|\n u|^2 dm_\psi.    \label{A363a}
 \end{align} 
 $\gamma_1$ may be chosen so as to satisfy the bound
 \begin{align}
 \gamma_1 \le d_1 M^{e_1}                    \label{A363b}
 \end{align}
 for constants $d_1, e_1$ 
 depending only on $c, \nu, \kappa$.
 \end{theorem}
\begin{proof}  With the goal of implementing the procedure of Theorem \ref{thmsg2}, we first 
choose $K$ so large that $4C_K e^{12(D + e^{-1})} \le 1/6$.
 For this it suffices by \eref{Adist3c}  to take  $K$ so that 
\begin{align}
4 \(M^{ \frac{\log 3}{(2c)^{-1} - \nu^{-1}} }/K^2\)^{a_\nu/(2-a_\nu)} e^{12(D + e^{-1})} \le 1/6 . \label{A480b}
\end{align}
Define $K$ by equality  in \eref{A480b}. Then 
\begin{align}
K^2 =    M^{ \frac{\log 3}{(2c)^{-1} - \nu^{-1}} } 
                     \(24  e^{12(D + e^{-1})} \)^{\frac{2- a_\nu}{a_\nu}}.       \label{A480}
\end{align}
 Second, we  choose $\ep$ so small that
\begin{align}
2K^2\( 2\gamma B_\ep +A_\ep \)  e^{ 12(D + e^{-1})}  &\le 1/6.
\end{align}
That is,      $\( 2\gamma B_\ep +A_\ep \) \le \frac{ e^{ -12(D + e^{-1})}}{12 K^2}.$		
	For this it suffices, by \eref{Adist2c3}, to take $\ep$ such that
\begin{align}
2\gamma A_\ep^{1/2}\ka^{-1}M^\ka+ A_\ep 
               \le \frac{ e^{ -12(D + e^{-1})}}{12 K^2}.              \label{A482}
\end{align}
Since, by its definition \eref{sg10}, we have $A_\ep \le 1$, we also have $A_\ep \le A_\ep^{1/2}$.
Thus it suffices to choose $\ep$ such that
\begin{align}
 A_\ep^{1/2} 
     \le   \frac{ e^{ -12(D + e^{-1})}}{12 K^2\(1 + 2\gamma \ka^{-1} M^\ka\)}. \label{A483}
\end{align} 
From  \eref{Adist1c}, we see that \eref{A483} will hold if
\begin{align}
(\ep M^{\alpha_1})^{s_1/2}    \le   \frac{ e^{ -12(D + e^{-1})}}{12 K^2\(1 + 2\gamma\ka^{-1}  M^\ka\)}. \label{A484}
\end{align}
Define $\ep$ by equality in \eref{A484}. 
 Then   
\begin{align}
\ep^2 = \Big\{ \frac{ e^{ -12(D + e^{-1})}}{12 K^2\(1 + 2\gamma\ka^{-1}  M^\ka\)}
                  \Big\}^{4/s_1}    M^{-2\alpha_1} .  \label{A485}
\end{align} 
The values of $K$ and $\ep$ defined in \eref{A480} and \eref{A485} satisfy \eref{A360}. 
We may therefore use them to bound $\gamma_1$ by   \eref{A365b}.            
We find
 \begin{align}
&K^2/\ep^2 
 = K^2 \Big\{\frac{12 K^2\(1 + 2\gamma\ka^{-1}  M^\ka\)}
 { e^{ -12(D + e^{-1})}}\Big\}^{4/s_1}  M^{2\alpha_1}    \notag \\
 &=\Big\{K^{2\{1+(4/s_1)\}}\(12e^{ 12(D + e^{-1})}\)^{4/s_1}\Big\} \(1 + 2\gamma \ka^{-1} M^\ka\)^{4/s_1}
                     M^{2\alpha_1}.  \label{A488a}
\end{align}
In view of \eref{A480} the factor in braces is
\begin{align}
&\Big\{K^{2\{1+(4/s_1)\}}\(12e^{ 12(D + e^{-1})}\)^{4/s_1}\Big\} \notag \\
& = \Big\{M^{ \frac{\log 3}{(2c)^{-1} - \nu^{-1}} } 
                     \(24  e^{12(D + e^{-1})} \)^{\frac{2- a_\nu}{a_\nu}}\Big\}^{1 +(4/s_1)}   
                        \(12e^{ 12(D + e^{-1})}\)^{4/s_1}    \notag\\
 & =M^{ \frac{\log 3}{(2c)^{-1} - \nu^{-1}} (1 +(4/s_1))} 
   \(24 e^{12(D+e^{-1})}\)^{\(\frac{2- a_\nu}{a_\nu}(1 +(4/s_1))+ 4/s_1\) } 2^{-4/s_1}. \label{A488b}
\end{align}
From \eref{Adist0} we see that $4/s_1= 4b_\ka - 2$, and therefore   
$$
\frac{2- a_\nu}{a_\nu}(1 +(4/s_1))+ 4/s_1 = \frac{2(4b_\ka -1)}{a_\nu} -1.
$$
Inserting this into \eref{A488b}, we find from \eref{A488a}
\begin{align}
K^2/\ep^2  = M^{\beta_1} e^{\beta_2 D}  
           \(1 + 2\gamma \ka^{-1}  M^\ka\)^{\beta_3} \beta_4,                      \label{A488c}
\end{align}
where
\begin{align}
\beta_ 1 & =   \frac{\log 3}{(2c)^{-1} - \nu^{-1}} (4b_\ka -1)  + 2\alpha_1      \label{A489a}\\
\beta_2  & =12\(\frac{2(4b_\ka -1)}{a_\nu} -1\) \label{A489b}\\
\beta_3 &= 4b_\ka -2         \label{A489c}\\
 \beta_4 &= \(24 e^{12/e}\)^{\(\frac{2(4b_\ka -1)}{a_\nu} -1\)} 2^{2- b_\kappa}, \label{A489d}
\end{align}
and $a_\nu =\sqrt{1 - (2c/\nu)} , b_\ka =  \sqrt{1 + (2c/\ka)} $, by \eref{L341d} and  \eref{s98f},  
 and $\alpha_1$ is given in \eref{Adist0}. 
All of the four constants $\beta_j$ depend only on $c, \nu$ and $\ka$ and are non-negative.

       In   Theorem \ref{thmsg2}, the constants $B$ and $D$ are arbitrary. 
 To apply our bounds from Theorem \ref{thmDLSI3}  we 
 use the form of the defective logarithmic Sobolev inequality given in 
  \eref{gs807b}. Thus we take   
  \beq
  B= 2a \ \ \ \text{and} \ \ D = 2 \log M^{a + \ell(a) + \sigma + \ell(\sigma)}. \label{A490}
  \eeq
For our choices, $a = \sigma = 2c_\nu b_\ka$, we see from \eref{Adist5} that 
\begin{align}
e^D = M^{2(2a + (c\log 3)/b_\ka)}       \label{A490a} 
\end{align}
and therefore $ e^{\beta_2 D} = M^{2\beta_2(2a + (c\log 3)/b_\ka)}$. 
Combining the first two factors in \eref{A488c} we find
\begin{align}
K^2/\ep^2  = M^{\beta_5}  
           \(1 + 2\gamma \ka^{-1} M^\ka\)^{\beta_3} \beta_4,  \label{A488d}
\end{align}
where
\begin{align}
\beta_5 = \beta_1 + 2\beta_2(2a + (c\log 3)/b_\ka).     \label{A489e}
\end{align}
Our  assumed logarithmic Sobolev inequality \eref{mt1} implies that the Poincar\'e inequality 
 \eref{503f} in Aida's hypothesis holds in our case with $\gamma = c$. See \cite[Theorem 2.5]{G1993}
 or \cite[Proposition 5.1.3]{BGL}   for a proof of this.

 Thus the bound \eref{A365b} yields  in our case 
\begin{align}
\gamma_1 \le 2a + 8c  M^{\beta_5}  
           \(1 + 2c\ka^{-1} M^\ka\)^{\beta_3} \beta_4,          \label{A489f}
\end{align}  
with $a = 2c_\nu b_\ka$. To reach the simple looking form \eref{A363b} we can use the overestimate
$1 \le M$, as in \eref{M3}, from which follows that  $1+ 2c\ka^{-1} M^\ka \le( 1+ 2c/\ka) M^{\ka}$
 and  $2a \le 2a M^{\beta_5 + \beta_3}$. Inserting these two bounds into \eref{A489f} we find
\begin{align}
\gamma_1 \le  d_1M^{e_1},        \label{A489h}
\end{align}
  where 
 \begin{align}
 d_1 = 2a + 8c (1+2c/\ka)^{\beta_3} \beta_4, \ \ \ \ \ \ \ e_1 =\beta_5+ \ka\beta_3      \label{A489g}
 \end{align}
 and $a = 2c_\nu b_\ka$ as usual.   This proves Theorem \ref{thmspecbd}.
\end{proof}

\subsection{Tightening: Proof of the main theorem}   \label{sectightening}

If the generator, $H$,  of a hyperbounded semigroup  has a spectral gap at the bottom of its spectrum
then the semigroup is in fact hypercontractive.
 This was first proven by 
J. Glimm, \cite[Lemma 5.1]{Glimm68},  and  later amplified by I. E. Segal \cite[Section 1]{Seg70}.
 In view of the equivalence of
hyperboundedness with logarithmic  Sobolev inequalities, \cite{G1}, one can restate this 
at an infinitesimal level: If a Dirichlet form satisfies both a defective logarithmic Sobolev inequality
and a Poincar\'e inequality  then it also satisfies  a logarithmic Sobolev inequality (without defect).
The initial form of this theorem was given by O. Rothaus, \cite{Rot5}, wherein he proved a key lemma for this
theorem, \cite[Lemma 9]{Rot5},  and applied it then to a specific geometric circumstance in the context of
isoperimetric inequalities, \cite[Theorem 10]{Rot5}, to remove the defect.   
        Deuschel and Stroock, \cite{DS1989}, gave another  proof of Rothaus' theorem and  
  Carlen and Loss, \cite{Ca2004}, gave another different proof. 
  We will use  the form of this theorem given in \cite[Proposition 5.1.3]{BGL}, which we quote here.

\begin{proposition} \label{propbgl} {\rm  (\cite[Proposition 5.1.3]{BGL}). } 
Suppose that $\mu$ is a probability measure on a Riemannian manifold. If
\begin{align}
Ent_\mu (f^2) \le 2C \int |\n f|^2 d\mu + D \int f^2 d\mu        \label{sls5}
\end{align}
and
\begin{align}
Var_\mu (f) \le C' \int |\n f|^2 d\mu                          \label{sls6}
\end{align}
then
\begin{align}
Ent_\mu (f^2) \le 2\(C +C'((D/2) + 1)\) \int |\n f|^2 d\mu.  \label{sls7}
\end{align}
\end{proposition}

We will apply this proposition to the defective logarithmic Sobolev inequality derived in Theorem \ref{thmDLSI3}
in combination with Aida's spectral gap estimate derived in Section \ref{secbasg}. As in both of those
inequalities, there  are  parameters that can be chosen according to needs in applications. 
We will use the  choices we made before to arrive at bounds of the simple form described 
in Theorem \ref{thmM}.
 
\bigskip
\noindent
\begin{proof}[Proof of Theorem \ref{thmM}] 
Items a. and b. in Theorem \ref{thmM} have been proved in Section \ref{sechyperp}.

For the proof of item c. we take $\mu = m_\psi$ in Proposition \ref{propbgl}.
Choose $a = \sigma = 2c_\nu b_\ka$ as we did  in \eref{Adist0}  (in Section \ref{secdist})
and take  $C= a$ in \eref{sls5}.  Then \eref{sls5} holds with 
\begin{align}
D = \log M^{2(2a + (c\log 3)/b_\ka)}   
\end{align}
by  \eref{A490a} and Section \ref{secAsg}.
In \eref{sls6}   we may, by Theorem \ref{thmspecbd},  take $C' = \gamma_1$, 
 where $\gamma_1 \le d_1 M^{e_1}$ and $d_1, e_1$ are given by \eref{A489g}. 
Proposition \ref{propbgl} then assures that
\begin{align}
Ent_{m_\psi}(f^2) \le 2c_1 \int_X |\n f|^2 dm_\psi
\end{align}
with
\begin{align}
c_1 \le a + \gamma_1( 1 +  \log M^{(2a + (c\log 3)/b_\ka)}).          \label{T10}
\end{align}

Item d. of Theorem  \ref{thmM} follows from item c. and the Rothaus-Simon theorem
 \cite{Rot1}, \cite{Simon1976}. A direct proof of the Rothaus-Simon theorem may be found in
  \cite[Theorem 2.5 ]{G1993}  or \cite[Proposition 5.1.3]{BGL}.

In item e. of Theorem \ref{thmM} the form of the bound on $c_1$ can be derived from \eref{T10}
by overestimating again,  
using $M \ge 1$, to find $1 +  \log M^{(2a + (c\log 3)/b_\ka)} \le M^{(2a + (c\log 3)/b_\ka)}$ and
therefore
\begin{align}
c_1& \le a +  d_1M^{e_1} M^{(2a + (c\log 3)/b_\ka)}     \notag \\
&\le (a+d_1) M^{e_1} M^{(2a + (c\log 3)/b_\ka)}. \notag\\
&=\alpha M^\beta
\end{align}
where $ \alpha = a + d_1$ and $ \beta = e_1 + (2a + (c\log 3)/b_\ka)$. The constants $d_1$ and $e_1$
are defined in \eref{A489g} and depend only on $c,\nu$ and $\kappa$.
This proves item e. of Theorem \ref{thmM} .
\end{proof}

\begin{remark} \label{remsg2} {\rm  The spectral gap for $m_\psi$ listed in 
item d. of Theorem \ref{thmM} is derived 
from the logarithmic Sobolev inequality \eref{mt5} described in item c..  
But our procedure for deriving \eref{mt5} includes deriving first the Poincar\'e inequality \eref{A363a} 
for $m_\psi$. The constant $\gamma_1$ in \eref{A363a} is much smaller than the Sobolev constant 
$c_1$ as one can see from \eref{T10}.  Therefore  we have actually a  
smaller Poincar\'e constant than that derived from $c_1$. 
 In particular the spectral gap is at least $d_1^{-1}M^{-e_1}$.
}
\end{remark}

\section{Examples and Applications}

\subsection{Consecutive ground state transforms}  \label{secgst}

If the potential in a Schr\"odinger operator is a sum of two potentials then the ground state 
transformation may factor into two ground state transformations, one for each potential, 
in the following sense.

\begin{lemma} \label{lemcgst} 
$($Consecutive  ground state transformations$)$. 
Suppose that $m$ is a smooth measure on a Riemannian manifold and that
$V_1$ and $V_2$ are two potentials. 

Assume that $H\equiv \n^*\n + V_1 + V_2$ has a unique ground state $\psi \in L^2(m)$. 
Denote by $U:L^2(m_\psi) \to L^2(m)$ the ground state  transformation.

Assume further that 
$H_1 \equiv \n^*\n +V_1$ has a unique ground state $\psi_1 \in L^2(m)$. Let $U_1:L^2(m_{\psi_1}) \to L^2(m)$ be the ground state transformation for  $H_1$.
 Further, suppose that the Schrodinger operator $\n_{m_{\psi_1}}^* \n + V_2$ has a unique 
ground state $\psi_2 \in L^2(m_{\psi_1})$. Denote by 
$U_2: L^2(\psi_2^2 m_{\psi_1}) \to L^2(m_{\psi_1})$ the ground state transformation. 
Then
\begin{align}
\psi &= \psi_1 \psi_2.          \label{E5}\\
m_\psi &= \psi_2^2 m_{\psi_1}.         \label{E6} \\ 
U &= U_1U_2.       \label{E7}
\end{align}
That is, the first line below factors as the second line.
\begin{align}
&L^2(m)  \xleftarrow{\ \ \ \ \ \ \ \ U\ \ \ \ \ \ \ \ \  } L^2( m_{\psi})       \label{E8}\\
 &L^2(m)\xleftarrow{U_1} L^2(m_{\psi_1})\xleftarrow{U_2}L^2(\psi_2^2 m_{\psi_1})  \label{E9}
\end{align}
\end{lemma}
     \begin{proof} Let us write $m_1 = \psi_1^2 m$ and $m_2 = \psi_2^2m_1 = \psi_2^2\psi_1^2 m$.
     If $\l_1$ is the bottom of the spectrum of $\n^*\n + V_1$ then
  \begin{align}
  U_1^{-1}(\n^*\n + V_1 - \l_1) U_1 = (\n)_{m_1}^*\n              \notag
  \end{align}
  by the definition of the ground state transformation for $\n^*\n + V_1$.
  Since $U_1$ is a multiplication operator it commutes with multiplication by $V_2$.
Therefore
\begin{align}
U_1^{-1}(\n^*\n + V_1+V_2- \l_1) U_1 = (\n)_{m_1}^*\n + V_2.        \label{E20}
\end{align}
By the definition of $U_2$ we have
\begin{align}
U_2^{-1}\((\n)^*_{m_1} \n +V_2-\l_2\)U_2 = (\n)_{m_2}^* \n. \label{E12}
\end{align}
where $\l_2 = \inf$ spectrum $(\n)^*_{m_1} \n +V_2$.
Insert \eref{E20} into \eref{E12} to find
\begin{align}
(U_1U_2)^{-1} \(\n^*\n + (V_1 + V_2) - (\l_1 + \l_2)\) U_1U_2  =   (\n)_{m_2}^* \n.   \notag
\end{align}
Apply this identity to the function identically one, which is a unit vector in $L^2(m_2)$. 
The right hand side is zero while $U_1U_2 1 = \psi_1\psi_2$. Therefore
\begin{align}
 \(\n^*\n + (V_1 + V_2) - (\l_1 + \l_2)\) \psi_1\psi_2 =0.             \label{E15}
\end{align}
Since  $\n^*\n + (V_1 + V_2)$ has a unique ground state $\psi$, and $\psi_1\psi_2$ is a positive normalized
function in $L^2(m)$ satisfying \eref{E15},  it follows that $\psi_1 \psi_2 = \psi$ and $\l_1 + \l_2 = \l$ and 
$U_1U_2 f= \psi_1\psi_2 f = \psi f = Uf$.  This proves \eref{E5} - \eref{E9}.
\end{proof}
      \begin{remark} {\rm  In case $X = \R^n$ and $m$ is Lebesgue measure 
      then $\n^*\n = -\Delta$ and we have the usual Schr\"odinger operator in the hypothesis of this lemma.
}
\end{remark}

\begin{remark} \label{remKS}   {\rm The use of consecutive ground state transforms
is implicit in \cite[Theorem 1.4]{KS87}.
}
\end{remark}

\bigskip

\begin{example} \label{expos} 
 Let $m$ be a smooth measure $($not necessarily finite$)$ on a Riemannian manifold $X$.
Let $V$ be a measurable potential and suppose that  $ V= V_0 + V_1$ with $V_1 \ge 0$. Assume that 
the Schr\"odinger operator $\n^*\n + V_0$ has a unique (positive) ground state  $\psi_0$ whose 
ground state measure $m_{\psi_0}$ satisfies a LSI. Assume also that 
\begin{align}
\int_Xe^{\ka V_1} dm_{\psi_0} < \infty\ \ \ \text{for some}\ \ \ka >0. \label{sls20}
\end{align}
Then the Schr\"odinger operator $\n^*\n +V$ has a unique (positive) ground state $\psi$ and the
ground state measure $m_\psi$ satisfies a LSI.
\end{example}
    \begin{proof}  Since $V_1 \ge 0$ and \eref{sls20} holds, the condition  \eref{mt2} holds for $V_1$ and 
    the  measure $m_{\psi_0}$.
    We can apply Theorem \ref{thmM} 
    to find a ground state $\psi_1$ for 
    the Schr\"odinger operator  $\n^*_{m_{\psi_0}} \n + V_1$ and the  ground state 
    measure $\psi_1^2 dm_{\psi_0}$
    satisfies a LSI.  By Lemma \ref{lemcgst}      
     the function     $\psi \equiv \psi_1 \psi_0$ is the ground state of the 
     Schr\"odinger operator $\n^*\n + V$. Moreover
    $m_\psi = \psi^2 dm= \psi_1^2 \psi_0^2 dm = \psi_1^2 dm_{\psi_0}$. Therefore $m_\psi$ satisfies a LSI.
    \end{proof}

\subsection{Gaussian precision} \label{secgp}

The two 
quadratic equations  \eref{W850} and \eref{L313} determine intervals of 
Lebesgue indices for which various moment bounds and hypercontractive
bounds hold. \eref{W850} is key in case the potential is positive and \eref{L313} is key in case the potential
is negative.
In this section we will show that the intervals of validity of these bounds are exact for 
Gaussians and therefore the intervals determined by these peculiar quadratic equations 
are not  just artifacts  of the proof.

\subsubsection{Negative potentials} \label{secnp}

Corollary \ref{corhb2} 
 shows that the semigroup $e^{-t(\n^*\n +V)}$ is bounded from $L^q(m)$ to $L^p(m)$
if the Dirichlet form for $m$ satisfies a logarithmic Sobolev inequality and if $t, q, p$ 
and $\|e^{-V}\|_{L^\nu(m)}$ are suitably related. There is, in addition, a surprising restriction on the 
allowed range 
 of $q$ and $p$, unlike in the usual hyperboundedness theorems. The restriction is determined by the
quadratic equation \eref{L313m}, whose two roots $q_0, p_0$ are conjugate indices. Boundedness,
$\|e^{-t(\n^*\n + V)}\|_{L^q(m)\to L^p(m)} < \infty$, is  
 assured by the corollary   
for large $t$, 
 but only in case $q_0 \le q \le p \le p_0$. In particular, the corollary shows that
  $e^{-t(\n^*\n +V)}$ is a strongly continuous semigroup
in $L^p(m)$ if $q_0 \le p \le p_0$.  
We will give an example  in which the latter fails
 if $p \notin [q_0,p_0]$. 
 
 Let
\begin{align}  
m = \gamma = \pi^{-1/2} e^{-x^2} dx, \ \ V(x) = - ax^2,\ \ a >0, \ \ \text{and}\ \  H = \n^*\n + V \label{gp3a}
\end{align}
Then \eref{mt1} holds with $ c = 1/2$,  \cite{G1}.

\begin{theorem} \label{thmgp1}   Let $\nu > 1$. Define $q_0$ and $p_0$ by \eref{L313q} with $c = 1/2$. 
Let $p_1 > p_0$. Then there exists  a real number $a >0$  
 such that
\begin{align}
&\int_\R e^{-\nu V} d\gamma  < \infty\ \ \ \ \ \ \ \text{and}       \label{gp4}     \\
&e^{-tH} g  \notin L^{p_1}(\gamma)\ \  \text{for some}\ \ 
                               g \in L^{p_1}(\gamma)  \text{ and some}\  t >0. \label{gp5}
\end{align}
In particular  $e^{-tH}$ does not operate as a strongly continuous semigroup in $L^{p_1}(\gamma)$.
\end{theorem}
For the proof, we will first  show,  in the next lemma, that the family of functions 
 \begin{align}
 f(t, x) =e^{b(t) + s(t) x^2} ,\ \ t \ge0      \label{gp3}
 \end{align}
  includes  functions of the form $e^{-tH} g$.  
  
  \begin{lemma} \label{lemnp1}  Define $\gamma$ as in \eref{gp3a}. 
  Let $H= \n^*\n -ax^2$. Then 
\begin{align}
a. \ & \n^*\n g (x)= - g''(x) +2x g'(x),\ \ \ g \in C^{\infty}(\R)\cap \D(\n^*\n).      \label{gp6}\\
b.\  & (H f)(t,x) = \( -2s + \{-4s^2 +4s  -a\}x^2\) f(t,x)\ \     \text{if}\ \ s(t) < 1/2.   \label{gp7}\\
c.\  &\dot f = (\dot b +\dot s x^2) f .     \label{gp8}
\end{align}
In particular, if $s(t) < 1/2$ for $t$ in an interval $[0, t_1]$ and
\begin{align}
\dot s &=  4s^2 -4s  +a\ \   \text{on} 
 \ \ [0, t_1] \ \ \ \text{and}\ \    \label{gp9}\\
 \dot b &= 2s      \label{gp10}
\end{align}
then  $ \dot f  = -Hf $  and 
\begin{align}
e^{-tH} f(0) = f(t) \ \ \ \text{on}\ \ \ [0, t_1]      \label{gp12}
\end{align}
\end{lemma}
\begin{proof} 
If $g \in C^{\infty}(\R)$ then the definition $(\n^*\n g, h)_{L^2(\gamma)} = \int_\R g'(x) h'(x) d \gamma(x)$,
valid for $h \in C_c^\infty(\R)$, together with an integration by parts proves \eref{gp6}. 
With $f$ given by \eref{gp3}, the identities \eref{gp7} and \eref{gp8} follow from  straight forward computations.
The technical issue as to whether $f(t,\cdot)$ is actually in the $L^2$ domain of $H$ is 
easily deduced from the fact 
that $f(t,\cdot)$ and $x\to x^2f(t,x) $ are in $L^2(\gamma)$ when $s(t) < 1/2$. The identity $\dot f = -Hf$ now follows from 
\eref{gp7} - \eref{gp10}.  The exponentiated version  of this, \eref{gp12}, follows by 
observing that, for $0 < t_2 < t_1$ and $h \in C_c^{\infty}(X)$, the
identity    
\begin{align*}
 (d/dt)(f(t), e^{(t - t_2)H}h)  &= (-H f(t), e^{(t - t_2)H}h) + (f(t), H e^{(t - t_2)H}h) \\
 &= 0\ \ \ \text{for}  \ \ 0 < t <t_2,
 \end{align*}
 implies that $ (f(t), e^{(t - t_2)H}h)$ is constant on $(0, t_2)$ and by strong continuity on $[0, t_2]$.
Therefore   $(e^{-t_2 H} f(0), h) =(f(0), e^{-t_2H}h) = (f(t_2), h)$ for all $h \in C_c^{\infty}(X)$.
Hence \eref{gp12} holds for all $t \in [0, t_1)$ and by strong continuity for all $t \in [0, t_1]$.
\end{proof}

\bigskip
\noindent
\begin{proof}[Proof of Theorem \ref{thmgp1}]   
We will construct a function $s(t)$  satisfying \eref{gp9}, 
which increases on $[0, t_1]$ and such that $f(0, \cdot) \in L^{p_1}(\gamma)$ but 
$f(t_1, \cdot) \notin L^{p_1}(\gamma)$. The theorem will then follow from \eref{gp12}.

 Since, by \eref{L313q},  $p_0$ is the larger zero of the upward opening parabola 
 $p\mapsto  p^2 -4\nu(p-1)$ and $p_1 > p_0$ 
 it follows that $p_1^2 - 4\nu(p_1 -1) >0$. Let $s_1 = 1/p_1$. Then 
\begin{align}
4s_1^2  - 4s_1 +1/\nu = \{s_1^2/\nu\}\(4\nu - 4\nu p_1 + p_1^2\)  >0.
\end{align}
Choose $\ep >0 $ so small that
\begin{align}
a \equiv \nu^{-1} - \ep &>0\ \ \ \text{and} \\
4s_1^2  - 4s_1 + \nu^{-1} - \ep &>0.
\end{align}
Then  $\int_\R e^{-\nu V} d\gamma = \pi^{-1/2}\int_\R e^{(\nu a -1)x^2} dx < \infty$ and therefore \eref{gp4} holds.
Choose $s_2 < s_1$ such that $4s^2  - 4s + \nu^{-1} - \ep >0$ on the interval $[s_2, s_1]$.  
Continuity of the quadratic polynomial ensures the existence of such a point $s_2$.
 Denote by $s(t)$ the solution to \eref{gp9}  with initial condition $s(0) = s_2$. 
Since $a = \nu^{-1} - \ep$, the right side of \eref{gp9} is strictly positive  for $s \in [s_2, s_1]$.
The solution will therefore increase and reach $s_1$ in a finite time $t_1 >0$. 
In particular  $s(t)\le s_1 =1/p_1 < 1/2$ for $0\le t \le t_1$.
Define $b(t) = \int_0^t 2s(t') dt'$. Then \eref{gp9} and \eref{gp10} are both satisfied 
on the interval $[0, t_1]$   and \eref{gp12}  holds  on this interval.  Let $g(x) = f(0, x)$.
Now 
\begin{align}
\int_\R (e^{b +sx^2})^{p} d\gamma(x) &= \pi^{-1/2}\ e^{pb} \int_\R e^{(ps -1)x^2}  dx  
                                            <\infty\ \ \text{iff}\ \ ps <1.                    \label{gp15}
 \end{align}  
 Since $p_1s(0) = p_1s_2 <p_1s_1 = 1$ we see from \eref{gp15} that  $g \in  L^{p_1}(\gamma)$.
On the other hand $p_1 s(t_1) = p_1 s_1 = 1$.  From \eref{gp15} we therefore find  that 
$f(t_1, \cdot) \notin L^{p_1}(\gamma)$. That is, $e^{-t_1 H} g \notin L^{p_1}(\gamma)$.
\end{proof}

\subsubsection{Positive potentials}\label{secpp}

 Theorem \ref{thmmp1} and Corollary \ref{corub1} give a sufficient condition on the 
 growth of $V$ that ensures that $\psi$ 
stays away from zero well enough that  $\int_X \psi^{-s} dm < \infty$  
 for all $s$ in an interval $(0, s_0)$. (cf \eref{W710a}.)
      We will show here, by example,  that this interval, determined by the quadratic equation \eref{W850},
       is not just an artifact of 
 the proof, and is close to necessary  in the sense that if $s >s_0$ then there is a potential
 $V$ such that $\|e^V\|_\ka < \infty$ while $\|\psi^{-1}\|_s = \infty$.

\begin{notation}     {\rm  For a real number $\w >0$ the Hamiltonian
\begin{align}
H_\w = - d^2/dx^2 + \w^2 x^2      \label{gp160}
\end{align}
has the normalized  ground state and, respectively, ground state measure
\begin{align}
\phi_\w(x)  = (\w/\pi)^{1/4} e^{-\w x^2/2}, \ \ \ \ m_\w = (\w/\pi)^{1/2} e^{-\w x^2} dx, \label{gp161}
\end{align}
as can be easily verified. The Gaussian measure $m_\w$ satisfies, by \cite{G1}, the logarithmic Sobolev inequality
\begin{align}
Ent_{m_\w}(f^2) \le \w^{-1} \int_\R |\n f|^2 dm_\w.        \label{gp163}
\end{align} 
}
\end{notation}

     In the preceding subsection we perturbed the potential
$\w^2 x^2$ (for $\w = 1$) by adding a negative quadratic potential $V \equiv -a x^2$ to $H_\w$ and found 
that in the ground state representation of $H_\w$ the resulting perturbed 
semigroup  $e^{-t(\hat H_\w +V)}$ had pathological behavior outside the 
interval of validity $[q_0, p_0]$ allowed by Corollary \ref{corhb2}. 
      In the present section we will
add a positive quadratic potential $V \equiv ax^2$ onto $H_\w$  to show that, for any $\ka >0$, 
\eref{gp170} and \eref{gp171}
can both hold  for suitable $a$.

      We should note in this example that the new ground state measure
associated to the perturbed Hamiltonian $-d^2/dx^2 +(\w^2 +a)x^2$ has a smaller Sobolev coefficient,
$(\w^2 + a)^{-1/2}$, than the unperturbed ground state  measure, whereas the perturbation method
we are using will always produce a bigger Sobolev coefficient than the unperturbed one.

\begin{theorem}  Let $\w >0$. Denote by $\n^*\n$ the Dirichlet 
form operator for $m_\w$. 
Let $\ka >0$ and define $s_0$ as in \eref{W851g}. 
Suppose that $s > s_0$.  
Then there is a potential $V \equiv a x^2$ such that
\begin{align}
\|e^V\|_{L^\ka(m_\w)} < \infty                       \label{gp170}
\end{align}
while the ground state $\psi$ for  $\n^*\n +V$ in $L^2(m_\w)$ satisfies
\begin{align}
\|\psi^{-1}\|_{L^s(m_\w)} = \infty.    \label{gp171}
\end{align}
\end{theorem}
      \begin{proof} For $a >0$ let $\alpha =\sqrt{\w^2 +a}$. Using the consecutive ground state
      theorem of Section \ref{secgst} we can compute the ground state $\psi$ for $\n^*\n + V$ by the ratio
      $\psi = \phi_\alpha/\phi_\w$. Thus 
      \begin{align}
      \psi(x) = (\alpha/\w)^{1/4} e^{(\w - \alpha)x^2/2}.       \label{gp172}
      \end{align}
      Hence
      \begin{align}
      \int_\R \psi^{-s} dm_\w = const. \int_\R e^{\(\frac{s(\alpha - \w)}{2}\  -\w\)x^2} dx.    \label{gp140a}
      \end{align}
This integral will be infinite if and only if the coefficient of $x^2$ is non-negative.
 That is, if and only if $s\alpha \ge (s+2)\w$. Squaring, we find the equivalent condition
 $s^2( \w^2 +a) \ge (s^2 + 4s + 4)\w^2$, and, equivalently, $s^2 a \ge 4(s+1) \w^2$, and, equivalently,
 $a\w^{-2} \ge 4(s^{-1} + s^{-2})$.     Thus the integrals in \eref{gp140a}
 are infinite if and only if $a\w^{-2} \ge 4(s^{-1} + s^{-2})$. 
  Now $\int_\R e^{\ka V} dm_\w  = const. \int_\R e^{(\ka a - \w) x^2} dx$, which is
        finite if and only if $a <\ka^{-1} \w$. That is, if and only if $ a\w^{-2} < (\w\ka)^{-1}$.
        Therefore  \eref{gp170} and \eref{gp171} both hold for some $a >0$ if and only if
    $4(s^{-1} + s^{-2}) <  (\w\ka)^{-1}$.

   Comparing \eref{gp163} with \eref{mt1} we see that $\w^{-1} = 2c$. Hence the equation  \eref{W850}
   for $s_0$ may be written $t^2 - 4\w\ka(t+1) = 0$. Therefore $(\w\ka)^{-1} = 4(s_0^{-1} +s_0^{-2})$.
   Since $s >s_0$ it follows that   $(\w\ka)^{-1} > 4(s^{-1} +s^{-2})$.  
      \end{proof}

\begin{remark}\label{remA4.13} {\rm  In  \cite[Remark 4.13 and Lemma 5.5]{Aida2001} 
Aida showed  that when $m$ is Gaussian 
it is sufficient for $\int e^{\ep V} dm < \infty$ for some $\epsilon >0$ 
(plus some conditions on the negative part of $V$) in order
for $\psi^{-1}$ to be in $L^{p}(m)$ for some $p > 0$.
}
\end{remark}

\subsection{Eckmann's theorem} \label{secEck}
   
 We  apply our techniques in this section to prove intrinsic hypercontractivity for  
 the one   dimensional Schr\"odinger operator 
\begin{align}
H\equiv - d^2/dx^2 + V.        \label{E200}
\end{align}
J.-P. Eckmann, \cite{Eck74}, described a class of potentials $V$ on $\R$ for 
  which the ground state measure of $H$ 
  satisfies a defective logarithmic  Sobolev inequality. 
  We will  derive a version of Eckmann's  theorem by a method 
  that illustrates  how to combine use
  of the Bakry-Emery criterion with Theorem \ref{thmM}, the main perturbation theorem of this paper.

 Suppose that $F$ is a continuous real valued function on $\R$ such that $dm \equiv e^{-2F} dx$ 
 is a probability measure. The Bakry-Emery criterion, \cite[Corollary 5.7.2]{BGL}, assures that 
  $m$ satisfies a logarithmic  Sobolev inequality if $F$ is uniformly  convex  on $\R$.
  Given a potential $V$,  we will construct  a uniformly convex function $F$ such that $e^{-F}$ is 
an approximate ground state, in some sense, for $H$.  
 We then use the  WKB identity \eref{W7} to produce a potential $W$ from $F$, whose ground 
 state is exactly $e^{-F}$. The probability measure $m \equiv e^{-2F} dx$ is therefore hypercontractive,
 but is the ground state measure for $W$, not $V$. We then apply our perturbation theorem, 
 Theorem \ref{thmM}, to the pair $m, V-W$ 
  to find another ground state measure satisfying a LSI and which, 
  by the consecutive transformation method  in Section \ref{secgst},   is exactly 
  the ground state measure for $V$. 

The notions of  intrinsic supercontractivity, \cite{Ros}, and intrinsic ultracontractivity,  \cite{DS84},
 are closely  related to intrinsic hypercontractivity and 
 have a large literature using very different techniques
from Eckmann's  and this paper's. See Remark \ref{reminap} for further discussion.

\bigskip

The following theorem is stated for an even potential for ease in reading. It's minor extension 
to more general potentials is  explained in Remark \ref{rem7.1}. 

 \begin{theorem} \label{thmeck4} $($Eckmann's Theorem$)$. 
 Let $V \in C^1(\R)$ and assume that $V$ is even. Suppose that there are constants $a >0$ and $k >0$
  and a number  $x_0 >0$ such that  $V(x) > 0$ when $x \ge x_0$ and
  \begin{align}
 &a.\ \  (d/dx) \sqrt{V(x)} \ge a\ \ \ \ \ \text{when}\ \ x \ge x_0 \ \   (\text{Eckmann's\  condition})   \label{E316}\\ 
 &b.\ \   (d/dx) V(x) \le k V(x)\ \      \text{when}\ \ x \ge x_0.          \label{E311b}
\end{align}
Then 
\begin{align}
-(d^2/dx^2) + V           \label{E312}
\end{align}
is bounded below. The bottom of the spectrum belongs to a unique positive ground state $\psi$. The
ground state measure $m_\psi = \psi^2 dx$ satisfies a logarithmic Sobolev inequality.
 \end{theorem}

The proof depends on the following construction of an intermediate ground state, 
which will be computationally useful in applications. 
 \begin{lemma}\label{leminterm}  $($Intermediate state$)$. 
 Let $V \in C^1(\R)$ and assume that $V$ is even. Suppose that there is a constant $a >0$ such that
 Eckmann's condition \eref{E316} holds.
  Let
\beq
F_0(x) = \int_{x_0}^x \sqrt{V(s)} ds\ \ \ \ \ \text{for}\ x \ge x_0.  \label{E314}
\eeq
Let 
\beq
b = \sqrt{V(x_0)} /x_0       \label{E314b}
\eeq
 and define
\begin{align}
F(x) =
\begin{cases} &F_0(x)  + bx_0^2/2 , \ \ \ x\ge x_0 \\
& b x^2/2, \qquad \ \ \  0\le x < x_0.
\end{cases}                                                                                \label{E173a}
\end{align}
Extend $F$  to be even on $\R$.  
Then $F$ and $F'$ are continuous on $\R$ and
\begin{align}
\int_\R e^{-pF} dx < \infty\ \ \text{for all}\ \ p >0.      \label{E315e}
\end{align}
Let  $\psi_0 \equiv Z^{-1} e^{-F}$  be normalized in $L^2(\R, dx)$ and define  $m^F = \psi_0^2 dx$. Then 
$m^F$ satisfies the  logarithmic Sobolev inequality 
\beq
Ent_{m^F}(f^2) \le \frac{1}{\min(b,a)} \int_\R (f')^2 dm^F.        \label{E116}
\eeq
In particular the Sobolev constant $c_F$ for $m^F$ satisfies $2c_F \le 1/\min(b,a)$.
\end{lemma}

    For the proof we need the following sublemma.
 
 \begin{sublemma} \label{lemqg2} If Eckmann's condition \eref{E316} holds then
\beq
F_0(x) \ge      \sqrt{V(x_0)}(x-x_0) +(a/2)(x-x_0)^2 \ \ \text{for all}\ x \ge x_0.        \label{E315c}
\eeq
In particular
\beq
\int_{x_0}^\infty e^{-pF_0(x)} dx < \infty,\ \ \text{for all}\ \ p >0. \label{E315d}
\eeq
\end{sublemma}
    \begin{proof} Let $u(s) = \sqrt{V(s)}$ for $s \ge x_0$. Then $u'(s) \ge a$ by \eref{E316}
and therefore  $u(s) \ge u(x_0) + a(s-x_0)$. Hence
$F_0(x) = \int_{x_0}^x u(s) ds \ge  u(x_0)(x-x_0) + (a/2)(x-x_0)^2$ for $x \ge x_0$.  \eref{E315d} follows.
\end{proof}   

\bigskip
\noindent
\begin{proof}[Proof of Lemma \ref{leminterm}] $F$ is clearly continuous on $\R$. It will suffice to make the following computations just for $x \ge 0$. Since $F$ is bounded on $[0, x_0]$, \eref{E315e} follows from
\eref{E315d}. $e^{-F}$ is normalizable in $L^2(\R, dx)$ and $m_F$ is a probability measure.

The first two derivatives of $F$ are given by
 \begin{align}
F'(x) &= 
\begin{cases} &\sqrt{V(x)},  \ \ \ \ \ x > x_0 \\
&b x, \qquad   0\le x < x_0.
\end{cases}                        \label{E110a}\\ 
F''(x) &= 
\begin{cases} &(d/dx)\sqrt{V(x)}  \ \ \ \ x > x_0 \\
& b  \qquad \ \ \ 0\le x < x_0
\end{cases}           \label{E111a}
\end{align}
$F'$ extends continuously to $[0, \infty)$ 
by the definition of $b$.    
Moreover $F''(x) \ge \min(b,a)$ everywhere except possibly at $x_0$.
Since $F'$ is continuous and $ F''$ is bounded away from
zero we can  apply the Bakry-Emery theorem 
(see e.g. \cite[Corollary 5.7.2]{BGL})    
 to the normalized measure $m^F = Z^{-2} e^{-2F} dx$. Bakry-Emery's theorem assures that 
\beq
Ent_{m^F}(f^2) \le 2c \int_\R (f')^2 dm^F
\eeq
 with $c= 1/\rho$ if $2F'' \ge  \rho$. In view of \eref{E111a} and Eckmann's condition, \eref{E316}, 
 we have $2F'' \ge 2\min (b,a)$.  So we may take $\rho = 2\min (b,a)$ and \eref{E116} follows. 
\end{proof}

\bigskip
\noindent
\begin{proof}[Proof of Theorem \ref{thmeck4}]
  Define a potential $W$ on $\R$ by applying the WKB equation \eref{W7} to the function $F$ defined
  in Lemma \ref{leminterm},  putting
\begin{align}
W =  - F''(x) + | F'(x)|^2.                   \label{E112}
\end{align}
At $x = \pm x_0$ this should be interpreted as a weak derivative.
Then the state $\psi_0$, defined in Lemma \ref{leminterm}, is the ground state for 
the Schr\"odinger  operator $-d^2/dx^2 + W$ and $m^F$ is the ground state measure.
$W$ can be computed explicitly 
with the help of  \eref{E110a} and \eref{E111a} as follows.
 \begin{align}
W &= - F''(x) + | F'(x)|^2       \notag\\
&=
\begin{cases}&-(d/dx)\sqrt{V(x)}   +V(x), \ \ \ \ x \ge x_0 \\
&=-b +b^2 x^2  , \qquad \qquad \ \ \ \ 0 \le x < x_0.
\end{cases}                \label{E182a}
\end{align}
Therefore
\begin{align}
V -W &= 
\begin{cases}&(d/dx)\sqrt{V(x)} , \qquad \ \ \ \ \ x \ge x_0 \\
& b -b^2 x^2  + V(x), \ \ \ 0 \le x < x_0.
\end{cases}       \label{E185}
\end{align}
In accordance with the consecutive ground state procedure of Section \ref{secgst}, the ground state, $\psi$,
for $-d^2/dx^2 +V$ in $L^2(\R, dx)$ is the product 
 of the ground state $\psi_0$  with the relative ground 
 state $\psi_1 \in L^2(m^F)$, defined as the ground state of
$\n^*\n + V-W$, where $\n^*\n$ is the Dirichlet form operator for $m^F$. 
           It suffices therefore to show that the ground state measure $(m^F)_{\psi_1}$ satisfies 
a logarithmic Sobolev inequality. 

For this we only need to verify the two hypotheses \eref{mt2} of Theorem  \ref{thmM} 
for the perturbation $V-W$ with $m = m^F$, since we already know that $\n^*\n$ satisfies the LSI \eref{E116}. 
$V- W$ is bounded on $[0,x_0]$ and, by \eref{E316}, 
 is positive on $[x_0, \infty)$. Therefore $V-W$ is bounded below and
\beq
\int_\R e^{-\nu (V-W)} dm^F = Z^{-2}\int_\R e^{-\nu (V-W)} e^{-2F} dx < \infty\ \ \ \text{for all}\ \ \ \nu >0.      \label{E187}
\eeq
This verifies the second  of the two conditions  \eref{mt2}.  
 To verify the first condition we need to show that
$\int_{x_0}^\infty e^{\kappa(V-W)  - 2F} dx < \infty$ for some $\kappa >0$, since $V,W$ and $F$ are all even.
We see from \eref{E311b} that $(d/dx) \sqrt{V(x)} \le (k/2)  \sqrt{V(x)} $ for $x \ge x_0$ and
therefore, by \eref{E185}, we have  $V- W \le (k/2)  \sqrt{V(x)} $ for $x \ge x_0$.
Hence $\kappa (V-W) - 2F \le \ka (k/2)  \sqrt{V(x)}  - 2F(x)$ on $[x_0, \infty)$. But
\begin{align*}
(d/dx) \(\ka (k/2)  \sqrt{V(x)}  - 2F\) &\le \ka (k/2)^2   \sqrt{V(x)}   -  2\sqrt{V(x)} \\
&=\{\ka (k/2)^2 -2\}   \sqrt{V(x)} .
\end{align*}
 Therefore, for some constant $C_1$ we have 
  \begin{align*}
&\kappa (V- W) - 2F  \le  \int_{x_0}^x  \{\ka (k/2)^2-2\} \sqrt{V(y)} dy +C_1 \\
& = \{2-\ka (k/2)^2\} (-F_0(x))  + C_1.
\end{align*}
Hence, if $ \kappa (k/2)^2 < 2$   then, by \eref{E315d}, we find 
\begin{align}
\int_{x_0}^\infty e^{ \kappa (V- W) - 2F} dx  &\le \int_{x_0}^\infty e^{-\{2-\ka (k/2)^2\} F_0(x)  + C_1} dx \notag\\
&<\infty.      \label{E190}
\end{align}

By Theorem \ref{thmM},  $\n^*\n +(V-W)$ is bounded below, has a unique ground state $\psi_1 \in L^2(m^F)$
 and the ground state measure for $\n^*\n +(V-W)$ satisfies a logarithmic Sobolev
 inequality.  Since, by the consecutive ground state procedure of Section \ref{secgst},
 this is the ground state measure for   $-d^2/dx^2 + V$, the theorem is proved.
\end{proof}

\begin{corollary} \label{correlgs}
Denote by $m^F$ the intermediate measure  defined in Lemma \ref{leminterm}  
 and by $\n^*\n$ its Dirichlet form operator. Let $\psi_1$ be the ground 
state for $\n^*\n +(V-W)$ in $L^2(m^F)$. Then $\psi_1$ is in $L^p(m^F)$ for all $ p < \infty$. 
In particular if  $f \ge 0$ then
\begin{align}
\int_\R f dm_\psi < \infty  \ \  \text{if} \ \ \    \int_\R f^q dm^F < \infty \ \ \text{for some}\ \ q >1.
\end{align}   
\end{corollary}
\begin{proof} As noted in the proof of \eref{E187} the potential $V-W$ for the relative ground state $\psi_1$
 is bounded below. It follows from Corollary  \ref{corpolbelow} 
  that $\psi_1 \in \cap_{p<\infty}L^p(m^F)$. Therefore, if $q >1$ and  $1/q + 1/p = 1$ then
  \begin{align}
  \int_\R f dm_\psi &=\int_\R f \psi_1^2 dm^F \notag \\
  &\le \|\psi_1^2\|_{L^p(m^F)} \(\int_\R f^q dm_F\)^{1/q} . \notag
  \end{align}
 \end{proof}

\begin{remark} \label{rem7.1} {\rm The restriction to an even potential can easily be removed. Suppose that
$V \in C_1(\R) $ and  that there are constants $a >0$ and $k >0$
  and a number  $x_0 >0$ such that  
 \begin{align}
 &a.\ \  (sgn\, x)\, (d/dx) \sqrt{V(x)} \ge a\ \  \text{when}\ \ |x| \ge x_0 \ \   (\text{Eckmann's\  condition})   \label{E316s}\\ 
 &b.\ \   (sgn\, x)\, (d/dx) V(x) \le k V(x)\ \      \text{when}\ \ |x| \ge x_0.          \label{E311s}
\end{align}
Then the conclusion of Theorem \ref{thmeck4} holds. The proof is the same if one takes into account
the change in signs on the negative half-line.
}
\end{remark}

 \bigskip
  
 Consider  a Schr\"odinger  operator  of the form
 \beq
 -d^2/dx^2 + V + V_1 \ \  \text{on}\ \  \R,          \label{E699}
 \eeq 
 in which  $V$ satisfies  the conditions of Eckmann's theorem, while  $V_1$ is merely measurable.
 If $m$ is the ground state measure for $-d^2/dx^2 +V$ then $m$ is hypercontractive by 
 Eckmann's theorem. 
 The consecutive ground state transformation method 
 together with Theorem \ref{thmM} can  in principle  be used to  show that the ground state 
 measure for the  full operator \eref{E699}  will be hypercontractive if $V_1$ satisfies 
 exponential bounds of the form \eref{mt2}.
  But Eckmann's theorem gives only indirect information about the Sobolev coefficient of 
  the measure  $m$.  
  In the next corollary we will establish the hypercontractivity of the ground state measure 
  of the operator   \eref{E699}, but   by 
  replacing $m$ by the explicit intermediate measure $m^F$, thereby getting 
  conditions on $V_1$  which are easily verified in 
  applications.

     \begin{corollary} \label{corpp5} 
Suppose that $V$ is a potential that  satisfies the conditions of Eckmann's theorem, 
\eref{E316} and \eref{E311b}.
Let $e^{-F}$ be the intermediate ground state for $V$, constructed in Lemma \ref{leminterm}. 
Denote by $c_F$ the Sobolev constant for the measure $m^F$. 
Let $V_1$ be a measurable potential such that
\begin{align}
\int_\R e^{-\nu_1 V_1} dm^F &< \infty\ \ \text{for some}\ \  \nu_1 > 2c_F \ \ \text{and}    \label{E701a}\\
\int_\R e^{\ka_1V_1} dm^F &< \infty\ \  \text{for some}\ \ \ka_1 >0.    \label{E700a}
\end{align}
Then the Schr\"odinger operator $-(d/dx)^2 + V+V_1 $ is bounded below, has a unique positive ground state
$\psi \in L^2(\R, dx)$ and the ground state measure $\psi^2 dx$ satisfies a logarithmic Sobolev inequality.
\end{corollary}
      \begin{proof}  Let $W$ be the intermediate potential defined in \eref{E182a}. 
      Writing $V+ V_1 = W + (V+ V_1 - W)$, we may apply the consecutive groundstate 
  transformation method to realize the ground state measure for
$-d^2/dx^2 + V + V_1$ as the ground state  measure for $\n^*\n + (V+V_1 -W)$, 
where $\n^*\n$ is the Dirichlet form operator for $m^F$.  
By Theorem \ref{thmM} we need to show then that
\begin{align}
\int_\R e^{-\nu(V_1 +V-W)} dm^F &< \infty   \ \ \text{for some}\ \nu > 2c_F\ \  \text{and}  \label{E705}\\
\int_\R  e^{\ka(V_1 +V-W)} dm^F &< \infty \ \ \text{for some}\  \ka >0.          \label{E706}
\end{align}
But $V-W$ is bounded below, by \eref{E185} and \eref{E316}.  
Therefore \eref{E705}, with $\nu =\nu_1$,
 follows from \eref{E701a}.
 
 For the proof of \eref{E706} suppose that $\kappa_2>0$ and that \eref{E190} holds for this value of $\kappa$.
It was shown in \eref{E190}  
 that any $\ka_2 \in (0, 8/k^2)$ will do. 
Let $\ka = (1/2)\min(\ka_1, \ka_2)$. Then
\begin{align}
\int_\R  e^{\ka(V_1 +V-W)} dm^F \le\(\int_\R e^{2\ka V_1} dm^F\)^{1/2} \(\int_R e^{2\ka(V-W)} dm^F\)^{1/2}.
\end{align}
Since $2\ka \le \ka_1$ the first factor  on the right is finite by \eref{E700a}. Since $2\ka \le \ka_2$ the second factor is finite by \eref{E190}.
\end{proof}

\subsubsection{Second order intermediate state}

 The next theorem further illustrates use of the combination 
of Bakry-Emery  convexity followed by  our perturbation theorem.
An additional derivative is assumed for $V$
but a wider variety of growth conditions are permitted.  
This modification of Eckmann's theorem is based on a second order approximation 
in the WKB method which is  discussed by  A. Dicke  in  Appendix IV of  \cite{Simon1970}.
We use the second order WKB approximation to construct an intermediate ground state, whose
potential is closer to the given potential $V$ than that in the preceding method.

\begin{theorem} \label{thmeck5}  Let $V$ be an even function in $C^1(\R)$.
 Suppose that there  is a number $x_0 >0$ 
 such that, on $[x_0, \infty)$, $V \in C^2$ and $V >0 $.
 Define  $F_0$ by \eref{E314} again. Let
 \begin{align}
 g(x) &= (1/4) (d/dx) \log V(x)    \ \   \ \text{for}\ \ x \ge x_0. \label{E340d}
 \end{align}
 Suppose  that there is a  constant $a >0$ such that   Eckmann's condition \eref{E316} holds and also,
 with $F_0$ given by \eref{E314}, assume that
  \begin{align}
 &a.\ \  (d/dx) \(\sqrt{V(x)}  + g(x) \)\ge a\ \ \text{when}\ \ x \ge x_0 \ \    \label{E360}   \\
&b.\ \    g(x)^2 =o(F_0(x))   \ \ \text{as}\ x \to  \infty      \label{E317c} \\
 &c. \ \   |g'(x)| = o(F_0(x))\ \ \text{as}\ x \to \infty.    \label{E318c}
\end{align}
Then, over $\R$, 
\begin{align}
-(d^2/dx^2) + V           \label{E320}
\end{align}
is bounded below. The bottom of the spectrum belongs to a unique positive ground state $\psi$. The
ground state measure $m_\psi := \psi^2 dx$ satisfies a logarithmic Sobolev inequality.
\end{theorem}

Note that the condition  \eref{E317c} is  weaker than \eref{E311b} since the latter assumes that $g$
is bounded on $[x_0, \infty)$ while the former allows $g$ to be unbounded  
by virtue of  Sublemma \ref{lemqg2}.
    
\bigskip
\noindent
\begin{proof}[Proof of Theorem \ref{thmeck5}]  Let 
\begin{align}
b &= \(\sqrt{ V(x_0)}  + g(x_0)\)/x_0   \label{E341c}
\end{align}
and define
\begin{align}
F(x) =
\begin{cases} &\int_{x_0}^x \(\sqrt{V(s)} +g(s)\) ds  + bx_0^2/2 , \ \ \ x\ge x_0 \\
& b x^2/2, \ \ \ \ 0\le x < x_0.
\end{cases}                                                                                \label{E173b}
\end{align}
Then on $[x_0, \infty)$ we have
\begin{align}
 F'(x) &= \sqrt{V(x)}  + g(x)    \label{E350} \\
 F'^2 &= V(x) + g(x)^2 + 2 g(x) \sqrt{V(x)}           \label{E351} \\
 -F''(x) &= -(d/dx)\sqrt{V(x)} - g'(x).      \label{E352}
 \end{align}
 Notice that $2g(x) \sqrt{V(x)} = V'(x)/(2\sqrt{V(x)}) = (d/dx)\sqrt{V(x)}$. The last  
 term in \eref{E351}  therefore cancels with the first term in \eref{E352} in the expression for the intermediate potential $W$ to give 
 \begin{align}
 W&= - F'' + (F')^2 \\
  &= - g'(x)  + V(x) + g(x)^2  
 \end{align}
 over the interval $[x_0, \infty)$. Hence over this interval we have
 \beq
 V-W = g'(x) - g(x)^2. \label{E353}
 \eeq
 It follows from  \eref{E317c} and \eref{E318c} that
 \begin{align}
  |V-W| = o(F_0(x)) \ \ \text{as}\ \   x \to \pm \infty .       \label{E354}
 \end{align}
 On the interval $[x_0, \infty)$ we have $g(x)=(1/4)V'(x)/V(x) =$ \linebreak
 $(1/2)\((d/dx)\sqrt{V(x)}\) / \sqrt{V(x)} \ge 0$. Therefore $ F \ge  F_0 +bx_0^2/2$ on this interval.
 Hence, for any real number $p$ we have, in view of \eref{E354},
 \begin{align}
 p(V-W) -2F &\le p(V-W) -2F_0 -  bx_0^2\ \ \  \text{for}\ x\ge x_0 \\ 
  &\le -F_0 -  bx_0^2 \ \ \  \text{for large $x$ depending on $p$}. 
 \end{align}
 Therefore, since $V -W$ is locally bounded, we have
 \begin{align}
 \int_\R e^{p(V-W) - 2F} dx < \infty      \label{E361}
 \end{align}
 for any $p \in \R$.  From \eref{E352}  and \eref{E360} we see that $F''\ge a$ on $[x_0, \infty)$.
 From \eref{E173b} we then find that $F'' \ge \min(a,b)$ everywhere except at $x = \pm x_0$.
 $m^F$ is therefore hypercontractive by the Bakry-Emery theorem.  In view of \eref{E361} 
 Theorem \ref{thmeck5} now follows 
 from Theorem \ref{thmM}.
\end{proof}

\subsubsection{Examples of Eckmann's theorem} \label{secExE}
 
          In each of the following examples we consider the one dimensional 
Schr\"odinger operator \eref{E200}. We take $V$ to be an even function  for simplicity.
    It suffices then to  compute derivatives for $x >0$. 
    
           In the first five examples we will 
           apply Eckmann's Theorem in the form of 
Theorem \ref{thmeck4}. But in the sixth example the potential grows too rapidly and we must use
the more refined theorem, Theorem \ref{thmeck5}, which is based on a second order WKB approximation.

\begin{example} \label{exE1}  {\rm  
   {\bf(Potential with power growth).}  
   Let  $V(x) = \lambda |x|^{2r}$ for  some $r \ge 1$ and $\lambda >0$.   
   Choose $x_0 \ge 1$. Then
   \begin{align}
 &\ \ (d/dx)\sqrt{V(x)} = \lambda^{1/2}r x^{r-1}\ge    \lambda^{1/2}  r x_0^{r-1}\ \ \text{when}\ \ x \ge x_0.  \label{E370}
 \end{align}
 So \eref{E316} holds with $a = \lambda^{1/2} r x_0^{r-1}$.
 Moreover for $x \ge 1$ we have $(d/dx)V(x) \le 2r V(x)$. 
 So  \eref{E311b}  holds with $k = 2r$  and Theorem \ref{thmeck4} applies. 
 Thus $H$ is bounded below, the bottom of the spectrum is an eigenvalue 
   of multiplicity one belonging  to a positive ground state $\psi$ and the ground state measure $\psi^2 dx$ satisfies a logarithmic Sobolev inequality.  In this example the intermediate ground state is $e^{-F}$, where  
   \beq
   F(x) = \lambda^{1/2} \int_{x_0}^x s^r ds  + const. =  \lambda^{1/2} x^{r+1}/ (r+1) \text{+ const. for large}\ x. \label{E371}
   \eeq
}
\end{example}

 \begin{remark}\label{remsuper} {\rm  In case $r >1$, \eref{E370} shows that
 $a$ can be chosen 
 large by choosing $x_0$ large. Moreover $b$, defined in \eref{E314b}, 
 also increases to $\infty$ as $x_0 \uparrow \infty$. Consequently the intermediate 
 measure $m^F$ can be chosen to 
 have an arbitrarily  small Sobolev coefficient by \eref{E116}.
 }
 \end{remark}

 \begin{example} \label{exE2} {\rm {\bf (Perturbation of power growth).} 
Let  $V = |x|^{2r} + V_1$ for  some $r\ge 1$ and some locally bounded, even, measurable 
function $V_1$  such that
\begin{align}
|V_1(x)| =  o(|x|^{r+1})\ \ \ \text{as}\ \ |x| \to \infty,        \label{E450a}
\end{align}
or, in case $r>1$,
\begin{align}
|V_1(x)| =  O(|x|^{r+1})\ \ \ \text{as}\ \ |x| \to \infty .        \label{E450b}
\end{align}
Then $-(d/dx)^2 + V$ is bounded below, has a unique positive ground 
state $\psi$ and the ground state measure $\psi^2 dx$ satisfies a logarithmic Sobolev inequality.

\begin{proof} We apply Corollary   \ref{corpp5} with $V =  |x|^{2r}$. Then \eref{E371}, with $\lambda =1$,
 determines the behavior near $\infty$  of the density of the intermediate measure $m^F$.
 If \eref{E450a} holds then we see that $\int_\R e^{ p |V_1|} dm^F < \infty$
for all real $p$. Therefore \eref{E701a} and \eref{E700a} both hold.
If only \eref{E450b} holds then  the integral in\eref{E701a} is only finite for some $\nu_1 >0$.
But by Remark  \ref{remsuper} a change in $F$ locally (by increasing $x_0$) can produce
an $F$ such that $m^F$ has arbitrarily small Sobolev coefficient $c_F$, while \eref{E371} still holds. 
Choose $x_0$ such that $c_F < (1/2) \nu_1$. Then  \eref{E701a} holds. 
 \eref{E700a} holds because $\ka_2$ can be chosen arbitrarily small. 
\end{proof}
}
\end{example}

 \begin{example}\label{exE3} {\rm
 {\bf (Polynomial potential).} 
 Let $V(x) = \sum_{j=0}^n a_jx^{2j}$ be an even  polynomial  with $a_n >0$.
 Choose $x_0 >0$  so that $V(x) \ge 1$ for $x \ge x_0$. 
 Since  $(d/dx) \sqrt{V(x)} =(1/2) V'(x)/\sqrt{V(x)}$ which, for large $x$
 behaves like $n a_nx^{2n-1}/\sqrt{a_n} x^n$ 
 $ = n\sqrt{a_n} x^{n-1}$, we can choose $x_0$ so that $(d/dx) \sqrt{V(x)} \ge a$ on $[x_0, \infty)$ for some
 $a >0$.  Moreover $V'(x)/V(x) \to 0$ as $x \to \infty$. So Theorem \ref{thmeck4} applies.
}
\end{example}

\begin{example} \label{exE4} {\rm {\bf (Potential with slow growth).} 
Suppose that $v(x);[0,\infty) \to (0, \infty)$ is $C_1$ and 
$0 < v' \le c < \infty$.
Let $V(x) =x^2 v(x)^2$ for $x \ge 0$ and extend $V$ to be even on  $\R$. Then, for $x >0$ we have
\begin{align}
(d/dx) \sqrt{V(x)} = (d/dx)\(x v(x)\) = v(x) + x v'(x)
\end{align}
Let $a = v(1)$. Then $(d/dx) \sqrt{V(x)} \ge a$ when $x \ge 1$ because $v$ is increasing. 
Moreover $V'(x) =2xv(x)^2 + 2x^2 vv'(x) \le 2x^2 v(x)^2 + 2cx^2v(x) \le 2(1 +(c/a))x^2 v(x)^2$ for $x \ge 1$.
 Therefore Theorem \ref{thmeck4} applies. The ground state  measure is hypercontractive. 
 
 In particular this example includes the cases   $v(x) = (\log(3+x))^{b}$ with $b >0$ 
 since $0 < v' \le c$  on $[0,\infty)$ for some constant $c<\infty$. 
 Therefore  the potentials $V(x) = x^2(\log(3+|x|))^{2b}$ have hypercontractive 
ground state measures, $m_\psi$.
Davies and Simon, \cite{DS84}, have shown that $m_\psi$ is ultracontractive if $2b >1$ but not if $2b \le 1$.
Our method does not distinguish these two cases.
}
\end{example}

\begin{example} \label{exE5}{\rm
 {\bf (Potential with exponential growth).} Let  $c >0$ and let  $V(x) = e^{2c|x|}$ for $|x| \ge 1$ 
 and be smooth and  even on $[-2,2]$. Choose $x_0 >1$. Then
\begin{align}
&\qquad (d/dx)\sqrt{V(x)}  = ce^{cx} \ge ce^{cx_0} \ \     \text{when}\ \ x \ge x_0.
 \end{align}
 So \eref{E316} holds with  $a = ce^{cx_0}$. 
 Moreover \eref{E311b} holds because $V'(x) = 2c V(x)$ for $x \ge x_0$. 
 Theorem \ref{thmeck4} therefore applies.
  }
 \end{example}

\begin{example}\label{exE6}  {\rm
{\bf (Potential with very rapid growth).} In the following example the potential grows too rapidly to satisfy the growth condition \eref{E311b}.
But the refined version of Eckmann's theorem, 
 Theorem \ref{thmeck5},  applies:
 Let $\alpha >0$ and let  $V(x) = e^{2\alpha x^2}$. Choose  $x_0 >0$.
 Then, for $x \ge x_0$, we  have
\begin{align}
&\qquad (d/dx)\sqrt{V(x)}  = 2\alpha x e^{\alpha x^2} \ge  2\alpha  x_0 e^{\alpha x_0^2} \\ 
&\qquad g(x) = \alpha x,\  g'(x) = \alpha.\\
&\qquad F_0(x) = \int_{x_0}^x e^{\alpha s^2} ds  
\end{align}
Eckmann's condition, \eref{E316},  holds with any choice of $x_0 >0$. But the growth
limitation \eref{E311b} doesn't hold because $g$ is unbounded.
Nevertheless \eref{E360}, \eref{E317c} and \eref{E318c} hold and Theorem 
\ref{thmeck5} applies.
}
\end{example}

\begin{remark} \label{reminap} 
 {\rm (Inapplicability to intrinsic ultracontractivity). 
 Davies and Simon  introduced in \cite{DS84} the terminology ``intrinsically ultracontractive" to refer
 to a   Schr\"odinger operator $H = -\Delta + V$ on $\R^n$ for which 
 $e^{-t\hat H}:L^2(\psi^2 dx) \to L^\infty$  is bounded for all $t >0$, where $\hat H$ is the  
 ground state transform of $H$.  
 Since the Schr\"odinger operators of interest for ultracontractivity
 typically have a mass gap, an intrinsically ultracontractive  Schr\"odinger  operator will also 
 be intrinsically hypercontractive by virtue of the Glimm-Segal-Rothaus theorem.
An intrinsically hypercontractive Schr\"odinger operator, however,  need not be intrinsically 
ultracontractive,  as we see in the harmonic oscillator. 
 There is a large literature  proving and exploiting intrinsic ultracontractivity, both 
 for Schr\"odinger operators  and for Dirichlet Laplacians on open subsets of $\R^n$.
   The proofs usually depend on dimension dependent estimates of the defect in the 
defective LSI for the ground state measure.  These in turn
depend on delicate pointwise estimates of the ground state, near infinity in the 
case of Schr\"odinger operators, or near the boundary in the case of the Dirichlet Laplacian. 
Intrinsic ultracontractivity is qualitatively stronger than intrinsic hypercontractivity because 
hypercontractivity  only yields bounds on $\|e^{-t\hat H}\|_{2\to p}$ for $p < \infty$.
 There is a trade-off between dimension independence
and sup-norm bounds. 
Our techniques in this paper are aimed at dimension independence. 

  A qualitative distinction between intrinsic hypercontractivity and intrinsic ultracontractivity
  is  already apparent in  Example \ref{exunb}, which describes a bounded potential whose addition to a
  hypercontractive Schr\"odinger operator $\n^*\n$ yields a ground state $\psi$ for which  
  $\psi$ and $\psi^{-1}$ are both unbounded. But \cite[Theorem 3.4]{DS84} shows that 
  the ground state for the perturbation
  of an intrinsically  ultracontractive Schr\"odinger operator by a bounded potential is always bounded
  and bounded away from zero.

In the literature on intrinsic ultracontractivity there are parallel versions of our Example  \ref{exE2}.
Their assumptions on the potential $V$  
 take the form   $f(x) \le V(x) \le g(x)$, where $f$ and $g$  are specified. See e.g. 
  \cite[Proposition 5.5]{Car79}, \cite[Theorem 6.3]{DS84}, 
 \cite[Lemma 4.5.1]{Da89}   and \cite{Ci2}.

For an avenue into this large literature on intrinsic ultracontractivity see the early 
papers Rosen, \cite{Ros}, Davies and Simon, \cite{DS84},  Davies, \cite{Da83,Da85,Da89}, 
Carmona, \cite{Car78,Car79}, Banuelos, \cite{Ban1991}, Murata, \cite{Mu1993},
Cowling and Meda, \cite{CoMe1993}, Lianantonakis, \cite{Lian1993}, 
Cipriani, \cite{Ci1,Ci2,CiG95},  Z-Q. Chen and R. Song, \cite{ChenZQ1997}, Tomisaki, \cite{Tom2007} 
and the  citations lists for these papers in Mathematical Reviews.
}
\end{remark}

\bigskip

\begin{remark}\label{reminfdim}{\rm   The bounds we have gotten in the general theory are dimension
independent. This reflects the fact that it was not necessary to use the classical Sobolev inequalities.
Eckmann's theorem is essentially one dimensional, although Eckmann also applied his methods
to radial potentials over $\R^n$.  We will see later, in the toy model for the quantum field $\phi^4_2$,
how dimension independence can expected to be used. But to emphasize the dimension independence
in a simple, though artificial example, consider the Dirichlet form operator over an abstract Wiener space
$(H, B, m)$, where $H$ is a real separable Hilbert space densely embedded in a Banach space $B$
and $m$ is the centered Gaussian measure on $B$ with covariance given by the 
inner product of $H$, \cite{G1967a}.
Denoting by $\n$ the gradient of functions on $B$ associated to differentiation only in $H$ directions,
referred to as the $H$ derivative in \cite{G1967b},
the Dirichlet form operator $\n^*\n$ is densely defined in $L^2(B,m)$ and is the well known
number operator of quantum field theory. The logarithmic Sobolev inequality 
$Ent_{m}(f^2) \le 2 \int_B |\n f|^2 dm$  
holds  on $B$ because it  reduces to the $n$ dimensional Gaussian LSI in case $f$ is a cylinder
function based on some $n$ dimensional subspace of $H$, while these functions form a core for $\n^*\n$.

The arguments in Sections \ref{secesa} and \ref{secEU}, showing that
if \eref{mt2} holds for some potential $V$ on $B$ then the Schr\"odinger operator $\n^*\n + V$ 
has a unique ground state $\psi$  which is strictly positive almost
everywhere, apply with no essential change even though $B$ is not finite dimensional. Theorem
\ref{thmM} also applies, from which we can conclude that $Ent_{m_\psi}(f^2) \le 2c_1 \int_B |\n f|^2 dm_\psi$
for a constant $c_1$ computed  as in Section \ref{sectightening}. 
For example, denoting by $\|\cdot\|$ the $B$ norm, the potentials 
$V(x) = \|x\|^{\beta}, 0\le \beta \le 2$ and $V(x) = \|x\| \sin \|x\|$ both satisfy the condition \eref{mt2}.
One needs only to use Fernique's theorem on the distribution of $\|\cdot\|$, \cite[Theorem 3.1]{Kuo}, 
or \cite[Theorem 3.6]{AMS94}, for both examples.
}
\end{remark}

\begin{remark} {\rm  In his paper, \cite{Eck74}, Eckmann also allowed potentials which have a strong 
singularity at zero. The techniques that we have been exploiting are not appropriate for such 
singularities.
}
\end{remark}

\subsection{Irregular potential over $\R$. } \label{secwp}

\begin{remark}\label{remcombined}{\rm  (Conditions on $\xi F - V$).
 We have been concerned with the probability measure $m^F \equiv e^{-2F}m$ only when 
$F = - \log \psi$, where $\psi$ is the ground state of a given Schrodinger operator $\n^*\n +V$.  But,
as mentioned in the Introduction, there is a large literature in which $F$ is given and is 
the primary object of interest, rather than the potential $V$.     
In that case one is interested in conditions on $F$ itself which ensure that $m^F$ is hypercontractive. 
The relation between these two problems was first discussed by Kusuoka and Stroock, \cite{KS85},
which appeared about  the same time as the paper of Davies and Simon  
that  introduced intrinsic hypercontractivity, \cite{DS84}.
 Kusuoka and Stroock explained that if, given $F$,  one defines an artificial potential  
 by $V_F = \n^*\n F + |\n F|^2$ then, taking
Davies and Simon's given potential $V$ to be $V_F$, the hypotheses in both papers are very similar.
In particular, they both depend on information about the combination $\xi F - V$ for 
various values of a real number $\xi$.
It can already be seen from \eref{gs720b} that if $F - cV$ is bounded above then $m_\psi$ satisfies 
a defective logarithmic Sobolev inequality if $m$ does. 
Conditions  which impose bounds on $\xi F - V$ from above figure prominently in either hypotheses or intermediate steps in the early papers Rosen \cite{Ros}, Carmona \cite{Car79},  
Davies and Simon, \cite{DS84}, Kusuoka and Stroock, \cite{KS85}. Davies's book, \cite{Da89},
 gives a self-contained exposition of parts of this also.
     Cattiaux, \cite[Section 5]{Cat2005}, takes $F$ as the primary object and imposes upper bounds 
 on $\xi F- V_F$ as well as various integral bounds.    
 $V_F$ arises naturally in 
 his paper because the Girsanov formula for change of
 density by $e^{-F}$ relates closely to the Feynman-Kac formula for $V_F$. See Carmona \cite{Car79a}
 for a discussion of this relation. 
  Cattiaux also gives another kind of mixed condition which shows how 
 close to necessary is our condition $\|\psi^{-1}\|_{L^p(m)} < \infty$   for proving a DLSI. 
 In  \cite[Theorem 2.5]{Cat2005}
 he shows
  that if $V_F$ is bounded below
 then for a DLSI to hold it is necessary and sufficient that $\psi^{-1} e^{-t(\n^*\n + V_F)} 1 \in L^p(m^F)$ for some $t >0$ and $ p > 2$.   For $t = 0$ this reduces to Aida's condition.

 Carlen and Loss, \cite{Ca2004}, also assume $F - cV_F$ is bounded above
 to show that $e^{-2F} d^nx$ satisfies a logarithmic Sobolev inequality. Their method is based on use of  a 
 perturbation of  the known Euclidean logarithmic Sobolev inequality. 
 Just such a combination of $F$ and $V$ is also used in the book 
 of Bakry, Gentil and Ledoux, \cite[Section 7.3]{BGL},
to determine a growth function for a general class of entropy-energy inequalities.
Bartier and Dolbeault \cite{BaDo2006}, in their perturbation theorem for a measure $m \equiv e^{-W}d^nx$
by a density $e^{-2F}$, assume that $V_F$ is bounded below and that $F$ is 
bounded above  to perturb a logarithmic Sobolev inequality,  and 
they also show that one can perturb the inequalities of Beckner, \cite{Beckner1989}, that are 
intermediate between Poincar\'e and LSI by assuming again that $V_F$ is bounded below and 
that $F\vee 0$ is in a suitable $L^p(m)$ class.
Another kind of hypothesis involving a combination of $F$ and $V_F$ is given by F.Y. Wang 
in \cite[Equ. (5.4)]{Wang2001}.
He assumes that $\int \exp{\epsilon F -c_\epsilon V_F} dm  < \infty$ for suitable $\ep$ and $c_\ep$.

 A change of density by a given factor $e^{-2F}$ arises also for purely geometric 
motives over Riemannian manifolds. In the papers \cite{CLR2015}, \cite{CR2019}, 
 Charalambous, Lu and Rowlett  use the ground state
transformation to gain information about the spectrum of the Dirichlet form operator for  
 the measure with density $e^{-2F}$ with respect to Riemann-Lebesgue measure.
 They transform the problem into  the study of the Schr\"odinger operator $-\Delta + V_F$ and impose
a uniform bound on $V_F$.  This is a very natural condition in this context.
 }
\end{remark}

Theorem \ref{thmM} shows that only conditions on the perturbing potential $V$ involving means
are required to produce a hypercontractive ground state measure.
The following one dimensional example of an irregular potential emphasizes this fact and 
at the same time shows that  $V$ can be so badly unbounded below 
  that a combined  condition such as $\sup (\xi F - V) < \infty$ can fail over every interval, even though
  the ground state measure for $-d^2/dx^2+V$ is hypercontractive.

\begin{example} {\rm (Irregular potential). We will construct a potential
$V$ over $\R$ which is unbounded below on every interval but 
 such that
 the Schr\"odinger operator
$H \equiv -d^2/dx^2 + V$  is bounded below, has an eigenvalue at the bottom of its spectrum, has a unique 
continuous 
ground state $\psi > 0$ a.e.  
 for which the ground state  measure $\psi^2 dx$ satisfies
a logarithmic Sobolev inequality.  In particular $\xi F - V$ is not bounded above on any interval for any real
number $\xi$.

 Let $r_1, r_2, \dots$ be an enumeration of the rational numbers
in $[0,1)$. Define $f(x) =\sum_{j=1}^\infty 2^{-j} |x- r_j|^{-1/2}$
for $0\le x < 1$ and $f(x) = 0$ elsewhere. Then $f$ is unbounded above on every open set in $ [0,1)$.
But  $\int_0^1 f(x)dx \le 4$ because  $\int_0^1 |x-r|^{-1/2} dx \le 2 \int_0^1 x^{-1/2} dx =  4$.
So $f$ is finite a.e.. Let $b >0$ and define
\begin{align}
V_1(x) = -(1/6)\sum_{n= -\infty}^\infty \log (1 + b2^{-|n|} f(x-n))      \label{Ir5}
\end{align}
The terms in the  sum have disjoint supports. Clearly $V_1$ is non-positive and is unbounded 
below on  every interval.  Define
\beq
V(x) = x^2/4 +V_1(x).
\eeq
$V$ is  unbounded below on every interval. We will use Theorem \ref{thmM} and the consecutive 
ground state transformation procedure  of Section \ref{secgst} to prove that $-d^2/dx^2 +V$ has the 
properties claimed above. 

     Let $m \equiv (2\pi)^{-1/2} e^{-x^2/2} dx$ be the standard
Gaussian measure on $\R$ of variance $1$.  Then  $m$ satisfies 
the LSI \eref{mt1} with $c =1$, \cite{G1}. Moreover $m$ is the ground state measure for the
Schrodinger operator $-d^2/dx^2 + x^2/4$ because $\sqrt{dm/dx}$ is the ground state, as
one can easily compute. By the consecutive ground state procedure  of Section \ref{secgst} it therefore
suffices to check the two conditions \eref{mt2} for $m$ and $V_1$. In view of the disjoint 
supports of the terms in \eref{Ir5} we find, for any $\nu > 0$, and then for $\nu =6$ that
\begin{align*}
\int_\R &e^{-\nu V(x)} dm(x) = \sum_{-\infty}^\infty \int_n^{n+1}(1+b2^{-|n|}f(x-n))^{\nu/6}  dm(x)\\
&=\sum_{-\infty}^\infty \int_n^{n+1} 1 dm + \sum_{-\infty}^\infty b 2^{-|n|}\int_n^{n+1} f(x-n) dm(x) \ \ 
                     \text{if}\ \nu =6\\
&\le 1 + b\sum_{-\infty}^\infty 2^{-|n|} \int_n^{n+1} f(x-n)  dx \\
&\le 1 + 12b .
\end{align*}
Since $6 > 2 = 2c$ the second condition in \eref{mt2} is satisfied for $m, V_1$. 
The first condition is automatic because $V_1 \le 0$. Therefore $H$ has a unique ground state $\psi$
whose ground state measure $\psi^2 dx$ satisfies a logarithmic Sobolev inequality. Moreover $\psi$
is given by $\psi(x) = (2\pi)^{-1/4} e^{-x^2/4} \psi_1(x)$ where $\psi_1$ is the ground state for
$\n^*\n + V_1$ and $\n^*\n$ is the Dirichlet form operator of $m$. 
In particular $\int_\R \psi_1'(x)^2 dm(x) < \infty$. So $\psi_1 \in H_{1, loc}(\R)$ and
one can  take $\psi_1$ to be continuous. Consequently  $\psi >0\ a.e.$ and is continuous. Hence
$F \equiv - \log \psi$ is continuous on some neighborhood of any point where $\psi(x) >0$ and bounded
on some smaller neighborhood. Therefore $\xi F - V$ is unbounded above on any open interval.
The same argument shows that in the intermediate space $L^2(\R, dm)$, 
the combination $\xi F_1  - V_1$  is also unbounded above on any open interval for any real number $\xi$,
wherein $F_1 = - \log \psi_1$. 
}
\end{example}

\bigskip

 \begin{remark}\label{remF} {\rm (Direct conditions on $F$). Size conditions on $F$ itself, not dependent on differentiability of $F$ and 
    in particular not dependent on $V_F$, which ensure that $m^F$ is hypercontractive 
    when $m$ is, have been explored with a view toward extending 
    the Deuschel-Holley-Stroock theorem  \cite{HS1987,DS1990}, 
     according to which it is sufficient for $F$ to be bounded. 
Hebisch and Zegarlinski,  \cite[Proposition A.1]{HZ1}, 
have  given a dramatic example showing that even if a density  on $\R$ is sandwiched  
between two densities that give a Poincar\'e inequality, the sandwiched density need not.
     Bakry, Ledoux and Wang, \cite{BLW2007}, have explored pure growth conditions on $F$
which ensure that $m^F$ satisfies a slightly weaker functional inequality than $m$ does in a scale
of inequalities interpolating between a Poincar\'e inequality and a logarithmic Sobolev inequality.
But their results suggest that invariance of LSI under some reasonable class of 
unbounded pointwise perturbations of $F$ may not hold. 
       The perturbation theorem of Barthe and Milman \cite{Barthe13} does not fall into any of 
these three categories of perturbation theorems. They consider, e.g., two measures
on $\R^n$, $\mu_i = e^{-2F_i} dx$, $i = 1,2$, and assume that $\mu_1$ satisfies a logarithmic Sobolev
inequality while $ Hess(F_2) \ge -\kappa$ for some $\kappa \ge 0$. The latter condition says, roughly,
that $F_2$ is not too badly non-convex. They show that even though the Bakry-Emery condition
(which requires $\kappa <0$) fails, nevertheless $\mu_2$ is hypercontractive if $e^{-2F_2} - e^{-2F_1}$ 
is small in a suitable $L^p$ sense. In this perturbation theorem a part of the hypothesis is placed directly
on the perturbed measure $\mu_2$. 
This paper contains a good, recent,  exposition of the use of logarithmic Sobolev inequalities 
in classical statistical mechanics and gives references to  
 related work in this large literature on spin systems. In this context it is natural to impose conditions
 directly on $F$, which is the Hamiltonian for a finite lattice spin system.

There are several other works imposing conditions directly on $F$,  which assume some differentiability  
of $F$ but are quite different from the conditions discussed in Remark \ref{remcombined} in that
the artificial potential $V_F$ is not used.
Typical of these are theorems imposing integrability with respect to $m$ of $\exp(|\n F|^2)$.
See Aida's Remark 4.13 in \cite{Aida2001}  for a comparison of these conditions.
Royer,  \cite[Theorem 3.2.7]{Roy07}, shows, by a simple proof, that if $|\n F|$ is bounded 
then $m^F$ is hypercontractive if $m$ itself is a measure  on $\R^n$ of the form $e^{-2F_1} d^n x$
for some uniformly convex function $F_1$.
}
\end{remark}

\subsection{Non-convexity}

We combined  the convexity techniques of the Bakry-Emery method and the perturbation theorem
of this paper to deduce  Eckmann's theorem over $\R$ in Section \ref{secEck}. 
But the final density of the ground state measure for a 
Schr\"odinger operator over $\R$ need not be log-concave in order 
for the ground state measure to be hypercontractive. Malrieu and Roberto \cite[Theorem 6.4.3]{Ane} 
have given an example of a density on the line which is far from log-concave but is hypercontractive.
Cattiaux, \cite[Example 5.5]{Cat2005}, further illuminated this example.   
Here we show how the hypercontractivity in their example can be deduced from Theorem \ref{thmM}.

\begin{example} \label{MRo-Cat}{\rm  (The example of Malrieu-Roberto, \cite[Theorem 6.4.3]{Ane}).
Let
\begin{align}
  F(x) &= x^2 + \beta x \sin x +C,  \ \ \ x, \beta \in \R,  \label{C5}
 \end{align}
 Define $\mu = e^{-2F} dx$. 
 It is clear that $\mu$ is a normalizable measure on $\R$. 
 We will ignore the normalization constant because it drops
 out in all of our calculations. We have $F' =2x +\beta x \cos x + \beta \sin x $ and 
 $F'' = -\beta x \sin x  +2 +2\beta \cos x$. Clearly $\liminf_{x \to \infty} F''(x) = -\infty$. 
 So $F$ is not convex  outside of any bounded set. 
 Malrieu-Roberto \cite[Theorem 6.4.3]{Ane} and  Cattiaux  \cite[Example 5.5]{Cat2005} have 
 shown that $\mu$ is hypercontractive if and only if  $|\beta| < 2$. We illustrate how our methods 
 show that  $|\beta| < 2$ is sufficient for hypercontractivity.
 
      According to the WKB equation, \eref{W7},  $e^{-F}$ is the ground state for the potential
 \begin{align}
 V_F &\equiv -F'' + (F')^2  \notag\\
 & = \beta x\sin x -\( 2 + 2\beta \cos x \) + \( x(2 + \beta \cos x) +  \beta \sin x\)^2 \notag \\
 & =x^2(2 + \beta \cos x)^2 +W        \label{C8}
\end{align}
where $W$ grows at most linearly: $|W(x)| \le c_1 |x| + c_2$. 

  Suppose that $|\beta| < 2$. 
Let $V_0 =x^2(2- |\beta|)^2$ and let $U = x^2(2 + \beta \cos x)^2 -x^2(2- |\beta)^2$.
Then $V_F = V_0 + U + W$. Moreover  $0\le U \le x^2 (2+|\beta|)^2$ .
Let $V_1 = U+W$. Then
\begin{align}
V_F = V_0 + V_1.        \label{C10}
\end{align}
The ground state measure for the quadratic potential $V_0$ is the Gaussian measure
$dm = (2\pi c)^{-1/2} e^{-x^2/(2c)} dx$ where $c^{-1} = 2 - |\beta|$. By the consecutive ground state transformation
procedure of Section \ref{secgst} we need only show that $V_1$ satisfies the two conditions \eref{mt2}
for some $\nu > 2c$ and some $\ka >0$. But since $U \ge 0$  and $W$ grows at most linearly we have
$\int_\R e^{-\nu V_1} dm < \infty$ for all $\nu >0$. Moreover $\int_\R e^{\ka V_1} dm < \infty$ 
whenever $\ka  (2+|\beta|)^2 < (2c)^{-1}$. Therefore we may apply Theorem \ref{thmM} to find that
$\mu$ is hypercontractive.
}
\end{example}

\begin{remark} {\rm (Inapplicability of Eckmann's method in this example). 
 Eckmann's method relies in part on non-oscillation of the given potential. In our example
the potential  whose ground state measure is $\mu$ is given by \eref{C8} and the 
derivative, $V_F'$, contains the highly oscillatory,
 quadratically growing term $-2x^2(2 + \beta \cos x) \sin x$ while $\sqrt{V_F}$ increases at most like $|x|$. 
 The condition \eref{E316} therefore fails.
}
\end{remark}

\begin{remark}  {\rm Since the example can be presented as the Gaussian measure 
$dm_0 \equiv Z^{-1}e^{-2x^2}dx$ with an additional density $e^{-2\beta x\sin x}$, 
as suggested by \eref{C5}, it is possible
to identify the measure in the example as the ground state measure for the Schrodinger operator
$\n^*\n + V_G$ in $L^2(m_0)$, where $\n^*$ is computed in $L^2(m_0)$ and 
$V_G = \n^*\n G + |\n G|^2$ with $G = \beta x \sin x$,  in accordance
with  \eref{W7} and Lemma \ref{lemcgst}.  A computation shows that this
approach will work only if $|\beta| <1$. The choice of decomposition of the potential $V_F$  given in \eref{C10},
for the application of the method of consecutive transformations, 
 therefore greatly affects the outcome.
}
\end{remark}

\subsection{A Toy Model}

Take $X= \R^n$. Define
\begin{align}
H = -\Delta +(Ax,x) +\l \sum_{j=1}^n x_j^4, \ \ \ \l >0     \label{tm5}
\end{align}
Here $(Ax,x)$ is any quadratic form, not necessarily positive.

  \begin{theorem} \label{thmtm1}
   The Hamiltonian \eref{tm5} is bounded below. The bottom of its
spectrum is an eigenvalue of multiplicity one. It has a unique strictly positive ground state $\psi$.
Let 
\begin{align}
m_\psi = \psi^2 d^nx  \ \ \ \text{on}\ \ \R^n.          \label{tm7}
\end{align}
Then there is a constant $c_1 < \infty$ such that
\begin{align}
Ent_{m_\psi}(u^2) \le 2c_1 \int_{\R^n} |\n u|^2 dm_\psi.      \label{tm9}
\end{align}
\end{theorem}
      \begin{proof}  By Example \ref{exE1} the Hamiltonian $-(d/dx)^2  + \l x^4$ has 
a unique positive ground  state $\psi_0$ and the measure 
$dm_{\psi_0} \equiv \psi_0^2 dx$ satisfies a LSI
      \begin{align}
      Ent_{m_{\psi_0}} (u^2) \le 2c \int_\R u'(x)^2 dm_{\psi_0}(x).     
 \end{align}    
 Let $m(dx_1,\dots, dx_n) =  m_{\psi_0}(dx_1) \cdots m_{\psi_0}(dx_n)$.
  By the additivity property of logarithmic Sobolev inequalities,  
 \cite{G1}, \cite{F75} or  \cite[Theorem 2.3]{G1993},
 the measure $m$ satisfies
 \begin{align}
 Ent_{m} (f^2) \le 2c \int_{\R^n} |\n f(x)|^2 dm(x).
 \end{align}
 We will  apply the perturbation theorem, Theorem \ref{thmM}  
 to $m$ with potential $V(x) = (Ax,x)$. 
     The consecutive ground state procedure of Section \ref{secgst} then shows that the ground state of the
     Hamiltonian $H$ is hypercontractive.  
 To verify the  hypotheses  \eref{mt2}   for $e^{\pm V}$
  we make the crude estimate $|(Ax,x)| \le b |x|^2$,
 where $b$ is the operator norm  of $A$ over $\R^n$.  For any real number $\alpha >0$ we have
\begin{align}
 \int_{\R^n} e^{\alpha |(Ax,x)| }dm &\le  \int_{\R^n} e^{\alpha b |x|^2}  dm \\
 &=  \(\int_{\R} e^{\alpha b x_1^2} dm_{\psi_0}(x_1)\)^n.
 \end{align}
 To show that this is finite it suffices, by Corollary \ref{correlgs}, to show that  \linebreak
 $\int_\R e^{\alpha b x_1^2} dm^F(x_1) < \infty$ for all positive $\alpha$, where $m^F$ is the intermediate 
 ground state measure in the construction of $m_{\psi_0}$. According to \eref{E371}, 
 $m^F$ has an even  density proportional to  $e^{-2\sqrt{\lambda}x^3/3}$ for large positive $x$. 
 Since  
 \newline $\int^{\infty} e^{\alpha b x^2 -2\sqrt{\lambda} x^3/3} dx < \infty$ for all real $\alpha$, the integral, 
 $\int_\R e^{\alpha b x_1^2} dm^F(x_1) < \infty$ and therefore $\int_{\R^n}e^{\alpha |(Ax,x)| }dm <\infty$
 for all real $\alpha$.
 We can therefore apply Theorem \ref{thmM} since \eref{mt2} holds for all real $\nu$ and $\kappa$.
 \end{proof}
         
\bigskip
     
For the significance of this example to $\phi^4$ models note that  the Hamiltonians  \eref{tm5} include
Hamiltonians of the form $H = -\Delta +(Bx,x) +\l \sum_{j=1}^n (x_j^4- a x_j^2), \ \ \ \l >0, a >0 $.

\bigskip
The author states that there is no conflict of interest.

\section{Bibliography}

\bibliographystyle{amsplain}
\bibliography{ymh1}

\end{document}